\documentclass{article}

\usepackage{fullpage}

\usepackage{amssymb}                                    % \rightrightarrows
\usepackage{mathrsfs}                    % \mathscr
\usepackage{amsmath}                    % Klammern, texte in Formeln
\usepackage{amsfonts}                   % Andere Schriftarten
\usepackage{amsthm}                     % Abschnitte, Theoreme
\usepackage{makeidx}
\usepackage{graphics}
\usepackage{hyperref}
\usepackage{rotating}
\usepackage{fancyhdr}

\usepackage{color}

\usepackage[all]{xy}

\theoremstyle{plain}
\newtheorem{theorem}{Theorem}[section]

\newtheorem{proposition}[theorem]{Proposition}

\theoremstyle{definition}
\newtheorem{definition}[theorem]{Definition}
\newtheorem{example}[theorem]{Example}

\newtheorem{remark}[theorem]{Remark}

\newtheorem{claim}[theorem]{Claim}

\renewcommand{\(}{\begin{equation}}
\renewcommand{\)}{\end{equation}}
\newcommand{\bea}{\begin{eqnarray}}
\newcommand{\eea}{\end{eqnarray}}

\begin{document}

\title{Higher prequantum geometry\footnote{This file is available at \href{http://ncatlab.org/schreiber/show/Higher+Prequantum+Geometry}{ncatlab.org/schreiber/show/Higher+Prequantum+Geometry}}}

\author{Urs Schreiber \thanks{CAS Prague and MPI Bonn}}

\maketitle

\thispagestyle{fancy}

This is a survey of motivations, constructions and applications of \emph{higher prequantum geometry}
\cite{dcct}\footnote{Section \ref{PrequantumLocalFieldTheoryInMotivation} is based on joint work with Igor Khavkine.
Section \ref{examplesinmotivation} is based on joint work with Domenico Fiorenza and Hisham Sati.
Section \ref{abstractprequantumgeometry} is based on joint work with David Carchedi}.
In section \ref{PrequantumLocalFieldTheoryInMotivation}
we highlight the open problem of prequantizing local field theory in a local and gauge invariant way,
and we survey how a solution to this problem exists in higher differential geometry.
In section \ref{examplesinmotivation}  we survey examples and problems of interest.
In section \ref{abstractprequantumgeometry}  we survey the abstract cohesive homotopy theory that serves to make all this precise and tractable.
Combining this cohesive with linear homotopy theory should serve to
non-perturbatively quantize higher prequantum geometry, see \cite{motquant}.

\tableofcontents

%%%%%%%%%%%%%%%%%%%%%%%%%%%%%%%%%%%%%%%%%%%%
\section{Prequantum field theory}
\label{PrequantumLocalFieldTheoryInMotivation}
%%%%%%%%%%%%%%%%%%%%%%%%%%%%%%%%%%%%%%%%

The geometry that underlies the physics of Hamilton and Lagrange's classical mechanics and classical field theory
has long been identified: this is
\emph{symplectic geometry} \cite{Arnold89} and \emph{variational calculus on jet bundles} \cite{Anderson-book,Olver93}.
In these theories, configuration spaces of physical systems are differentiable manifolds, possibly infinite-dimensional,
and the physical dynamics is all encoded by way of certain globally defined differential forms on these spaces.

But fundamental physics is of course of quantum nature, to which classical physics is but an approximation
that applies at non-microscopic scales. Of what mathematical nature are systems of quantum physics?

%%%%%%%%%%%%%%%%%%%%%%%%%%%%%%%%%%%%%%%%%%%%%%%%%%
\subsection{The need for prequantum geometry}
\label{Theneedforprequantum}
\label{TheNeedForPrequantumGeometry}
%%%%%%%%%%%%%%%%%%%%%%%%%%%%%%%%%%%%%%%%%%%%%%%%%%

A sensible answer to this question is given by algebraic deformation
theory. One considers a deformation of classical physics to quantum
physics by deforming a Poisson bracket to the commutator in a
non-commutative algebra, or by deforming a classical measure to a
quantum BV operator.
$$
  
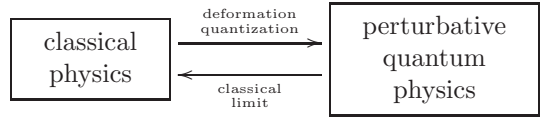
\xymatrix{
    \fbox{
    \begin{tabular}{c}
      classical
      \\
      physics
    \end{tabular}
    }
    \ar@<+6pt>[rr]^-{\mathrm{deformation} \atop \mathrm{quantization}}
    \ar@{<-}@<-6pt>[rr]_-{\mathrm{classical} \atop \mathrm{limit}}
    &&
    \fbox{
      \begin{tabular}{c}
        perturbative
        \\
        quantum
        \\
        physics
      \end{tabular}
    }
  }
$$
However,
this tends to work only perturbatively, in the
infinitesimal neighbourhood of classical physics, expressed in terms of
formal (possibly non-converging) power series in Planck's constant
$\hbar$.

There is a genuinely non-perturbative mathematical formalization of quantization,
called \emph{geometric quantization} \cite{Souriau70, Souriau74, Kostant75, BatesWeinstein}.
A key insight is that before genuine quantization even applies,
there is to be a \emph{pre-quantization} step in which
the classical geometry is supplemented by \emph{global} coherence data.
For actions of global gauge groups, this is also known as
cancellation of \emph{classical anomalies} \cite[5.A]{Arnold89}.
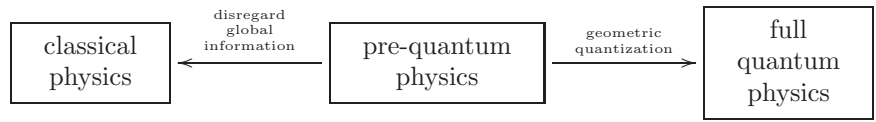
$$
  \xymatrix{
    \fbox{
      \begin{tabular}{c} classical \\ physics \end{tabular}
    }
    \ar@{<-}[rr]^-{\mbox{\tiny \begin{tabular}{c}disregard \\ global \\ information \end{tabular}}}
    &&
    \fbox{
      \begin{tabular}{c}
        pre-quantum \\ physics
      \end{tabular}
    }
    \ar[rr]^-{\mathrm{geometric} \atop \mathrm{quantization}}
    &&
    \fbox{
      \begin{tabular}{c} full \\ quantum \\ physics \end{tabular}
    }
    \\
  }
$$

The archetypical example of pre-quantization is \emph{Dirac charge quantization} \cite{Dirac31},
\cite[5.5]{Frankel}, \cite{Freed00}.
The classical mechanics of an electron propagating in an electromagnetic field on a spacetime $X$ is all
encoded in a differential 2-form on $X$, called the Faraday tensor $F$, which encodes the classical
Lorentz force that the electromagnetic field exerts on the electron.
But this data is insufficient for passing to the quantum theory of the electron: locally, on a coordinate chart $U$,
what the quantum electron really couples to is the ``\emph{vector potential}'', a differential 1-form
$A_U$ on $U$, such that $d A_U = F|_U$. But globally such a vector potential may not exist.
Dirac realized\footnote{Dirac considered this in the special case where spacetime is the complement
in 4-dimensional Minkowski spacetime of the worldline
of a magnetic point charge. The homotopy type of this space is the 2-sphere and hence in this case
principal connections may be exhibited by what in algebraic topology is called a
\emph{clutching construction}, and this is what Dirac described. What the physics literature knows as the
``Dirac string'' in this context is the ray whose complement gives one of the two hemispheres in the clutching
construction.}
that what it takes to define the quantized electron globally is,
in modern language, a lift of the locally defined vector potentials to an $(\mathbb{R}/\mathbb{Z})$-principal connection
on a $(\mathbb{R}/\mathbb{Z})$-principal bundle over spacetime. The first Chern class of this principal bundle is
quantized, and this is identified with the quantization of the magnetic charge whose induced force the electron feels. This
quantization effect, which needs to be present before the quantization of the dynamics of the electron itself
even makes sense globally, is an example of \emph{pre-quantization}.

A variant of this example occupies particle physics these days. As we pass attention from electrons to quarks, these
couple to the weak and strong nuclear force, and this coupling is, similarly, locally described by a 1-form $A_U$,
but now with values in a Lie algebra $\mathfrak{su}(n)$,
from which the strength of the nuclear force field is encoded by the 2-form $F|_U := d A_U + \tfrac{1}{2}[A_U \wedge A_U]$.
For the consistency of the quantization of quarks, notably for the consistent global definition of
\emph{Wilson loop observables}, this local data must be lifted to an $\mathrm{SU}(n)$-principal connection
on a $\mathrm{SU}(n)$-principal bundle over spacetime. The second Chern class of this bundle is quantized,
and is physically interpreted as the number of \emph{instantons}\footnote{Strictly speaking, the term ``instanton'' refers
to a principal connection that in addition to having non-trivial topological charge also minimizes Euclidean energy.
Here we are just concerned with the nontrivial topological charge, which in particular is insensitive to
and independent of any ``Wick rotation''.}.
In the physics literature instantons are expressed via Chern-Simons 3-forms, mathematically these
constitute the pre-quantization of the 4-form $\mathrm{tr}(F \wedge F)$ to a 2-gerbe with 2-connection, more on this in a moment.

The vacuum which we inhabit is filled with such instantons at a density of the order of one instanton per femtometer in every direction.
(The precise quantitative theoretical predictions of this \cite{SchaeferShuryak96} suffer from an infrared regularization ambiguity, but
numerical simulations demonstrate the phenomenon \cite{Gr13}.)
This ``instanton sea'' that fills spacetime governs the mass of the $\eta'$-particle \cite{Witten79, Veneziano79}
as well as other non-perturbative chromodynamical phenomena, such as the quark-gluon plasma seen in experiment \cite{Shuryak01}.
It is also at the heart of the standard hypothesis for the
mechanism of primordial baryogenesis \cite{Sakharov67, tHooft76, RiottoTrodden}, the fundamental explanation of
a universe filled with matter.

Passing beyond experimentally observed physics, one finds that the qualitative structure of the
standard model of particle physics coupled to gravity, namely the structure of Einstein-Maxwell-Yang-Mills-Dirac-Higgs theory,
follows naturally if one assumes that the 1-dimensional worldline theories of particles such as electrons and quarks are,
at very high energy,
accompanied by higher dimensional worldvolume theories of fundamental objects called strings, membranes and
generally $p$-branes (e.g. \cite{Duff99}).
While these are hypothetical as far as experimental physics goes, they are interesting examples of
the mathematical formulation of field theory,
and hence their study is part of mathematical physics, just as the study of the Ising model or $\phi^4$-theory.
These $p$-branes are subject to a higher analog of the Lorentz force, and this is subject to
a higher analog of the Dirac charge quantization condition, again a prequantum effect for the worldvolume theory.

For instance the \emph{strong CP-problem} of the standard model of particle physics has several hypothetical solutions,
one is the presence of particles called \emph{axions}. The discrete shift symmetry (Peccei-Quinn symmetry) that characterizes these
may naturally be explained as the result of $\mathbb{R}/\mathbb{Z}$-brane charge quantization in the hypothetical case that axions are
wrapped membranes \cite[section 6]{Wittenaxions}.

More generally, $p$-brane charges are not quantized in ordinary integral cohomology, but in generalized cohomology theories.
For instance 1-branes (strings) are well-known to carry charges whose quantization is in K-theory (see \cite{Freed00}).
While the physical existence of fundamental strings remains hypothetical, since the boundaries of strings are particles
this does impact on known physics, for instance on the quantization of phase spaces that are not symplectic but just Poisson \cite{Nuiten}.

Finally, when we pass from fundamental physics to low energy effective physics such as solid state physics, then
prequantum effects control topological phases of matter: symmetry protected topological phases are being argued 
\cite{cglw11} to be described, at low energy, by higher dimensional WZW models, 
mathematically of just the same kind as those hypothetical fundamental super $p$-brane models.

\begin{center}
\begin{tabular}{|c|c|}
  \hline
  \begin{tabular}{c}
    {\bf worldvolume}
    \\
    {\bf field theory}
  \end{tabular}
  &
  {\bf prequantum effect}
  \\
  \hline
  \hline
   electron & \begin{tabular}{c} Dirac charge quantization, \\ magnetic flux quantization \end{tabular}
   \\
   \hline
   quark & \begin{tabular}{c} instantons, \\ baryogenesis
           \end{tabular}
   \\
   \hline
   $p$-brane & \begin{tabular}{c} brane charge quantization, \\ axion shift symmetry \end{tabular}
   \\
   \hline
\end{tabular}
\end{center}

\medskip

These examples show that pre-quantum geometry is at the heart of the description of fundamental and of effective physical reality.
Therefore, before rushing to discuss the mathematics of quantum geometry proper, it behooves us to first carefully
consider the mathematics of pre-quantum geometry. This is what we do here.

\medskip

If the prequantization of the Lorentz force potential 1-form $A$
for the electron is a connection on a $(\mathbb{R}/\mathbb{Z})$-principal
bundle, then what is the prequantization of the Chern-Simons 3-form counting instantons, or of the
higher Lorentz force potential $(p+1)$-form of a $p$-brane for higher $p$?

This question has no answer in traditional differential geometry.
It is customary to consider it only after transgressing the $(p+1)$-forms down to 1-forms
by splitting spacetime/worldvolume as a product $\Sigma = \Sigma_p \times [0,1]$
of $p$-dimensional spatial slices with a time axis, and integrating the $(p+1)$-forms
over $\Sigma_p$
\begin{center}
\begin{tabular}{|ccc|}
  \hline
  {\bf global in space} && {\bf local in spacetime}
  \\
  \hline
  1-form $A_1$ && $(p+1)$-form $A_{p+1}$
  \\
  $
    \underset{[0,1]}{\int} A_1
  $
  &
  $
    \xymatrix{
      \ar@{=}[rr]^{A_{1} := \int_{\Sigma_p} A_{p+1}}_{\mathrm{fiber}\;\mathrm{integration}}
      &&
    }
  $
  &
  $
    \underset{\Sigma_p \times [0,1]}{\int} A_{p+1}
  $
  \\
  \hline
\end{tabular}
\end{center}
This transgression reduces $(p+1)$-dimensional field theory to 1-dimensional field theory, hence to mechanics,
on the moduli space of spatial field configurations. That 1-dimensional field theory may
be subjected to the traditional theory of prequantum mechanics.

But clearly this space/time decomposition is a brutal step for relativistic field theories.
It destroys their inherent symmetry and makes their analysis hard.
In physics this is called the ``non-covariant'' description of field theory, referring
to covariance under application of diffeomorphisms.
We need prequantum geometry for spacetime local field theory where $(p+1)$-forms may be
prequantized by regarding them as connections on higher degree analogs of principal bundles \cite{hgp,KhavkineSchreiber}.
This is what we discuss below.

Where an ordinary principal bundle is a smooth manifold, hence a smooth set, with certain
extra structure, a higher principal bundle needs to be a smooth homotopy type.
\begin{center}
\begin{tabular}{|c|c|}
  \hline
  \multicolumn{2}{|c|}{\bf prequantum bundle}
  \\
  {\bf global in space} & {\bf local in spacetime}
  \\
  \hline
  smooth set & smooth homotopy type
  \\
  \hline
\end{tabular}
\end{center}

The generalization of geometry to higher geometry, where sets -- which may be thought of as homotopy 0-types --
are generalized to homotopy $p$-types
for higher $p$, had been envisioned in \cite{PursuingStacks} and a precise general framework has eventually been obtained in
\cite{Lurie}. This may be specialized to higher differential geometry \cite{dcct}, which is the context we will be using here.
The description of pre-quantum field theory local in spacetime is related to the description of
topological quantum field theory local-to-the-point known as ``extended'' or ``multi-tiered'' field theory
\cite{LurieTFT}\cite{Bergner}.
\begin{center}
\begin{tabular}{|c||c|c|}
  \hline
  & {\bf classical} & {\bf prequantum}
  \\
  \hline
  \hline
  {\bf global in space} & \begin{tabular}{c} symplectic geometry  \\ \hline classical mechanics\end{tabular} &
  \begin{tabular}{c} prequantum geometry  \\ \hline prequantum mechanics\end{tabular}
  \\
  \hline
  {\bf local in spacetime} &
  \begin{tabular}{c} diffiety geometry \\ \hline classical field theory \end{tabular} &
  \begin{tabular}{c} higher prequantum geometry \\ \hline prequantum field theory \end{tabular}
  \\
  \hline
\end{tabular}
\end{center}

Once we are in a context of higher geometry where higher prequantum bundles exist, several other
subtleties fall into place.

\begin{center}
\begin{tabular}{|c|c|}
  \hline
  \begin{tabular}{c}
    {\bf ingredient of }
    \\
    {\bf variational calculus}
  \end{tabular}
  &
  \begin{tabular}{c}
    {\bf new examples}
    \\
    {\bf available in }
    \\
    {\bf higher geometry}
  \end{tabular}
  \\
  \hline
  spacetime & orbifolds
  \\
  \hline
  field bundle & \begin{tabular}{c} instanton sectors of gauge fields,\\ integrated BRST complex\end{tabular}
  \\
  \hline
  prequantum bundle & \begin{tabular}{c} global Lagrangians for WZW-type models \end{tabular}
  \\
  \hline
\end{tabular}
\end{center}

A well-kept secret of the traditional formulation of variational calculus on jet bundles is that it does
not in fact allow to properly formulate global aspects of local gauge theory. Namely the only way to
make the fields of gauge theory be sections of a traditional field bundle is to fix the instanton number (Chern class)
of the gauge field configuration. The gauge fields then are taken to be connections on that fixed bundle.
One may easily see \cite{higherfieldbundles} that it is impossible to have a description of gauge fields as sections of a field bundle
that is both local and respects the gauge principle. However, this \emph{is} possible with a higher field bundle.
Indeed, the natural choice of the field bundle for gauge fields has as typical fiber the smooth moduli stack of
principal connections. Formulated this way, not only does the space of all field configurations then span all
instanton sectors, but it also has the gauge transformations between gauge field configurations built into it.
In fact it is then the globalized (integrated) version of what in the physics literature is known as the
(off-shell) BRST complex of gauge theory.

Moreover, in a context of higher geometry also spacetime itself is allowed to be a smooth homotopy type. This is relevant
at least in some hypothetical models of fundamental physics, which require spacetime to be an \emph{orbifold}. Mathematically,
an orbifold is a special kind of Lie groupoid, which in turn is a special kind of smooth homotopy 1-type.

%%%%%%%%%%%%%%%%%%%%%%%%%%%%%%%%%%%%%%%%%%%%%%%%%%%%%%%%%%%%%%%%
\subsection{The principle of extremal action -- comonadically}
\label{Principleofextremalactioncomonadically}
%%%%%%%%%%%%%%%%%%%%%%%%%%%%%%%%%%%%%%%%%%%%%%%%%%%%%%%%%%%%%%%%

Most field theories of relevance in theory and in nature are \emph{local Lagrangian field theories}
(and those that are not tend to be holographic boundary theories of those that are).
This means that their equations of motion are partial differential equations obtained as
Euler-Lagrange equations of a local variational principle. This is the modern incarnation of the time-honoured
\emph{principle of least action} (really: of extremal action).

We review how this is formalized, from a category-theoretic point of view that will point the way to
prequantum covariant field theory below in section \ref{elgerbesinintroduction}, \cite{KhavkineSchreiber}.

\medskip

The kinematics of a field theory is specified by
a smooth manifold $\Sigma$ of dimension $(p+1)$ and a smooth bundle $E$ over $\Sigma$. A \emph{field configuration}
is a smooth section of $E$. If we think of $\Sigma$ as being spacetime, then typical examples of fields
are the electromagnetic field or the field of gravity. But we may also think of $\Sigma$ as being the
worldvolume of a particle (such as the electron in the above examples)
or of a higher dimensional ``brane'' that propagates in a fixed background of such spacetime fields,
in which case the fields are the maps that encode a given trajectory.

The dynamics of a field theory is specified by an \emph{equation of motion}, a partial differential equation
for such sections. Since differential equations are equations among all the derivatives of such sections, we consider the spaces that these form:
the \emph{jet bundle} $J^\infty_\Sigma E $ is the bundle over $\Sigma$ whose fiber over a point $\sigma \in \Sigma$ is the space
of sections of $E$ over the infinitesimal neighbourhood $\mathbb{D}_\sigma$ of that point:
  $$
    \left\{
    \raisebox{20pt}{
    \xymatrix{
      & J^\infty_\Sigma E  \ar[d]
      \\
      \ast \ar@{-->}[ur] \ar[r]^-{\sigma} &  \Sigma
    }
    }
    \right\}
    \;\;\;\simeq\;\;\;
    \left\{
    \raisebox{20pt}{
    \xymatrix{
      & E \ar[d]
      \\
      \mathbb{D}_\sigma \ar@{-->}[ur] \ar@{^{(}->}[r] & \Sigma
    }
    }
    \right\}
  $$
Therefore every section $\phi$ of $E$
yields a section $j^\infty(\phi)$ of the jet bundle, given by $\phi$
and all its higher order derivatives.
$$
  
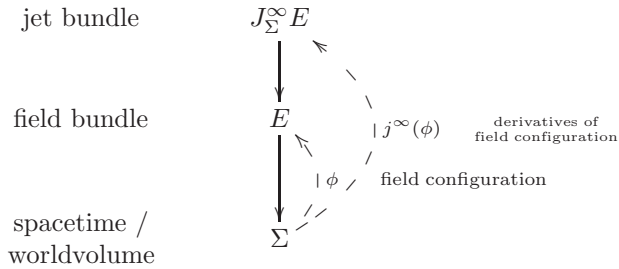
\xymatrix{
     \mbox{jet bundle} & J^\infty_\Sigma E
     \ar[d]
     \\
     \mbox{field bundle} & E
     \ar[d]
     \\
     \mbox{\begin{tabular}{c} spacetime /\\ worldvolume\end{tabular}} & \Sigma \ar@{-->}@/_1.2pc/[u]_{\phi\;\;\;\;\;\;\;\mathrm{field}\;\mathrm{configuration}}
     \ar@{-->}@/_3pc/[uu]_{j^\infty(\phi)\;\;\;\;\;{{\mathrm{derivatives}\; \mathrm{of}} \atop {\mathrm{field}\;\mathrm{configuration}}}}
  }
$$
Accordingly, for $E, F$ any two smooth bundles over $\Sigma$, then a bundle map
$$
  \xymatrix{
    J^\infty E \ar[dr] \ar[rr]^f &&  F \ar[dl]
    \\
    & \Sigma
  }
$$
encodes a (non-linear) differential operator $D_f : \Gamma_\Sigma(E) \longrightarrow \Gamma_\Sigma(F)$ by
sending any section $\phi$ of $E$ to the section  $f \circ j^\infty(\phi)$ of $F$.
Under this identification, the composition of differential operators $D_g \circ D_f$ corresponds to the
\emph{Kleisli-composite} of $f$ and $g$, which is
$$
  \xymatrix{
     J^\infty E \ar[r] \ar[drr]
      &  J^\infty J^\infty E \ar[dr] \ar[r]^-{J^\infty f}
      & J^\infty F \ar[r]^-{g} \ar[d]
      & G \ar[dl]
    \\
    && \Sigma
  }
  \,.
$$
Here the first map is given by re-shuffling derivatives and gives the jet bundle construction $J^\infty_\Sigma$ the structure of a comonad
-- the \emph{jet comonad}.

Differential operators are so ubiquitous in the present context that it is convenient
to leave them notationally implicit and understand \emph{every} morphism of bundles $E \longrightarrow F$ to designate
a differential operator $D : \Gamma_\Sigma(E)\longrightarrow \Gamma_\Sigma(F)$. This is what we will do from now on.
Mathematically this means that we are now in the co-Kleisli category $\mathrm{Kl}(J^\infty_\Sigma)$ of the jet comonad
$$
  \mathrm{DiffOp}_\Sigma \simeq \mathrm{Kl}(J^\infty_\Sigma)
  \,.
$$

For example the de Rham differential is a differential operator from sections of $\wedge^{p}T^\ast \Sigma$
to sections of $\wedge^{p+1}T^\ast \Sigma$ and hence now appears as a morphism of the form
$d_H : \wedge^{p}T^\ast \Sigma \longrightarrow \wedge^{p+1}T^\ast \Sigma$.
With this notation, a \emph{globally defined local Lagrangian} for fields that are sections of some bundle $E$
over spacetime/worldvolume $\Sigma$ is simply a morphism of the form
$$
  L : E \longrightarrow \wedge^{p+1}T^\ast \Sigma
  \,.
$$
Unwinding what this means, this is a function that at each point of $\Sigma$ sends the value of field configurations and all their spacetime/worldvolume
derivatives at that point to a $(p+1)$-form on $\Sigma$ at that point. It is this pointwise local (in fact: infinitesimally local) dependence
that the term \emph{local} in \emph{local Lagrangian} refers to.

Notice that this means that $\wedge^{p+1} T^\ast \Sigma$
serves the role of the \emph{moduli space} of horizontal $(p+1)$-forms:
$$
  \Omega^{p+1}_H(E) = \mathrm{Hom}_{\mathrm{DiffOp}_\Sigma}(E,\wedge^{p+1}T^\ast \Sigma)
  \,.
$$

Regarding such $L$ for a moment as just a differential form on $J^\infty_\Sigma(E)$, we may apply the
de Rham differential to it. One finds that this uniquely decomposes as a sum of the form
\begin{equation}
  \label{differentialofLagrangian}
  d L = \mathrm{EL} - d_H (\Theta + d_H(\cdots))
  \,,
\end{equation}
for some $\Theta$ and for $\mathrm{EL}$ pointwise the
pullback of a vertical 1-form on $E$; such a differential form is called a \emph{source form}:
$\mathrm{EL} \in \Omega^{p+1,1}_S(E)$.
This particular source form is of paramount importance: the equation
$$
  \underset{v\in \Gamma(V E)}{\forall} j^\infty(\phi)^\ast \iota_v\mathrm{EL} = 0
$$
on sections $\phi \in \Gamma_\Sigma(E)$
is a partial differential equation, and
this is called the \emph{Euler-Lagrange equation of motion} induced by $L$. Differential equations arising this way
from a local Lagrangian are called \emph{variational}.

A little reflection reveals that this is indeed a re-statement of the traditional prescription of obtaining the Euler-Lagrange equations by
locally varying the integral over the Lagrangian and then applying partial integration to turn all variation of derivatives (i.e. of jets) of fields
into variation of the fields themselves. Here we do not consider this under the integral, and hence the boundary terms arising from the
would-be partial integration show up as the contribution $\Theta$.

We step back to say this more neatly. In general, a differential
equation on sections of a bundle $E$ is what characterizes the kernel of a differential operator. Now such
kernels do not in general exist in the Kleisli category $\mathrm{DiffOp}_\Sigma$
of the jet comonad that we have been using, but (as long as it is non-singular) it does exist
in the full Eilenberg-Moore category $\mathrm{EM}(J^\infty_\Sigma)$ of jet-coalgebras. In fact,
that category turns out \cite{Marvan86} to be equivalent to the category $\mathrm{PDE}_\Sigma$ whose objects are differential equations on sections of bundles,
and whose morphisms are solution-preserving differential operators :
$$
  \mathrm{PDE}_\Sigma \simeq \mathrm{EM}(J^\infty_\Sigma)
  \,.
$$
Our original category of bundles with differential operators between them sits in $\mathrm{PDE}_\Sigma$
as the full subcategory on the trivial differential equations, those for which every section is a solution.
This inclusion extends to (pre-)sheaves via left Kan extension; so we are now in the sheaf topos (e.g. \cite{Johnstone02})
$\mathrm{Sh}(\mathrm{PDE}_\Sigma)$.

And while source forms such as the Euler-Lagrange form $\mathrm{EL}$ are not representable in $\mathrm{DiffOp}_\Sigma$,
it is still true that for
$f : E\longrightarrow F$ any differential operator then the property of source forms
is preserved by precomposition with this map,
hence we have the induced pullback operation on source forms: $f^\ast : \Omega^{p+1,1}_S(F) \longrightarrow \Omega^{p+1,1}_S(E)$.
This means that source forms do constitute a presheaf on $\mathrm{DiffOp}_\Sigma$,
hence by left Kan extension an object in the topos over partial differential equations:
$$
 \mathbf{\Omega}^{p+1,1}_S
 \in
 \mathrm{Sh}(\mathrm{PDE}_\Sigma)
 \,.
$$

Therefore now the Yoneda lemma applies to say that $\mathbf{\Omega}^{p+1,1}_S$ is the moduli space for source forms in this context: a source form on $E$ is now just a morphism of the form $E \longrightarrow \mathbf{\Omega}^{p+1,1}_S$. Similarly, the Euler variational derivative is now incarnated as a morphism of moduli spaces of the form $\mathbf{\Omega}^{p+1}_H \stackrel{\delta_V}{\longrightarrow} \mathbf{\Omega}^{p+1,1}_S$, and applying the variational differential to a Lagrangian is now incarnated as the composition of the corresponding two modulating morphisms
$$
  \mathrm{EL} := \delta_V L : E \stackrel{L}{\longrightarrow} \mathbf{\Omega}^{p+1}_H \stackrel{\delta_V}{\longrightarrow} \mathbf{\Omega}^{p+1,1}_S
$$.
Finally, and that is the beauty of it, the Euler-Lagrange differential equation $\mathcal{E}$ induced by the Lagrangian $L$ is now incarnated simply as the kernel of $\mathrm{EL}$:\footnote{\label{BVremark}That kernel always exists in the topos $\mathrm{Sh}(\mathrm{PDE}_\Sigma)$, but it may not be representable
by an actual sub\emph{manifold} of $J^\infty_\Sigma E$ if there are singularities. Without any changes to the general discussion
one may replace the underlying category of manifolds by one of ``derived manifolds'' formally dual to ``BV-complexes'',
where algebras of smooth functions are replaced by higher homotopy-theoretic algebras, for instance by graded algebras equipped with a differential $d_{\mathrm{BV}}$.}
$\mathcal{E} \stackrel{\mathrm{ker}(\mathrm{EL})}{\hookrightarrow} E$.

In summary, from the perspective of the topos over partial differential equations,
the traditional structure of local Lagrangian variational field theory
is captured by the following diagram:
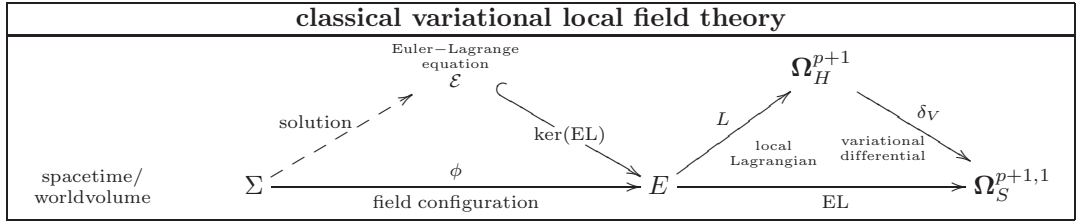
\begin{center}
\begin{tabular}{|c|}
\hline
{\bf classical variational local field theory}
\\
\hline
$
  \xymatrix@C=40pt{
     & &
     {{{\mathrm{Euler-Lagrange} \atop \mathrm{equation}}} \atop {\mathcal{E}}}
     \ar@{^{(}->}[dr]|-{\mathrm{ker}(\mathrm{EL})}
     &&\mathbf{\Omega}^{p+1}_H
      \ar[dr]^{\delta_V}_{\mathrm{variational}\atop \mathrm{differential}}
     \\
     {\mathrm{spacetime/}
     \atop
     \mathrm{worldvolume}}
     \!\!\!\!\!\!\!
     &
     \Sigma
     \ar@{-->}[ur]^-{\mathrm{solution}}
     \ar[rr]^-\phi_-{\mathrm{field}\;\mathrm{configuration}}
     &&
     E
     \ar[ur]^-{L}_-{\mathrm{local} \atop \mathrm{Lagrangian}} \ar[rr]_{\mathrm{EL}}
     &&
     \mathbf{\Omega}^{p+1,1}_S
  }
$
\\
\hline
\end{tabular}
\end{center}

So far, all this assumes that there is a globally defined Lagrangian form $L$ in the first place,
which is not in fact the case for all field theories of interest. Notably it is in general not the case
for field theories of higher WZW type. However, as the
above diagram makes manifest, for the purpose of identifying the classical equations of motion,
it is only the variational Euler differential  $\mathrm{EL} := \delta_V L$ that matters.
But if that is so, the variation being a local operation, then we should still call equations
of motion $\mathcal{E}$ \emph{locally variational} if there is a cover $\{U_i \to E\}$ and
Lagrangians on each patch of the cover $L : U_i \to \mathbf{\Omega}^{p+1}_H$, such that there
is a globally defined Euler-Lagrange form $\mathrm{EL}$ which restricts on each patch $U_i$ to
the variational Euler-derivative of $L_i$.
Such \emph{locally variational} classical field theory is discussed in  \cite{AndersonDuchamp, FPW-locvar}.
\begin{center}
\begin{tabular}{|c|}
\hline
{\bf classical locally variational local field theory}
\\
\hline
$
  \xymatrix@C=40pt{
    &
    &
    {{{\mathrm{Euler-Lagrange} \atop \mathrm{equation}}} \atop {\mathcal{E}}}
    \ar[dr]|{\mathrm{ker}(\mathrm{EL})}
    &
    \coprod_{i \in I} U_i
    \ar@{->>}[d]
    \ar[r]^-{ {(L_i)_{i \in I}}}_{\,\,\,{\mathrm{locally}\,\mathrm{defined}} \atop {\mathrm{local}\,\mathrm{Lagrangians}}}
    &
    \mathbf{\Omega}^{p+1}_H
    \ar[dr]^{\delta_V}_{\mathrm{variational} \atop \mathrm{differential}}
    \\
     {\mathrm{spacetime/}
     \atop
     \mathrm{worldvolume}}
     \!\!\!\!\!\!\!
     &
    \Sigma
    \ar@{-->}[ur]^{\mathrm{solution}}
    \ar[rr]^-{\phi}_{\mathrm{field}\,\mathrm{configuration}}
    &&
    E
    \ar[rr]_{\mathrm{EL}}
    &&
    \mathbf{\Omega}^{p+1,1}_S
  }
  \,.
$
\\
\hline
\end{tabular}
\end{center}
But when going beyond classical field theory, the Euler-Lagrange equations of motion $\mathcal{E}$ are not the
end of the story. As one passes to the quantization of a classical field theory, there are further global
structures on $E$ and on $\mathcal{E}$ that are relevant. These are the action functional and the
Kostant-Souriau prequantization of the covariant phase space. For these one needs to promote a patchwise system of local Lagrangians to a
$p$-gerbe connection. This we turn to now.

%%%%%%%%%%%%%%%%%%%%%%%%%%%%%%%%%%%%%%%%%%%%%%%%%%
\subsection{The global action functional -- cohomologically}
\label{globalactionfunctionalinintroduction}
%%%%%%%%%%%%%%%%%%%%%%%%%%%%%%%%%%%%%%%%%%%%%%%%%

For a globally defined Lagrangian $(p+1)$-form $L_{p+1}$ on the jet bundle of a given field bundle, then the value of the action functional
on a compactly supported field configuration $\phi$ is simply the integral
$$
  S(\phi) := \int_\Sigma j^\infty(\phi)^\ast L_{p+1}
$$
of the Lagrangian, evaluated on the field configuration, over the spacetime/worldvolume $\Sigma$.

But when Lagrangian forms are only defined patchwise on a cover $\{U_i \to E\}_i$ as in the locally
variational field theories mentioned above in \ref{Principleofextremalactioncomonadically},
then there is no way to globally make invariant sense of the action functional!
As soon as sections pass through several patches, then making invariant sense of such an integral requires more
data, in particular it requires more than just a compatibility condition of the locally
defined Lagrangian forms on double intersections.

The problem of what exactly it takes to define global integrals of locally defined forms
has long found a precise answer in mathematics, in the theory of ordinary differential cohomology.
This has several equivalent incarnations, the one closest to classical constructions in
differential geometry involves Cech cocycles:
one first needs to choose on each intersection $U_{i j}$ of two patches $U_i$ and $U_j$ a differential form $(\kappa_{p})_{i j}$ of degree
$p$, whose horizontal de Rham differential is the difference between the two Lagrangians restricted to that intersection
$$
  (L_{p+1})_j - (L_{p+1})_i = d_H (\kappa_p)_{i j} \;\;\;\;\mbox{on $U_{i j}$}
  \,.
$$
Then further
one needs to choose on each triple intersection $U_{i j k}$ a horizontal differential form $(\kappa_{p-1})_{i j k}$ of degree $p-1$
whose horizontal differential is the alternating sum of the relevant three previously defined forms:
$$
  (\kappa_p)_{j k} - (\kappa_p)_{i k} + (\kappa_p)_{i j} = d_H (\kappa_{p-1})_{i j k }
  \;\;\;\;\;
  \mbox{on $U_{i j k}$}
  \,.
$$
And so on. Finally on $(p+2)$-fold intersections one needs to choose smooth functions $(\kappa_0)_{i_0 \cdots i_{p+1}}$ whose
horizontal differential is the alternating sum of $(p+2)$ of the previously chosen horizontal 1-forms, and, moreover,
on $(p+3)$-fold intersections the alternating sum of these functions has to vanish.
Such a tuple $(\{U_i\}; \{(L_{p+1})_i\},\, \{(\kappa_p)_{i j}\}, \cdots)$ is a horizontal \emph{Cech-de Rham cocycle}
in degree $(p+2)$.

Given such, there is then a way to make sense of global integrals: one chooses a triangulation
subordinate to the given cover, then integrates the locally defined Lagrangians $(L_{p+1})_i$ over the
$(p+1)$-dimensional cells of the
triangulation, integrates the gluing forms $(\kappa_p)_{i j}$ over the $p$-dimensional faces of these cells, the
higher gluing forms $(\kappa_p)_{i j k}$ over the faces of these faces, etc., and sums all this up.
This defines a global action functional, which we may denote by
$$
  S(\phi) := \int_{\Sigma} j^\infty(\phi)^\ast(\{{L}_i\},\, \{(\kappa_p)_{i j}\}, \cdots)
  \,.
$$
This horizontal Cech-de Rham cocycle data is subject to fairly evident coboundary relations (gauge transformations)
that themselves are parameterized by systems $(\rho_\bullet)$ of $(p+1)-k$-forms on $k$-fold intersections:
$$
  \begin{aligned}
    L_i & \mapsto L_i + d_H (\rho_{p})_i
    \\
    (\kappa_{p})_{i j} & \mapsto (\kappa_p)_{i j} + d_H (\rho_{p-1})_{i j} + (\rho_p)_j - (\rho_p)_i
    \\
    \vdots
  \end{aligned}
$$
The definition of the global integral as above is preserved by these gauge transformations.
This is the point of the construction: if we had only integrated
the $(L_{p+1})_i$ over the cells of the triangulation without the contributions of the gluing forms
$(\kappa_{\bullet})$, then the resulting sum would not be invariant under the operation of shifting the
Lagrangians by horizontally exact terms (``total derivatives'') $L_i \mapsto L_i + d_H \rho_i$.

It might seem that this solves the problem. But there is one more subtlety: if the action
functional takes values in the real numbers, then the functions assigned to $(p+2)$-fold intersections
of patches are real valued, and then one may show that there exists a gauge transformation as above that collapses
the whole system of forms back to one globally defined Lagrangian form after all.
In other words: requiring a globally well-defined $\mathbb{R}$-valued action functional forces the
field theory to be globally variational, and hence rules out all locally variational field theories,
such as those of higher WZW-type.

But there is a simple way to relax the assumptions such that this restrictive conclusion is evaded.
Namely we may pick a discrete subgroup $\Gamma \hookrightarrow \mathbb{R}$ and relax the condition
on the functions $(\kappa_0)_{i_0 \cdots i_{p+1}}$ on $(p+2)$-fold intersections to the demand that on $(p+3)$-fold intersections
their alternating sum vanishes only modulo $\Gamma$. A system of $(p+2)-k$-forms on $k$-fold intersections
with functions regarded modulo $\Gamma$ this way is called a (horizontal) \emph{$\mathbb{R}/\Gamma$-Cech-Deligne cocycle} in
degree$(p+2)$.

For instance for field theories of WZW-type,
as above, we may take $\Gamma$ to be the discrete group of periods of the closed form $\omega$. Then
one may show that a lift of $\omega$ to a Cech-Deligne cocycle of local Lagrangians with gluing data
always exists. Indeed in general more than one inequivalent lift exists. The choice of these
lifts is a choice of prequantization.

However, modding out a discrete subgroup $\Gamma$ this ways also affects the induced global integral:
that integral itself is now only defined modulo the subgroup $\Gamma$:
$$
  S(\phi) = \int_{\Sigma} j^\infty(\phi)^\ast(\{(L_{(p+1)})_i\},\, \{(\kappa_{p})_{i j}\}, \cdots) \;\;\;\; \in \mathbb{R}/\Gamma
  \,.
$$

Now, there are not that many discrete subgroups of $\mathbb{R}$. There are the subgroups isomorphic to the integers, and then
there are dense subgroups, which make the quotient $\mathbb{R}/\Gamma$ ill behaved. Hence we
focus on the subgroup of integers.

The space of group inclusions $i : \mathbb{Z} \hookrightarrow \mathbb{R}$
is parameterized by a non-vanishing real number $2\pi \hbar \in \mathbb{R}-\{0\}$, given by $i : n \mapsto 2\pi \hbar n$.
The resulting quotient $\mathbb{R}/_{\!\hbar}\mathbb{Z}$ is isomorphic to the circle group $SO(2)\simeq U(1)$,
exhibited by the short exact exponential sequence
\begin{equation}
  \label{thexponentialsequenceforprequantization}
  \xymatrix{
    0
    \ar[r]
    &
    \mathbb{Z}
    \ar@{^{(}->}[rr]^-{2\pi \hbar(-)}
    &&
    \mathbb{R}
    \ar@{->>}[rr]^-{\exp(\tfrac{i}{\hbar}(-))}
    &&
    U(1)
    \ar[r]
    &
    0
  }
\end{equation}

Hence in the case that we take $\Gamma := \mathbb{Z}$, then we get locally variational field theories whose action
functional is well defined modulo $2\pi\hbar$. Equivalently the exponentiated action functional
is well defined as a function with values in $U(1)$:
$$
  \exp(\tfrac{i}{\hbar}S(\phi)) \; \in \; U(1) \simeq \mathbb{R}/_{\!\hbar}\mathbb{Z}
  \,.
$$

The appearance of Planck's constant $\hbar$ here signifies that requiring a locally variational
classical field theory to have a globally well-defined action functional is related to preparing it
for quantization. Indeed, if we consider the above discussion for $p= 0$, then the
above construction reproduces equivalently Kostant-Souriau's concept of
geometric \emph{pre-quantization}. Accordingly we may think of the Cech-Deligne cocycle data
$(\{U_i\};\{(L_{p+1})_i\}, \{(\kappa_p)_{i j}\},\cdots)$ for general $p$ as encoding
 \emph{higher pre-quantum geometry}.

\medskip

Coming back to the formulation of variational calculus in terms of diagrammatics in the
sheaf topos $\mathrm{Sh}(\mathrm{PDE}_\Sigma)$ as in section \ref{Principleofextremalactioncomonadically} above,
what we, therefore, are after is a context
in which the moduli object $\mathbf{\Omega}^{p+1}_H$ of globally defined horizontal $(p+1)$-forms
may be promoted to an object which we are going to denote $\mathbf{B}^{p+1}_H(\mathbb{R}/_{\!\hbar}\mathbb{Z})_{\mathrm{conn}}$
and which modulates horizontal Cech-Deligne cocycles, as above.

Standard facts in homological algebra and sheaf cohomology say that in order to achieve this we are to pass
from the category of sheaves on $\mathrm{PDE}_\Sigma$ to the ``derived category'' over $\mathrm{PDE}_\Sigma$.
We may take this to be the category of chain complexes of sheaves, regarded as a homotopy theory by understanding
that a morphism of sheaves of chain complexes that is locally a quasi-isomorphism counts as a weak equivalence.
In fact we may pass a bit further. Using the Dold-Kan correspondence to identify chain complexes in
non-negative degree with simplicial abelian groups, hence with group objects in Kan complexes,
we think of sheaves of chain complexes as special cases
of sheaves of Kan complexes \cite{Brown}:
$$
  \xymatrix{
    \mathrm{Sh}(\mathrm{PDE}_\Sigma, \mathrm{ChainCplx})
    \ar[rr]^-{\mathrm{Dold-Kan}}
  &&
  \mathrm{Sh}(\mathrm{PDE}_\Sigma,\mathrm{KanCplx})
  \simeq
  \mathrm{Sh}_\infty(\mathrm{PDE}_\Sigma)
  }
  \,.
$$
In such a homotopy-theoretically enlarged context we find the sheaf of chain complexes
that is the $(p+1)$-truncated de Rham complex with the integers included into the 0-forms:
$$
  \mathbf{B}^{p+1}_H( \mathbb{R}/_{\!\hbar}\mathbb{Z})_{\mathrm{conn}}
  :=
  [\mathbb{Z} \stackrel{2\pi \hbar}{\hookrightarrow} \mathbf{\Omega}^0_H \stackrel{d_H}{\to} \mathbf{\Omega}^2_H \stackrel{d_H}{\to} \cdots \stackrel{d_H}{\to}\mathbf{\Omega}^{p+1}_H]
  \,.
$$
This chain complex of sheaves is known as the (horizontal) \emph{Deligne complex} in degree $(p+2)$.
The horizontal Cech-Deligne cocycles that we saw before are
exactly the cocycles in the sheaf hypercohomology with coefficients in the horizontal Deligne complex.
Diagrammatically in $\mathrm{Sh}_\infty(\mathrm{PDE}_\Sigma)$ these are simply morphisms
$\mathbf{L} : E \to \mathbf{B}^{p+1}(\mathbb{R}/_{\!\hbar}\mathbb{Z})$ from the field
bundle to the Deligne moduli:
$$
  \left\{ (\{U_i\}, \{(L_{p+1})_i\}, \{(\kappa_p)_{i j}\}, \cdots ) \right\}
  \;\;\;
  \simeq
  \;\;\;
  \left\{
    E \stackrel{\mathbf{L}}{\longrightarrow} \mathbf{B}^{p+1}_H( \mathbb{R}/_{\!\hbar}\mathbb{Z})_{\mathrm{conn}}
  \right\}
  \,.
$$
This is such that a smooth homotopy between two maps to the Deligne moduli is equivalently a
coboundary of Cech cocycles:
$$
  \hspace{-.6cm}
  \left\{
     \raisebox{34pt}{
     \xymatrix{
        (\{U_i\}, \{(L_{p+1})_i\}, \{(\kappa_p)_{i j}\}, \cdots )
        \ar[dd]^{\left( \{U_i\}, \{(\rho_p)_i\}, \{(\rho_{p-1})_{i j}\}, \cdots \right)}
        \\
        \\
        (\{U_i\}, \{(L_{p+1})_i + d_h (\rho_p)_i \}, \{(\kappa_p)_{i j} + d_H (\rho_{p-1})_{i j} + (\rho_p)_j - (\rho_p)_i\}, \cdots )
     }
     }
  \right\}
  \;\;\;\;
  \simeq
  \;\;\;\;
  \left\{
    \xymatrix{
       E
       \ar@/^2pc/[rr]^{\mathbf{L}}_{\ }="s"
       \ar@/_2pc/[rr]_{\mathbf{L}'}^{\ }="t"
       &&
       \mathbf{B}^{p+1}_H(\mathbb{R}/_{\!\hbar}\mathbb{Z})
       \ar@{=>} "s"; "t"
    }
  \right\}
$$
Evidently, the diagrammatics serves as a considerable compression of data. In the following all
diagrams we displays are filled with homotopies as on the right above, even if we do not always make them
notationally explicit.

There is an evident morphism
$\mathbf{\Omega}^{p+1}_{H} \longrightarrow \mathbf{B}^{p+1}_H( \mathbb{R}/_{\!\hbar}\mathbb{Z})_{\mathrm{conn}}$
which includes the globally defined horizontal forms into the horizontal Cech-Deligne cocycles
(regarding them as Cech-Deligne cocycles with all the gluing data $(\kappa_\bullet)$ vanishing).
This morphism turns out to be the analog of a covering map in traditional differential geometry,
it is an atlas of smooth stacks:
\begin{center}
\begin{tabular}{|c|c|}
  \hline
  \begin{tabular}{c}
    {\bf atlas of}
    \\
    {\bf a smooth manifold}
  \end{tabular}
  &
  \begin{tabular}{c}
    {\bf atlas of}
    \\
    {\bf a smooth $\infty$-groupoid}
  \end{tabular}
  \\
  \hline
  $\xymatrix{
     \coprod_i U_i
     \ar@{->>}[d]
     \\
     E
    }$
    &
    $\xymatrix{
      \mathbf{\Omega}_H^{p+1}
      \ar@{->>}[d]
      \\
      \mathbf{B}^{p+1}_{H}(\mathbb{R}/_{\!\hbar}\mathbb{Z})_{\mathrm{conn}}
    }$
  \\
  \hline
\end{tabular}
\end{center}

Via this atlas, the Euler variational differential $\delta_V$ on horizontal forms that we have seen in
section \ref{Principleofextremalactioncomonadically} extends to horizontal Deligne coefficients to induce a curvature map on these coefficients.
$$
  \xymatrix{
    \mathbf{\Omega}^{\bullet\leq p+1}
    \ar@{->>}[drr]
    \ar[rrrr]^{\delta_V}
    && &&
    \mathbf{\Omega}^{p+1}_{S,\mathrm{cl}}
    \\
    &&
    \mathbf{B}^{p+1}_H( \mathbb{R}/_{\!\hbar}\mathbb{Z} )_{\mathrm{conn}}
    \ar[urr]_{\mathrm{curv}}
    &&
  }
$$
A \emph{prequantization} of a source form $\mathrm{EL}$ is a lift through this curvature map,
hence a horizontal Cech-Deligne cocycle of locally defined local Lagrangians for $\mathrm{EL}$, equipped with gluing data:
$$
  \xymatrix{
    && \mathbf{B}^{p+1}_H(\mathbb{R}/_{\!\hbar}\mathbb{Z})_{\mathrm{conn}}
    \ar[d]^{\mathrm{curv}}
    \\
    E
    \ar[rr]_-{\mathrm{EL}}
    \ar@{-->}[urr]^{\mathbf{L}}
    &&
    \mathbf{\Omega}^{p+1,1}_S
  }
  \,.
$$

Hence in conclusion we find that in the $\infty$-topos $\mathrm{Sh}_\infty(\mathrm{PDE}_\Sigma)$ the diagrammatic picture of prequantum local field theory
is this:

\medskip
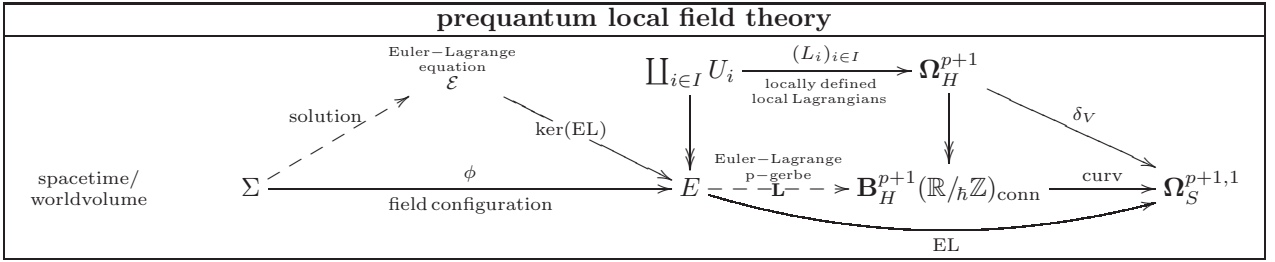
\begin{tabular}{|c|}
\hline
{\bf prequantum local field theory}
\\
\hline
$
  \xymatrix@C=40pt{
    &&
    {{{\mathrm{Euler-Lagrange} \atop \mathrm{equation}}} \atop {\mathcal{E}}}
    \ar[dr]|{\mathrm{ker}(\mathrm{EL})}
    &
    \coprod_{i \in I} U_i
    \ar@{->>}[d]^{\ }="t"
    \ar[r]^-{(L_i)_{i \in I}}_{\, {\mathrm{locally}\,\mathrm{defined}} \atop {\mathrm{local}\, \mathrm{Lagrangians}}  }
    &
    \mathbf{\Omega}^{p+1}_H
    \ar@{->>}[d]_{\ }="s"
    \ar[dr]^{\delta_V}
    \\
     {\mathrm{spacetime/}
     \atop
     \mathrm{worldvolume}}
     \!\!\!\!\!\!\!
     &
    \Sigma
    \ar@{-->}[ur]^{\mathrm{solution}}
    \ar[rr]^-{\phi}_{\mathrm{field}\,\mathrm{configuration}}
    &&
    E
    \ar@{-->}[r]|-{\mathbf{L}}^-{ \mathrm{Euler-Lagrange} \atop \mathrm{p-gerbe} }
    \ar@/_1.3pc/[rr]_{\mathrm{EL}}
    &
    \mathbf{B}^{p+1}_H(\mathbb{R}/_{\!\hbar}\mathbb{Z})_{\mathrm{conn}}
    \ar[r]^-{\mathrm{curv}}
    &
    \mathbf{\Omega}^{p+1,1}_S
    %
    %\ar@{=>}|\simeq "s"; "t"
  }
$
\\
\hline
\end{tabular}

\medskip

In summary, comparing this to the diagrammatics for variational and locally variational
classical field theory which we discussed in section \ref{Principleofextremalactioncomonadically},
we have the following three levels of description of local Lagrangian field theory:
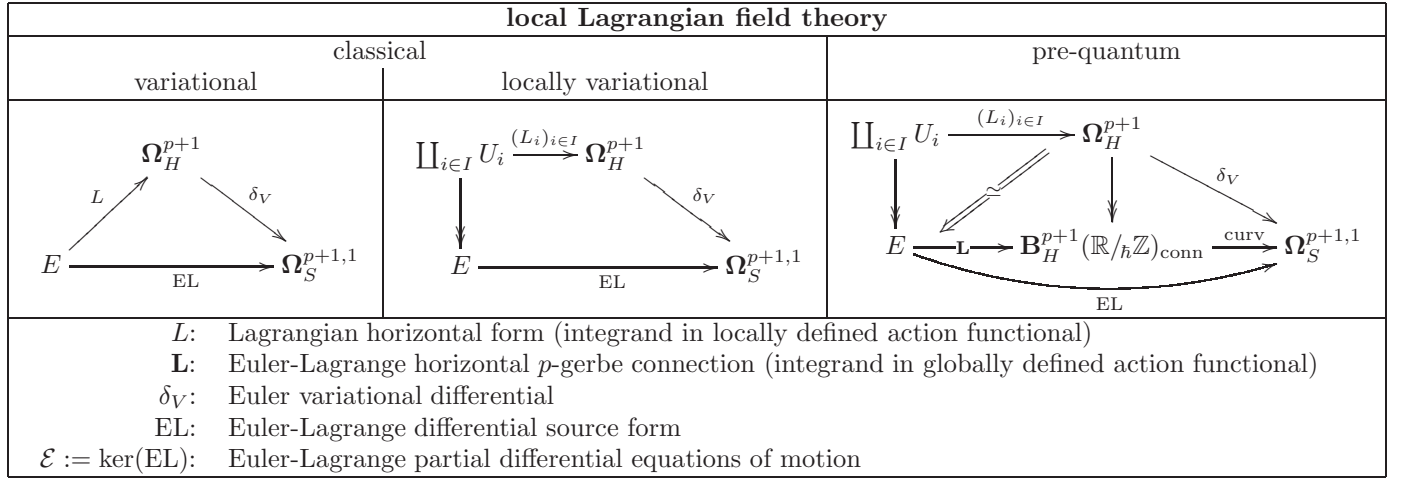
\begin{center}
\begin{tabular}{|c|c|c|}
  \hline
  \multicolumn{3}{|c|}{{\bf local Lagrangian field theory}}
  \\
  \hline
  \multicolumn{2}{|c|}{\hspace{-1cm}classical} & pre-quantum
  \\
  variational & locally variational &
  \\
  \hline
  $
  \raisebox{-5pt}{
  \xymatrix{
    &
    \mathbf{\Omega}^{p+1}_H
    \ar[dr]^{\delta_V}
    \\
    E
    \ar[ur]^{L}
    \ar[rr]_{\mathrm{EL}}
    &&
    \mathbf{\Omega}^{p+1,1}_S
  }
  }
 $
  &
  $
  \raisebox{-5pt}{
  \xymatrix{
    \coprod_{i \in I} U_i
    \ar@{->>}[d]
    \ar[r]^-{(L_i)_{i \in I}}
    &
    \mathbf{\Omega}^{p+1}_H
    \ar[dr]^{\delta_V}
    \\
    E
    \ar[rr]_{\mathrm{EL}}
    &&
    \mathbf{\Omega}^{p+1,1}_S
  }
  }
 $
  &
  $\xymatrix{
    \coprod_{i \in I} U_i
    \ar@{->>}[d]
    \ar[r]^-{(L_i)_{i \in I}}_>>>{\ }="s"
    &
    \mathbf{\Omega}^{p+1}_H
    \ar@{->>}[d]
    \ar[dr]^{\delta_V}
    \\
    E
    \ar[r]|-{\mathbf{L}}^<<<{\ }="t"
    \ar@/_1.3pc/[rr]_{\mathrm{EL}}
    &
    \mathbf{B}^{p+1}_H(\mathbb{R}/_{\!\hbar}\mathbb{Z})_{\mathrm{conn}}
    \ar[r]^-{\mathrm{curv}}
    &
    \mathbf{\Omega}^{p+1,1}_S
    \ar@{=>}|\simeq "s"; "t"
  }$
  \\
  \hline
  \multicolumn{3}{|l|}{
    \begin{tabular}{rl}
      $L$: & Lagrangian horizontal form (integrand in locally defined action functional)
      \\
      $\mathbf{L}$: & Euler-Lagrange horizontal $p$-gerbe connection (integrand in globally defined action functional)
      \\
      $\delta_V$: & Euler variational differential
      \\
      $\mathrm{EL}$: & Euler-Lagrange differential source form
      \\
      $\mathcal{E} := \mathrm{ker}(\mathrm{EL})$: & Euler-Lagrange partial differential equations of motion
    \end{tabular}
  }
  \\
  \hline
\end{tabular}
\end{center}

%%%%%%%%%%%%%%%%%%%%%%%%%%%%%%%%%%%%%%%%%%%%%
\subsection{The covariant phase space -- transgressively}
\label{elgerbesinintroduction}
%%%%%%%%%%%%%%%%%%%%%%%%%%%%%%%%%%%%%%%%%%%%%%%%%

The Euler-Lagrange $p$-gerbes discussed above are singled out as being exactly the
right coherent refinement of locally defined local Lagrangians that may be integrated over a $(p+1)$-dimensional spacetime/worldvolume
to produce a \emph{function}, the action functional.
In a corresponding manner there are further refinements of locally defined Lagrangians by
differential cocycles  that are adapted to integration over submanifolds of $\Sigma_{p+1}$ of positive codimension.
In codimesion $k$ these will yield not functions, but $(p-k)$-gerbes.

We consider this now for codimension 1 and find the covariant phase space of a locally variational field theory
equipped with its canonical (pre-)symplectic structure and equipped with a Kostant-Souriau prequantization of that.

\medskip

First consider the process of transgression in general codimension.

Given a smooth manifold $\Sigma$, then the mapping space $[\Sigma,\mathbf{\Omega}^{p+2}]$
into the smooth moduli space of $(p+2)$-forms
is the smooth space defined by the property that for any other smooth manifold $U$, there is a
natural identification
$$
  \left\{
    U \longrightarrow [\Sigma, \mathbf{\Omega}^{p+2}]
  \right\}
  \;\;
  \simeq
  \;\;
  \Omega^{p+2}(U \times \Sigma)
$$
of smooth maps into the mapping space with smooth $(p+2)$-forms on the product manifold $U \times \Sigma$.

Now suppose that $\Sigma = \Sigma_{d}$ is an oriented closed smooth manifold of dimension $d$. Then
there is the fiber integration of differential forms on $U \times \Sigma$ over $\Sigma$ (e.g \cite{BottTu}), which gives a map
$$
  \underset{(U \times \Sigma_{d}) / U}{\int} : \Omega^{p+2}(U \times \Sigma_{d}) \longrightarrow \Omega^{p+2-d}(U)
  \,.
$$
This map is natural in $U$, meaning that it is compatible with pullback of differential forms along any
smooth function $U_1 \to U_2$. This property is precisely what is summarized by saying that the
fiber integration map constitutes a morphism in the sheaf topos of the form
$$
  \int_\Sigma \;:\; [\Sigma_{d}, \mathbf{\Omega}^{p+2}] \longrightarrow \mathbf{\Omega}^{p+2-d}
  \,.
$$
This provides an elegant means to speak about transgression. Namely given a differential form $\alpha \in \Omega^{p+2}(X)$
(on any smooth space $X$)
modulated by a morphism $\alpha : X \longrightarrow \mathbf{\Omega}^{p+2}$, then its transgression to the mapping space
$[\Sigma,X]$ is simply the form in $\Omega^{p+2-d}([\Sigma,X])$ which is modulated by the composite
$$
  \int_{\Sigma} [\Sigma, -]
  \;:\;
  [\Sigma,X]
  \stackrel{[\Sigma,\alpha]}{\longrightarrow}
  [\Sigma, \mathbf{\Omega}^{p+2}]
  \stackrel{\int_{\Sigma}}{\longrightarrow}
  \mathbf{\Omega}^{p+2-d}
$$
of the fiber integration map above with the image of $\alpha$ under the functor $[\Sigma,-]$ that forms
mapping spaces out of $\Sigma$.

Moreover, this statement has a prequantization \cite[2.8]{FiorenzaSatiSchreiberIV}: the fiber integration of curvature
forms lifts to a morphism of differential cohomology coefficients
$$
  \int_\Sigma \;:\;
  [\Sigma, \mathbf{B}^{p+1}(\mathbb{R}/_{\!\hbar}\mathbb{Z})_{\mathrm{conn}}]
   \longrightarrow
   \mathbf{B}^{p+1-d} (\mathbb{R}/_{\!\hbar}\mathbb{Z})_{\mathrm{conn}}
$$
and hence the transgression of a $p$-gerbe $\nabla : X \longrightarrow \mathbf{B}^{p+1}(\mathbb{R}/_{\!\hbar}\mathbb{Z})_{\mathrm{conn}}$
(on any smooth space $X$) to the mapping space $[\Sigma, X]$ is given by the composite $\int_\Sigma \circ [\Sigma,-]$
$$
  \int_\Sigma [\Sigma,-]
  \;:\;
  [\Sigma,X] \stackrel{[\Sigma,\nabla]}{\longrightarrow} [\Sigma,\mathbf{B}^{p+1}(\mathbb{R}/_{\!\hbar}\mathbb{Z})_{\mathrm{conn}}]
    \stackrel{\int_\Sigma}{\longrightarrow}
    \mathbf{B}^{p+1-d}(\mathbb{R}/_{\!\hbar}\mathbb{Z})_{\mathrm{conn}}
  \,.
$$

All this works verbatim also in the context of PDEs over $\Sigma$. For instance if $L : E \longrightarrow \mathbf{\Omega}^{p+1}_H$
is a local Lagrangian on (the jet bundle of) a field bundle $E$ over $\Sigma_{p+1}$ as before, then the action functional that it induces,
as in section \ref{globalactionfunctionalinintroduction},
is the transgression to $\Sigma_{p+1}$:
$$
  S : [\Sigma,E]_\Sigma \stackrel{[\Sigma,L]_\Sigma}{\longrightarrow} [\Sigma, \mathbf{\Omega}^{p+1}_H]_\Sigma
   \stackrel{\int_\Sigma}{\longrightarrow} \mathbf{\Omega}^0
   \,.
$$
But now the point is that we have the analogous construction in higher codimension $k$, where the Lagrangian does not
integrate to a function (a differential 0-form) but to a differential $k$-form.
And all this goes along with passing from globally defined differential forms to Cech-Deligne cocycles.

To apply this for codimension $k = 1$,
consider now $p$-dimensional submanifolds $\Sigma_p \hookrightarrow \Sigma$ of spacetime/worldvolume.
We write $N^\infty_\Sigma \Sigma_p$ for the infinitesimal normal neighbourhood of $\Sigma_p$ in
$\Sigma$.
In
practice one is often, but not necessarily, interested in $\Sigma_p$ being a
\emph{Cauchy surface}, which means that the induced restriction map
$$
  [\Sigma_{p+1},\mathcal{E}]
  \longrightarrow
  [N^\infty_\Sigma \Sigma_p, \mathcal{E}]
$$
(from field configurations solving the equations of motion on all of $\Sigma$ to normal jets of
solutions on $\Sigma_p$) is an equivalence.
An element in the solution space $[\Sigma_{p+1},\mathcal{E}]$ is a \emph{classical state} of the physical
system that is being described, a classical trajectory of a field configuration over all of spacetime. Its image in
$[\Sigma_{p+1},\mathcal{E}]$ is the restriction of that field configuration and of all its derivatives to
$\Sigma_p$.

In many -- but not in all -- examples of interest, classical trajectories are fixed once
their first order derivatives over a Cauchy surface is known. In these cases the phase space
may be identified with the cotangent bundle of the space of field configurations on the
Cauchy surface
$$
  [N^\infty_\Sigma \Sigma_p, \mathcal{E}]
  \simeq
  T^\ast [\Sigma_p, E]
  \,.
$$
The expression on the right is often taken as the definition of \emph{phase spaces}.
But since the equivalence with the left hand side does not hold generally,
we will not restrict attention to this simplified case and instead consider the solution space $[\Sigma,\mathcal{E}]_\Sigma$
as the phase space. To emphasize this more general point of view, one sometimes speaks of the
\emph{covariant phase space}. Here ``covariance'' refers to invariance under the action of the diffeomorphism
group of $\Sigma$, meaning here that no space/time split in the form of a choice of Cauchy surface is made
(or necessary) to define the phase space, even if a choice of Cauchy surface is possible and potentially useful for parameterizing
phase space in terms of initial value data.

Now it is crucial that the covariant phase space $[\Sigma,\mathcal{E}]_\Sigma$ comes equipped with further
geometric structure which remembers that this is not just any old space, but the space of solutions of a
locally variational differential equation.

To see how this comes about, let us write $(\mathbf{\Omega}^{p+1}_{\mathrm{cl}})_{\Sigma}$ for the
moduli space of all closed $p+1$-forms on PDEs.
This is to mean that if $E$ is a bundle over $\Sigma$, and regarded as representing the space of solutions
 of the trivial PDE on sections of $E$, then morphisms $E \longrightarrow (\mathbf{\Omega}^{p+1})_{\Sigma}$ are equivalent
to closed differential $(p+2)$-forms on the jet bundle of $E$.
$$
  \left\{
    E \longrightarrow (\mathbf{\Omega}^{p+1})_\Sigma
  \right\}
  \;\;\simeq\;\;
  \Omega^{p+1}_{\mathrm{cl}}(J^\infty_\Sigma E)
  \,.
$$
The key now is that there is a natural filtration on these differential forms adapted to spacetime codimension.
This is part of a bigrading structure on differential forms on jet bundles known as
the \emph{variational bicomplex} \cite{Anderson-book}.
In its low stages it looks as follows:
$$
  \xymatrix{
    (\mathbf{\Omega}^{p+1})_{\Sigma}
    \ar[r]
    &
    \cdots
    \ar[r]
    &
    \mathbf{\Omega}^{p+1,1}_{S} \oplus \mathbf{\Omega}^{p,2}
    \ar[r]
    &
    \mathbf{\Omega}^{p+1,1}_S
  }
  \,.
$$
The lowest item here is what had concerned us in section \ref{Principleofextremalactioncomonadically}
and \ref{globalactionfunctionalinintroduction}, it is the moduli of $p+2$-forms which have $p+1$ of their legs along
spacetime/worldvolume $\Sigma$ and whose remaining vertical leg along the space of local field configurations
depends only on the field value itself, not on any of its derivatives. This was precisely the correct recipient
of the variational curvature, hence the variational differential of horizontal $(p+1)$-forms representing local Lagrangians.

But now that we are moving up in codimension, this coefficient will disappear, as these forms do not contribute
when integrating just over $p$-dimensional hypersurfaces. The correct coefficient for that case is instead clearly
$\mathbf{\Omega}^{p,2}$, the moduli space of those $(p+2)$-forms on jet bundles which have $p$ of their
legs along spacetime/worldvolume, and the remaining two along the space of local field configurations.
(There is a more abstract way to derive this filtration from first principles,
and which explains why we have restriction to ``source forms'' (not differentially depending on the jets),
indicated by the subscript, only in the bottom row. But for the moment we just take that little subtlety for granted.)

So $\mathbf{\Omega}^{p,2}$ is precisely the space of those $(p+2)$-forms on the jet bundle that
become (pre-)symplectic 2-forms on the space of field configurations once evaluated on a $p$-dimensional
spatial slice $\Sigma_p$ of spacetime $\Sigma_{p+1}$. We may think of this as a \emph{current} on spacetime with values in
2-forms on fields.

Indeed, there is a \emph{canonical} such \emph{presymplectic current} for every locally variational
field theory \cite{Zuckerman}\cite{Kh-covar}. To see this, we ask for a lift of the purely horizontal locally defined Lagrangian $L_i$
through the variational bicomplex to a $(p+1)$-form on the jet bundle whose curvature $d(L_i + \Theta_i)$
coincides with the Euler-Lagrange form $\mathrm{EL} = \delta_V L_i$ in vertical degree 1.
Such a  lift
$L_i + \Theta_i$ is known as a \emph{Lepage form} for $L_i$ (e.g. \cite[2.1.2]{GMStextbook}).

Notice that it is precisely the restriction to the shell $\mathcal{E}$ that makes the Euler-Lagrange form
$\mathrm{EL}_i$ disappear, by construction, so that only the new curvature component $\Omega_i$ remains as the curvature of the
Lepage form on shell:
$$
  \xymatrix{
    \mathcal{E}|_{U_i} \ar[d] \ar@/^1pc/[rr]^{\Omega_i = d_V \Theta_i} & & \mathbf{\Omega}^{p+1,1}_S \oplus \mathbf{\Omega}^{p,2}
    \ar[d]
    \\
    E|_{U_i} \ar@{-->}[urr]^{d(L_i + \Theta_i)} \ar[rr]_{\mathrm{EL}_i = \delta_V L_i} && \mathbf{\Omega}^{p+1,1}_S
  }
$$

The condition means that the horizontal
differential of $\Theta_i$ has to cancel against the horizontally exact part that appears when decomposing the
differential of $L_i$ as in equation \ref{differentialofLagrangian}. Hence, up to horizontal derivatives, this
$\Theta_i$ is in fact uniquely fixed by $L_i$:
$$
  \begin{aligned}
    d (L_i + \Theta_i) & = (\mathrm{EL}_i - d_H (\Theta_i + d_H(\cdots)))  + (d_H \Theta_i + d_V \Theta_i)
    \\
    & = \mathrm{EL}_i + d_V \Theta_i
    \\
    & =: \mathrm{EL}_i + \Omega_i
  \end{aligned}
  \,.
$$
The new curvature component $\Omega \in \Omega^{p,2}(J^\infty_\Sigma E)$
whose restriction to patches is given this way,
$\Omega|_{U_i} := d_V \Theta_i$,
is known as the \emph{presymplectic current} \cite{Zuckerman}, \cite{Kh-covar}. Because, by the way we
found its existence, this is such that its transgression over a codimension-1 submanifold
$\Sigma_p \hookrightarrow \Sigma$ yields a closed 2-form (a ``presymplectic 2-form'') on the
covariant phase space:
$$
  \omega := \int_{\Sigma_p} [\Sigma_p, \Omega]
  \in
  \Omega^2([N^\infty_\Sigma \Sigma_p, \mathcal{E}])
  \,.
$$
Since $\Omega$ is uniquely specified by the local Lagrangians $L_i$, this gives the covariant phase space
canonically the structure of a presymplectic space $([\Sigma,\mathcal{E}], \omega)$.
This is the reason why phase spaces in classical mechanics are given by (pre-)symplectic geometry
as in \cite{Arnold89}.

Since $\Omega$ is a conserved current, the canonical presymplectic form $\omega$ is indeed canonical,
it does not depend on the choice of (Cauchy-)surface: if $\partial_{\mathrm{in}} \Sigma$ and $\partial_{\mathrm{out}}\Sigma$
are the incoming and outgoing Cauchy surfaces, respectively, in a piece of spacetime $\Sigma$, then the corresponding
presymplectic forms agree\footnote{\label{antibracketremark} If the shell $\mathcal{E}$ is taken to be resolved by a derived manifold/BV-complex as in
footnote \ref{BVremark}, then any on-shell vanishing condition becomes vanishing up to a $d_{\mathrm{BV}}$-exact term, hence
then there is a 2-form $\omega_{\mathrm{BV}}$ of BV-degreee -1 such that
$\omega_{\mathrm{out}} - \omega_{\mathrm{in}} = d_{\mathrm{BV}}\omega_{\mathrm{BV}}$.
(In \cite{CattaneoMnevReshetikhin12} this appears as BV-BFV axiom (9).) The Poisson bracket induced
from this ``shifted symplectic form'' $\omega_{\mathrm{BV}}$ is known as the ``BV-antibracket'' (e.g. \cite{HenneauxTeitelboim}).}

$$
   \omega_{\mathrm{out}} - \omega_{\mathrm{in}} = 0
   \,.
$$

But by the discussion in \ref{globalactionfunctionalinintroduction}, we do not just consider a locally variational field classical field
theory to start with, but a prequantum field theory. Hence in fact there is more data before transgression
than just the new curvature components $d_V \Theta_i$, there is also Cech cocycle coherence data that
glues the locally defined $\Theta_i$ to a globally consistent differential cocycle.

We write $\mathbf{B}^{p+1}_L(\mathbb{R}/_{\!\hbar}\mathbb{Z})$ for the moduli space for such coefficients
(with the subscript for ``Lepage''), so that morphisms $E \longrightarrow \mathbf{B}^{p+1}_L(\mathbb{R}/_{\!\hbar}\mathbb{Z})$
are equivalent to properly prequantized globally defined Lepage lifts of Euler-Lagrange $p$-gerbes.
In summary then, the refinement of an Euler-Lagrange $p$-gerbe $\mathbf{L}$ to a Lepage-$p$-gerbe
$\Theta$ is given by the following diagram
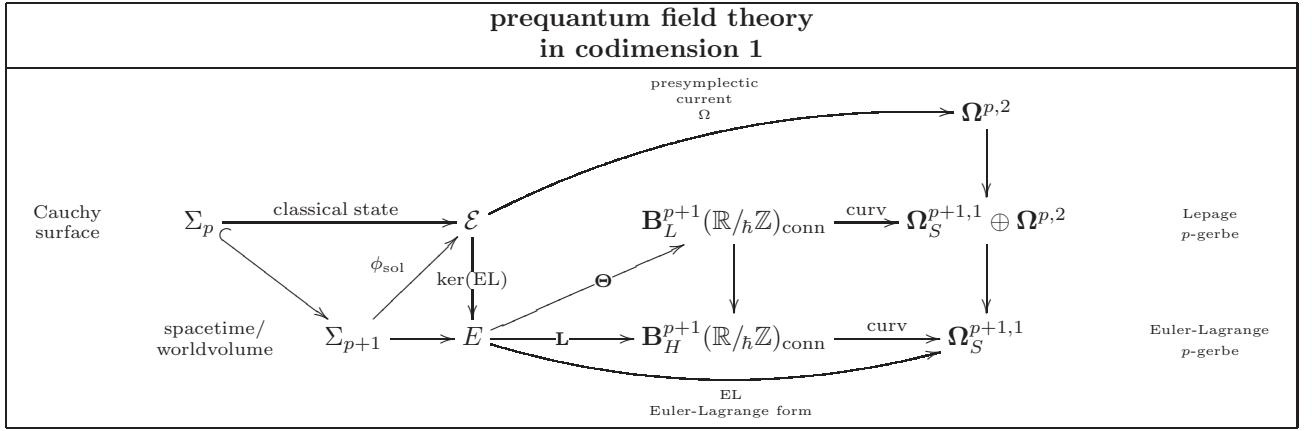
\begin{center}
\begin{tabular}{|c|}
\hline
 {\bf prequantum field theory}
 \\
 {\bf in codimension 1}
\\
\hline
$
  \xymatrix{
    && && &&
    \mathbf{\Omega}^{p,2}
    \ar[d]
     \\
     {\mathrm{Cauchy}
     \atop
     \mathrm{surface}}
     \!\!\!\!\!\!\!
    & \Sigma_p
     \ar[rr]^{\mathrm{classical}\;\mathrm{state}}
     \ar@{^{(}->}[dr]
    &&
    \mathcal{E} \ar[d]|{\mathrm{ker}(\mathrm{EL})}
    %\ar[rr]|-{}
    \ar@/^1.2pc/[urrr]^{\mbox{ \tiny \begin{tabular}{c} presymplectic \\ current \\ $\Omega$  \end{tabular} } }
    &&
    \mathbf{B}^{p+1}_L(\mathbb{R}/_{\!\hbar}\mathbb{Z})_{\mathrm{conn}}
    \ar[d] \ar[r]^-{\mathrm{curv}}
    &
    \mathbf{\Omega}^{p+1,1}_S \oplus \mathbf{\Omega}^{p,2}
    \ar[d]
    &
    {\mbox{\tiny Lepage} \atop \mbox{\tiny $p$-gerbe}}
    \\
    &
     {\mathrm{spacetime/}
     \atop
     \mathrm{worldvolume}}
     \!\!\!\!\!\!\!
    &
    \Sigma_{p+1}
    \ar[ur]^{\phi_{\mathrm{sol}}}
    \ar[r]
    &
    E
    \ar[urr]|{\mathbf{\Theta}}
    \ar[rr]|-{\mathbf{L}}^<<<{\ }="t"
    \ar@/_1.3pc/[rrr]_{\mbox{ \tiny \begin{tabular}{c} $\mathrm{EL}$ \\ Euler-Lagrange form \end{tabular} }}
    &&
    \mathbf{B}^{p+1}_H(\mathbb{R}/_{\!\hbar}\mathbb{Z})_{\mathrm{conn}}
    \ar[r]^-{\mathrm{curv}}
    &
    \mathbf{\Omega}^{p+1,1}_S
    &
    {\mbox{\tiny Euler-Lagrange} \atop \mbox{\tiny $p$-gerbe}}
  }
$
\\
\hline
\end{tabular}
\end{center}
\medskip

And now higher prequantum geometry bears fruit: since transgression is a natural operation, and since
the differential coefficients $\mathbf{B}_H^{p+1}(\mathbb{R}/_{\!\hbar}\mathbb{Z})_{\mathrm{conn}}$
and $\mathbf{B}^{p+1}_L(\mathbb{R}/_{\!\hbar}\mathbb{Z})$ precisely yield the coherence data to make the
local integrals over the locally defined differential forms $L_i$ and $\Theta_i$ be globally well defined,
we may now hit this entire diagram with the transgression functor $\int_{\Sigma_p} [N^\infty_\Sigma \Sigma_p,-]$
to obtain this diagram:

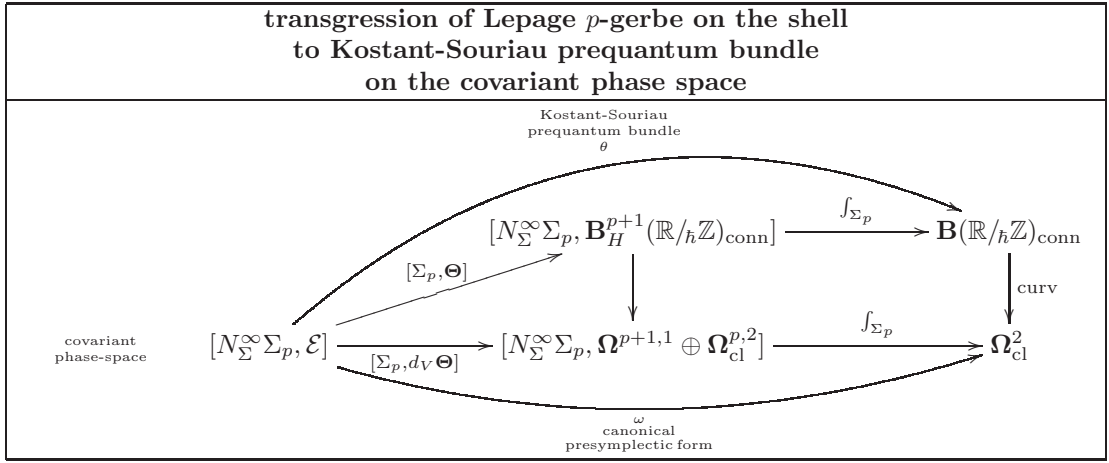
\begin{center}
\begin{tabular}{|c|}
\hline
{\bf transgression of Lepage $p$-gerbe on the shell}
\\
{\bf to Kostant-Souriau prequantum bundle}
\\
{\bf on the covariant phase space}
\\
\hline
$
  \xymatrix{
   &
   &&
   [N^\infty_\Sigma \Sigma_p, \mathbf{B}^{p+1}_H(\mathbb{R}/_{\!\hbar}\mathbb{Z})_{\mathrm{conn}}]
   \ar[d]
   \ar[rr]^-{\int_{\Sigma_p}}
   &&
   \mathbf{B}(\mathbb{R}/_{\!\hbar}\mathbb{Z})_{\mathrm{conn}}
   \ar[d]^{\mathrm{curv}}
   \\
   \mbox{ \tiny \begin{tabular}{c} covariant \\ phase-space \end{tabular} }
   \!\!\!\!\!\!\!\!\!
   &
   [N^\infty_\Sigma \Sigma_p, \mathcal{E}]
   \ar@/^4pc/[urrrr]^{\mbox{ \tiny \begin{tabular}{c} Kostant-Souriau \\ prequantum bundle \\ $\mathbf{\theta}$\end{tabular} }}
   \ar@/_2pc/[rrrr]_{\omega \atop { \mathrm{canonical}\atop {\mathrm{presymplectic}\, \mathrm{form}} }}
   \ar[rr]_-{[\Sigma_p, d_V \mathbf{\Theta}]}
   \ar[urr]^-{[\Sigma_p, \mathbf{\Theta}]}
   &&
   [N^\infty_\Sigma \Sigma_p, \mathbf{\Omega}^{p+1,1} \oplus \mathbf{\Omega}^{p,2}_{\mathrm{cl}}]
   \ar[rr]^-{\int_{\Sigma_p}}
   &&
   \mathbf{\Omega}^2_{\mathrm{cl}}
  }
$
\\
\hline
\end{tabular}
\end{center}

This exhibits the transgression
$$
  \mathbf{\theta} := \int_{\Sigma_p} [N^\infty_\Sigma \Sigma_p,\mathbf{\Theta}]
$$
of the Lepage $p$-gerbe $\mathbf{\Theta}$ as a $(\mathbb{R}/_{\!\hbar}\mathbb{Z})$-connection
whose curvature is the canonical presymplectic form.

But this $([\Sigma,\mathcal{E}]_\Sigma,\mathbf{\theta})$ is just the structure that Souriau originally called and demanded as a prequantization of the
(pre-)symplectic phase space $([\Sigma,\mathcal{E}]_\Sigma,\omega)$ \cite{Souriau70, Souriau74, Kostant75}.
Conversely, we see that the Lepage $p$-gerbe $\mathbf{\Theta}$ is a ``de-transgression'' of the Kostant-Souriau
prequantization of covariant phase space in codimension-1 to a higher prequantization in full codimension. In particular,
the higher prequantization constituted by the Lepage $p$-gerbe constitutes a compatible choice of Kostant-Souriau
prequantizations of covariant phase space for \emph{all} choices of codimension-1 hypersurfaces at once.
This is a genuine reflection of the fundamental locality of the field theory, even if we look at field
configurations globally over all of a (spatial) hypersurface $\Sigma_p$.

%%%%%%%%%%%%%%%%%%%%%%%%%%%%%%%%%%%%%%%%%%%%%%%%%%%%
\subsection{The local observables -- Lie theoretically}
\label{symmetriescurrentsinintroduction}
%%%%%%%%%%%%%%%%%%%%%%%%%%%%%%%%%%%%%%%%%%%%%%%%%%%%

We discuss now how from the previous considerations naturally follow the concepts of local observables
of field theories
and of the Poisson bracket on them, as well as the concept of conserved currents and the variational Noether theorem relating them
to symmetries. At the same time all these concepts are promoted to prequantum local field theory.

\medskip

In section \ref{elgerbesinintroduction} we have arrived at a perspective of prequantum local field theory
where the input datum is a partial differential equation of motion $\mathcal{E}$ on sections of a bundle $E$ over spacetime/worldvolume $\Sigma$
and equipped with a prequantization exhibited by a factorization of the Euler-Lagrange form
$E \stackrel{\mathrm{EL}}{\longrightarrow} \mathbf{\Omega}^{p+1,1}_{S}$ and of the presymplectic current form
$\mathcal{E} \stackrel{\Omega}{\longrightarrow} \mathbf{\Omega}^{p,2}$ through higher Cech-Deligne cocycles
for an Euler-Lagrange $p$-gerbe $\mathbf{L}$ and for a Lepage $p$-gerbe $\mathbf{\Theta}$:

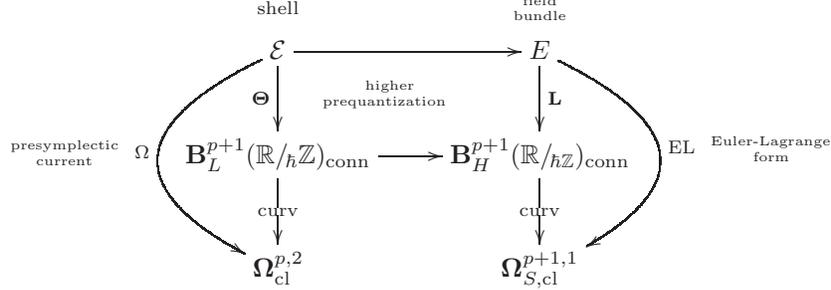
\vspace{-1cm}

$$
  \xymatrix{
    \ar@{}[d]|{\mathrm{shell}}
    &
    \ar@{}[d]|{\mathrm{field} \atop \mathrm{bundle}}
    \\
    \mathcal{E}
    \ar@{}[dr]|-{\;\;\mathrm{higher} \atop \mathrm{prequantization}}
    \ar[r]
    \ar@/_3.8pc/[dd]_-{\mbox{\tiny \begin{tabular}{c}presymplectic \\ current\end{tabular}} \Omega}
    \ar[d]_-{\mathbf{\Theta}}
    & E
    \ar@/^3.8pc/[dd]^-{\mathrm{EL} \mbox{\tiny \begin{tabular}{c}Euler-Lagrange\\ form\end{tabular}}}
    \ar[d]^-{\mathbf{L}}
    \\
    \mathbf{B}^{p+1}_L(\mathbb{R}/_{\!\hbar}\mathbb{Z})_{\mathrm{conn}}
    \ar[d]|-{\mathrm{curv}}
    \ar[r]
    &
    \mathbf{B}^{p+1}_H(\mathbb{R}/_{\!\hbar \mathbb{Z}})_{\mathrm{conn}}
    \ar[d]|-{\mathrm{curv}}
    \\
    \mathbf{\Omega}^{p,2}_{\mathrm{cl}}
    &
    \mathbf{\Omega}^{p+1,1}_{S,\mathrm{cl}}
  }
$$
This local data then transgresses to spaces of field configurations over codimension-k submanifolds
of $\Sigma$.
Transgressing to codimension-0 yields the globally defined exponentiated action functional
$$
  
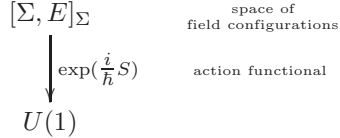
\xymatrix{
    [\Sigma,E]_\Sigma
    \ar[d]^-{\exp(\tfrac{i}{\hbar}S) \;\;\;\;\;\;\;\;\;\mbox{\tiny action functional}}
    & \mbox{\tiny \begin{tabular}{c} space of \\ field configurations\end{tabular}}
    \\
    U(1)
  }
$$
and transgressing to a codimension-1 (Cauchy-)surface $\Sigma_p \hookrightarrow \Sigma$
yields the covariant phase space as a prequantized pre-symplectic manifold

\vspace{-1.5cm}

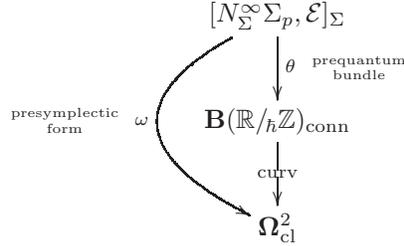
$$
  \raisebox{45pt}{
  \xymatrix{
    \ar@{}[d]|{\mathrm{covariant} \atop {\mathrm{phase}\, \mathrm{space} }}
    \\
    [N^\infty_\Sigma \Sigma_p,\mathcal{E}]_\Sigma
    \ar@/_3.8pc/[dd]_{\mbox{\tiny \begin{tabular}{c} presymplectic \\ form \end{tabular}} \omega }
    \ar[d]^-{\mathbf{\theta} \mbox{\tiny \begin{tabular}{c}prequantum \\ bundle\end{tabular}}}
    \\
    \mathbf{B}(\mathbb{R}/_{\!\hbar}\mathbb{Z})_{\mathrm{conn}}
    \ar[d]|-{\mathrm{curv}}
    \\
    \mathbf{\Omega}^2_{\mathrm{cl}}
  }
  }
  \,.
$$
Given any space equipped with a map into some moduli space like this, an automorphism of this
structure is a diffeomorphism of the space together with a homotopy which which exhibits the
preservation of the given map into the moduli space.

We consider now the automorphisms of the prequantized covariant phase space and of the Euler-Lagrange $p$-gerbe
that it arises from via transgression, and find that these recover and make globally well-defined the
traditional concepts of symmetries and conserved currents, related by the Noether theorem,
and of observables equipped with their canonical Poisson bracket.

\medskip

The correct automorphisms of presymplectic smooth spaces
$([\Sigma_p,\mathcal{E}]_\Sigma,\omega) \longrightarrow ([\Sigma_p,\mathcal{E}]_\Sigma,\omega)$
are of course diffeomorphisms $\phi : [\Sigma_p,\mathcal{E}]_\Sigma \longrightarrow [\Sigma_p,\mathcal{E}]_\Sigma$ such that
the presymplectic form is preserved, $\phi^\ast \omega = \omega$.
In the diagrammatics this means that $\phi$ fits into a triangle of this form:
$$
  
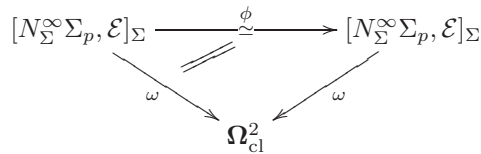
\xymatrix{
    [N^\infty_\Sigma \Sigma_p,\mathcal{E}]_\Sigma \ar[dr]_{\omega}^{\ }="t" \ar[rr]^-\phi|\simeq_-{\ }="s"
    && [N^\infty_\Sigma \Sigma_p,\mathcal{E}]_\Sigma \ar[dl]^{\omega}
    \\
    & \mathbf{\Omega}^2_{\mathrm{cl}}
    \ar@{=} "s"; "t"
  }
$$

Viewed this way, there is an evident definition of an automorphism of a prequantization
$([N^\infty_\Sigma \Sigma_p,\mathcal{E}]_\Sigma,\mathbf{\theta})$ of $([N^\infty_\Sigma\Sigma_p,\mathcal{E}]_\Sigma,\omega)$. This must be a diagram
of the following form
$$
  \raisebox{20pt}{
  \xymatrix{
    [N^\infty_\Sigma\Sigma_p, \mathcal{E}]_\Sigma
    \ar[rr]^-\phi|-\simeq_-{\ }="s"
    \ar[dr]_{\mathbf{\theta}}^{\ }="t"
    &&
    [N^\infty_\Sigma\Sigma_p,\mathcal{E}]_\Sigma
    \ar[dl]^{\mathbf{\theta}}
    \\
    & \mathbf{B}(\mathbb{R}/_{\!\hbar}\mathbb{Z})_{\mathrm{conn}}
    \ar@{=>}^{\eta} "s"; "t"
  }
  }
$$
hence a diffeomorphism $\phi$ together with a homotopy $\eta$ that relates the modulating morphism of the translated
prequantum bundle back to the original prequantum bundle. By the discussion in section \ref{globalactionfunctionalinintroduction}
such homotopies are equivalently coboundaries between the Cech-Deligne cocycles that correspond to the maps that
the homotopy goes between.
Here this means that the homotopy in the above diagram is an isomorphism
$\eta : \phi^\ast \mathbf{\theta} \stackrel{\simeq}{\longrightarrow} \mathbf{\theta}$
of circle bundles with connection.
These pairs $(\phi,\eta)$ are what Souriau called the \emph{quantomorphisms}. Via their canonical action on the space
of section of the prequantum bundle, these become the quantum operators.

To see what this is in local data, consider the special case that $\theta$ is a globally defined 1-form
and suppose that $\phi = \exp(t v)$ is the flow of a vector field $v$ .
$$
  \raisebox{20pt}{
  \xymatrix{
    [N^\infty_\Sigma\Sigma_p,\mathcal{E}]_\Sigma \ar[dr]_{\mathbf{\theta}}^{\ }="t" \ar[rr]^{\exp(t v)}|\simeq_{\ }="s" &&
    [N^\infty_\Sigma\Sigma_p,\mathcal{E}]_\Sigma \ar[dl]^{\mathbf{\theta}}
    \\
    & \mathbf{B}(\mathbb{R}/_{\!\hbar}\mathbb{Z})_{\mathrm{conn}}
    \ar@{=>}^{\exp(\tfrac{i}{\hbar} t \alpha)} "s"; "t"
  }
  }
  \,,
$$
Then the homotopy filling the
previous diagram is given by a smooth function $\exp(i t \alpha )$ such that
$$
  \exp(t v)^\ast \theta - \theta = t d \alpha
  \,.
$$
Infinitesimally, for $t \to 0$, this becomes
$$
  \mathcal{L}_v \theta = d \alpha
  \,.
$$
Using Cartan's formula for the Lie derivative on the left, and the fact that $d \theta = \omega$, by prequantization,
this is equivalent to
\begin{equation}
  \label{equationforHamiltonian}
  d \underset{H}{\underbrace{( \alpha - \iota_v \theta )}} = \iota_v \omega
  \,.
\end{equation}
This is the classical formula \cite{Arnold89} which says that
$$
  H := \alpha - \iota_v \theta
$$
is a Hamiltonian for the vector field $v$.

There is an evident smooth group structure on the homotopies as above, and one checks that the
induced Lie bracket on Hamiltonians $H$ with Hamiltonian vector fields $v$ is the following
$$
  [(v_1,H_1), (v_2,H_2)] = ( [v_1,v_2], \iota_{v_2}\iota_{v_1}\omega )
  \,.
$$
Traditionally this is considered only in the special case that $\omega$ is symplectic, hence
equivalently, in the case that equation (\ref{equationforHamiltonian}) uniquely associates a Hamiltonian vector field
$v$ with any Hamiltonian $H$. In that case we may identify a pair $(v,H)$ with just $H$ and then the above Lie bracket
becomes the \emph{Poisson bracket} on smooth functions induced by $\omega$.
Hence the Poisson bracket Lie algebra is secretly the infinitesimal symmetries of the prequantum line bundle $\mathbf{\theta}$.
This is noteworthy. For instance in the example of the phase space $(T^\ast \mathbb{R} = \mathbb{R}^2, \omega = d p \wedge dq)$
and writing $q,p : \mathbb{R}^2 \to \mathbb{R}$ for the two canonical coordinates ($p$ being called the ``canonical momentum''), then
the Poisson bracket, as above, between these two is
\begin{center}
\fbox{
  $[q,p] = i \hbar  \;\;\; \in \;\; i \hbar \mathbb{R} \hookrightarrow \mathfrak{Pois}(\mathbb{R}^2,d p \wedge d q)$
}
\end{center}
This equation is often regarded as the hallmark of quantum theory. In fact it is a
prequantum phenomenon. Notice how the identification of the central term with $i \hbar$ follows
here from the first prequantization step back around equation (\ref{thexponentialsequenceforprequantization}).

From equation (\ref{equationforHamiltonian}) it is clear that the Poisson bracket is a Lie extension of the Lie algebra
of (Hamiltonian) vector fields by the locally constant Hamiltonians, hence by constant functions
in the case that $X$ is connected. The non-trivial Lie integration of this statement is the
Kostant-Souriau extension, which says that the quantomorphism group of a connected phase space is
a $U(1)$-extension of the diffeological group of Hamiltonian symplectomorphisms.

Hence in summary the situation for observables on the covariant phase space in codimension 1 is
as follows:

\vspace{1cm}

\hspace{-3cm}
\begin{tabular}{|c|ccccc|}
\hline
& \raisebox{-4pt}{\begin{tabular}{c} {\bf Kostant-Souriau} \\ {\bf extension} \\ {\tiny (connected phase space)} \end{tabular}}
&& {\bf observables} && {\bf flows}
\\
\hline \hline
\raisebox{-29pt}{
\begin{turn}{90}\small $\;$infinitesimally$\;$ \end{turn}
}
&\raisebox{8pt}{$i \hbar \mathbb{R}$} & \raisebox{8pt}{$\longrightarrow$} &
\raisebox{0pt}{\begin{tabular}{c}$\mathfrak{Pois}([N^\infty_\Sigma \Sigma_p,\mathcal{E}]_\Sigma,\;\omega)$ \vspace{.2cm}  \\
\fbox{Poisson bracket}\end{tabular}} & \raisebox{8pt}{$\longrightarrow$} & \raisebox{8pt}{$\mathrm{Vect}(X)$}
\\
\hline
\raisebox{-37pt}{
\begin{turn}{90}\small \hspace{.3cm} finitely \hspace{.3cm}\end{turn}
}
& \raisebox{-9pt}{$U(1)$} & \raisebox{-9pt}{$\longrightarrow$} &
\raisebox{-17pt}{
\begin{tabular}{c} $\mathbf{QuantMorph}([N^\infty_\Sigma\Sigma_p,\mathcal{E}]_\Sigma,\;\mathbf{\theta})$ \vspace{.2cm}\\ \fbox{quantomorphism group} \end{tabular}} & \raisebox{-9pt}{$\longrightarrow$}
&
\raisebox{-9pt}{$\mathbf{Diff}(X)$}
\\
\hline
\raisebox{-17pt}{
\begin{turn}{90}abstractly\end{turn}}
&
$
  \left\{
  \raisebox{29pt}{
  \xymatrix@R=39pt{
    [N^\infty_\Sigma \Sigma_p,\mathcal{E}]_\Sigma
      \ar@/^2pc/[d]^{\mathbf{\theta}}_<<<{\ }="s"
      \ar@/_2pc/[d]_{\mathbf{\theta}}^>>>{\ }="t"
    \\
    \mathbf{B}(\mathbb{R}/_{\!\hbar}\mathbb{Z})_{\mathrm{conn}}
    \ar@{=>}|<<<{\mbox{\tiny \begin{tabular}{c} locally \\ constant \\ Hamiltonian\end{tabular}}} "s"; "t"
  }
  }
  \right\}
$
&
$\longrightarrow$
&
$
  \left\{
  \raisebox{52pt}{
  \xymatrix@R=39pt{
    [N^\infty_\Sigma \Sigma_p,\mathcal{E}]_\Sigma
    \ar@/_1pc/[rdd]_\omega
    \ar[dr]_{\theta}^-{\ }="t" \ar[rr]|\simeq^{\mathrm{flow}}_>>>>>>>>>>>>>>{\ }="s"
    &&
    [N^\infty_\Sigma \Sigma_p,\mathcal{E}]_\Sigma \ar[dl]^{\theta}
    \ar@/^1pc/[ddl]^{\omega}
    \\
    & \mathbf{B}(\mathbb{R}/_{\!\hbar}\mathbb{Z})_{\mathrm{conn}}
     \ar[d]^{\mathrm{curv}}
    \\
    & \mathbf{\Omega}^2_{\mathrm{cl}}
    \ar@{=>}|<<<<<<{\mathrm{Hamiltonian}} "s"; "t"
  }
  }
  \right\}
$
&
$\longrightarrow$
&
$
  \left\{
  \xymatrix{
    [N^\infty_\Sigma \Sigma_p,\mathcal{E}] \ar[r]|\simeq & [N^\infty_\Sigma \Sigma_p,\mathcal{E}]_\Sigma
  }
  \right\}
$
\\
\hline
\end{tabular}

\vspace{.3cm}

Generally. the symmetries of a $p$-gerbe connection $\nabla : X \to \mathbf{B}^{p+1}(\mathbb{R}/_{\!\hbar}\mathbb{Z})_{\mathrm{conn}}$
form an extension of the
symmetry group of the underlying space by the higher group of flat $(p-1)$-gerbe connections \cite{hgp}:

\vspace{.2cm}

\hspace{-1cm}
\begin{tabular}{|c|ccccc|}
\hline
&
\begin{tabular}{c}{\bf higher} \\ {\bf extension} \end{tabular}
&&
\begin{tabular}{c}{\bf symmetry of} \\ {\bf $p$-gerbe} \\ {\bf $(p+1)$-connection } \end{tabular}
&&
\begin{tabular}{c} {\bf automorphisms} \\ {\bf of base space} \end{tabular}
\\
\hline
\hline
\raisebox{-31pt}{
\begin{turn}{90}\hspace{.1cm} infinitesimally \hspace{.1cm}\end{turn}}
& $\mathrm{Ch}_{\mathrm{dR},\mathrm{cl}}^{\bullet \leq p}(X)$
& $\longrightarrow$ &
 \begin{tabular}{c} $\mathfrak{sym}_X(F)$ \vspace{.2cm}\\ \fbox{stablizer $L_\infty$-algebra} \end{tabular}
 &$\longrightarrow$&
 $\mathrm{Vect}(X)$
\\
\hline
\raisebox{-25pt}{
\begin{turn}{90}\hspace{.3cm} finitely \hspace{.6cm}\end{turn}
}
&
$\mathbf{Ch}^{\bullet\leq p}_{\mathrm{cl}}(X,U(1))$
&$\longrightarrow$&
\raisebox{-7pt}{\begin{tabular}{c} $\mathbf{Stab}_{\mathbf{Aut}(X)}(\nabla)$ \vspace{.2cm}\\ \fbox{stabilizer $\infty$-group}\end{tabular}}
&$\longrightarrow$&
$\mathbf{Aut}(X)$
\\
\hline
\raisebox{-20pt}{
\begin{turn}{90}abstractly\end{turn}
}
&
$
  \left\{
  \raisebox{29pt}{
  \xymatrix@R=39pt{
    X
      \ar@/^1.7pc/[d]^{\nabla}_<<<{\ }="s"
      \ar@/_1.7pc/[d]_{\nabla}^>>>{\ }="t"
    \\
    \mathbf{B}^{p+1}(\mathbb{R}/_{\!\hbar}\mathbb{Z})_{\mathrm{conn}}
    \ar@{=>}^\simeq "s"; "t"
  }
  }
  \right\}
$
&
$\longrightarrow$
&
$
  \left\{
  \raisebox{52pt}{
  \xymatrix@R=39pt{
    X
    \ar@/_1pc/[rdd]_{F}
    \ar[dr]_{\nabla}^>>>{\ }="t" \ar[rr]|-\simeq^-{\mathrm{automorphism}}_>>>>>>>>>>>>>>{\ }="s"
    && X \ar[dl]^{\nabla}
    \ar@/^1pc/[ddl]^{F}
    \\
    & \mathbf{B}^{p+1}(\mathbb{R}/_{\!\hbar}\mathbb{Z})_{\mathrm{conn}}
    \ar[d]^{\mathrm{curv}}
    \\
    & \mathbf{\Omega}^{p+2}_{\mathrm{cl}}
    \ar@{=>}|<<<<<<{\mathrm{homotopy} \atop \mathrm{stabilization}} "s"; "t"
  }
  }
  \right\}
$
&
$\longrightarrow$
&
$
  \left\{
  \xymatrix{
    V \ar[rr]|\simeq^-{\mathrm{}} && V
  }
  \right\}
$
\\
\hline
\end{tabular}

\newpage

Specifying this general phenomenon to the Lepage $p$-gerbes, it gives
a Poisson bracket $L_\infty$-algebra on higher currents (local observables)
\cite{LocalObservables} and its higher Lie integration to a higher quantomorphism
group constituting a higher Kostant-Souriau extension of the differential automorphisms of the field
bundle. This is determined by the (pre-)symplectic current $p+2$-form
$\Omega$ in analogy to how the ordinary Poisson bracket is determined by the
(pre-)symplectic 2-form $\omega$, hence this is a Poisson $L_\infty$-bracket for what has
been called ``multisymplectic geometry'' (see \cite{Rogers:2010nw}):

\medskip

\hspace{-1cm}
\begin{tabular}{|c|ccccc|}
\hline
&
\begin{tabular}{c}{\bf higher} \\ {\bf Kostant-Souriau}  \\ {\bf extension} \end{tabular}
&&
\begin{tabular}{c}{\bf symmetry of} \\ {\bf Lepage $p$-gerbe}  \end{tabular}
&&
\begin{tabular}{c} {\bf differential} \\ {\bf automorphisms} \\ {\bf of dynamical shell} \end{tabular}
\\
\hline
\hline
\raisebox{-31pt}{
\begin{turn}{90}\hspace{.1cm} infinitesimally \hspace{.1cm}\end{turn}}
& $\mathrm{Ch}_{\mathrm{dR},\mathrm{cl}}^{\bullet \leq p}(\mathcal{E})$
& $\longrightarrow$ &
 \begin{tabular}{c} $\mathfrak{Pois}(\mathcal{E},\;\Omega)$ \vspace{.2cm}\\ \fbox{Poisson bracket $L_\infty$-algebra} \end{tabular}
 &$\longrightarrow$&
 $\mathrm{Vect}(\mathcal{E})$
\\
\hline
\raisebox{-25pt}{
\begin{turn}{90}\hspace{.3cm} finitely \hspace{.6cm}\end{turn}
}
&
$\mathbf{Ch}^{\bullet\leq p}_{\mathrm{cl}}(\mathcal{E},U(1))$
&$\longrightarrow$&
\raisebox{-7pt}{\begin{tabular}{c} $\mathbf{Stab}_{\mathbf{Aut}(\mathcal{E})}(\mathbf{\Theta})$ \vspace{.2cm}\\
\fbox{quantomorphism $\infty$-group}\end{tabular}}
&$\longrightarrow$&
$\mathbf{Aut}(\mathcal{E})$
\\
\hline
\raisebox{-20pt}{
\begin{turn}{90}abstractly\end{turn}
}
&
$
  \left\{
  \raisebox{29pt}{
  \xymatrix@R=39pt{
    \mathcal{E}
      \ar@/^1.9pc/[d]^{\mathbf{\Theta}}_<<<{\ }="s"
      \ar@/_1.9pc/[d]_{\mathbf{\Theta}}^>>>{\ }="t"
    \\
    \mathbf{B}^{p+1}_L(\mathbb{R}/_{\!\hbar}\mathbb{Z})_{\mathrm{conn}}
    \ar@{=>}|<<<{\mathrm{topological} \atop \mathrm{Hamiltonian}} "s"; "t"
  }
  }
  \right\}
$
&
$\longrightarrow$
&
$
  \left\{
  \raisebox{52pt}{
  \xymatrix@R=39pt{
    \mathcal{E}
    \ar@/_1pc/[rdd]_{\Omega}
    \ar[dr]_{\mathbf{\Theta}}^>>>{\ }="t" \ar[rr]|-\simeq^-{\mbox{\tiny on-shell symmetry}}_>>>>>>>>>>>>>>{\ }="s"
    &&
    \mathcal{E} \ar[dl]^{\mathbf{\Theta}}
    \ar@/^1pc/[ddl]^{\Omega}
    \\
    & \mathbf{B}^{p+1}_L(\mathbb{R}/_{\!\hbar}\mathbb{Z})_{\mathrm{conn}}
    \ar[d]^{\mathrm{curv}}
    \\
    & \mathbf{\Omega}^{p,2}_{\mathrm{cl}}
    \ar@{=>}|<<<<<<{\mathrm{Hamiltonian} \atop \mathrm{current}} "s"; "t"
  }
  }
  \right\}
$
&
$\longrightarrow$
&
$
  \left\{
  \xymatrix{
    \mathcal{E} \ar[rr]|\simeq^-{\mbox{\tiny on-shell symmetry}} && \mathcal{E}
  }
  \right\}
$
\\
\hline
\end{tabular}

\medskip

So far this concerned the covariant phase space with its prequantization via the Lepage $p$-gerbe.
In the same way there are the higher symmetries of the field space with its prequantization via the Euler-Lagrange $p$-gerbes $\mathbf{L}$
$$
  \raisebox{20pt}{
  \xymatrix{
    E
    \ar[dr]_{\mathbf{L}}^{\ }="t"
    \ar[rr]^-\phi|\simeq_{\ }="s"
    && E
    \ar[dl]^{\mathbf{L}}
    \\
    & \mathbf{B}^{p+1}(\mathbb{R}/_{\!\hbar}\mathbb{Z})_{\mathrm{conn}}
    \ar@{=>}^\simeq "s"; "t"
  }
  }\,.
$$
To see what these are in components, consider again the special case
that $\mathbf{L}$ is given by a globally defined horizontal form, and consider a one-parameter flow of such symmetries
$$
  \raisebox{20pt}{
  \xymatrix{
    E
    \ar[dr]_{\mathbf{L}}^{\ }="t"
    \ar[rr]^-{\exp(t v)}|\simeq_{\ }="s"
    && E
    \ar[dl]^{\mathbf{L}}
    \\
    & \mathbf{B}^{p+1}_H(\mathbb{R}/_{\!\hbar}\mathbb{Z})_{\mathrm{conn}}
    \ar@{=>}^{\exp(\tfrac{i}{\hbar} t \Delta)} "s"; "t"
  }
  }\,.
$$
In Cech-Deligne cochain components this diagram equivalently exhibits the equation
$$
  \exp(t v)^\ast L - L = t \, d_H \Delta
$$
on differential forms on the jet bundle of $E$, where $v$ is a vertical vector field.
Infinitesimally for $t\to 0$ this becomes
$$
  \mathcal{L}_v L = d_H \Delta
  \,.
$$
Since $L$ is horizontal while $v$ is vertical, the left hand reduces,
by equation \ref{differentialofLagrangian}, to
$$
  \iota_v d L = \iota_v( \mathrm{EL} - d_H\Theta )
  \,,
$$
Therefore the infinitesimal symmetry of $\mathbf{L}$ is equivalent to
$$
  d_H \underbrace{( \Delta - \iota_v \Theta )}_{J} = \iota_v \mathrm{EL}
  \,.
$$
This says that associated to the symmetry $v$ is a current
$$
  J := \Delta - \iota_v \Theta
$$
which is conserved (horizontally closed) on shell (on the vanishing locus $\mathcal{E}$ of the Euler-Lagrange form $\mathrm{EL}$).
This is precisely the statement of Noether's theorem (the first variational theorem of Noether, to be precise).
Indeed, in its modern incarnation (see \cite[section 3]{BarnichHenneaux}), Noether's theorem is understood as stating a Lie algebra extension
of the Lie algebra of symmetries by topological currents to the Lie-Dickey bracket on equivalence classes of conserved currents
.

Hence the $\infty$-group extension of symmetries of the Euler-Lagrange $p$-gerbe
promotes Noether's theorem to the statement that higher Noether currents form an
$L_\infty$-algebra extension of the infinitesimal symmetries by topological currents:

\vspace{.4cm}

\hspace{-1cm}
\begin{tabular}{|c|ccccc|}
\hline
&
\begin{tabular}{c}{\bf higher} \\ {\bf topological charge}  \\ {\bf extension} \end{tabular}
&&
\begin{tabular}{c}{\bf symmetry of} \\ {\bf Euler-Lagrange $p$-gerbe}  \end{tabular}
&&
\begin{tabular}{c} {\bf differential} \\ {\bf automorphisms} \\ {\bf of field bundle} \end{tabular}
\\
\hline
\hline
\raisebox{-31pt}{
\begin{turn}{90}\hspace{.1cm} infinitesimally \hspace{.1cm}\end{turn}}
& $\mathrm{Ch}_{\mathrm{dR},\mathrm{cl}}^{\bullet \leq p}(E)$
& $\longrightarrow$ &
 \begin{tabular}{c} $\mathfrak{curr}(E,\;\mathrm{EL})$ \vspace{.2cm}\\ \fbox{Dickey bracket current $L_\infty$-algebra} \end{tabular}
 &$\longrightarrow$&
 $\mathrm{Vect}(E)$
\\
\hline
\raisebox{-25pt}{
\begin{turn}{90}\hspace{.3cm} finitely \hspace{.6cm}\end{turn}
}
&
$\mathbf{Ch}^{\bullet\leq p}_{\mathrm{cl}}(E,U(1))$
&$\longrightarrow$&
\raisebox{-7pt}{\begin{tabular}{c} $\mathbf{Stab}_{\mathbf{Aut}(E}(\mathbf{L})$ \vspace{.2cm}\\
\fbox{de-transgressed Kac-Moody $\infty$-group}\end{tabular}}
&$\longrightarrow$&
$\mathbf{Aut}(E)$
\\
\hline
\raisebox{-20pt}{
\begin{turn}{90}abstractly\end{turn}
}
&
$
  \left\{
  \raisebox{29pt}{
  \xymatrix@R=39pt{
    E
      \ar@/^1.9pc/[d]^{\mathbf{L}}_<<<{\ }="s"
      \ar@/_1.9pc/[d]_{\mathbf{L}}^>>>{\ }="t"
    \\
    \mathbf{B}^{p+1}_H(\mathbb{R}/_{\!\hbar}\mathbb{Z})_{\mathrm{conn}}
    \ar@{=>}|<<<{\mathrm{topological} \atop \mathrm{current}} "s"; "t"
  }
  }
  \right\}
$
&
$\longrightarrow$
&
$
  \left\{
  \raisebox{52pt}{
  \xymatrix@R=39pt{
    E
    \ar@/_1pc/[rdd]_{\mathrm{EL}}
    \ar[dr]_{\mathbf{L}}^>>>{\ }="t" \ar[rr]|-\simeq^-{\mbox{\tiny variational symmetry}}_>>>>>>>>>>>>>>{\ }="s"
    &&
    E \ar[dl]^{\mathbf{L}}
    \ar@/^1pc/[ddl]^{\mathrm{EL}}
    \\
    & \mathbf{B}^{p+1}_H(\mathbb{R}/_{\!\hbar}\mathbb{Z})_{\mathrm{conn}}
    \ar[d]^{\mathrm{curv}}
    \\
    & \mathbf{\Omega}^{p+1,1}_{S,\mathrm{cl}}
    \ar@{=>}|<<<<<<{\mathrm{Noether} \atop \mathrm{current}} "s"; "t"
  }
  }
  \right\}
$
&
$\longrightarrow$
&
$
  \left\{
  \xymatrix{
    E \ar[rr]|\simeq^{\mathrm{symmetry}} && E
  }
  \right\}
$
\\
\hline
\end{tabular}

\medskip

In summary, physical local observables arise from symmetries of higher
prequantum geometry as follows.

\medskip

\hspace{-1.3cm}
\begin{tabular}{|c|c|c|c|c|}
  \hline
  \begin{tabular}{c}
    {\bf prequantum}
    \\
    {\bf geometry}
  \end{tabular}
  &
  \begin{tabular}{c}
    {\bf automorphism}
    \\
    up to
    \\
    {\bf homotopy}
  \end{tabular}
  &
  \begin{tabular}{c}
    {\bf Lie derivative}
    \\
    up to
    \\
    {\bf differential}
  \end{tabular}
  &
  {\bf equivalently}
  &
  {\bf physical quantity}
  \\
  \hline
  \hline
  \begin{tabular}{c}
    prequantum
    \\
    shell
  \end{tabular}
  &
  $
  \raisebox{30pt}{
  \xymatrix{
    \mathcal{E}
    \ar@/_1pc/[rdd]_{\Omega}
    \ar[rr]^-{\exp(t v)}_{\ }="s"
    \ar[dr]_{\mathbf{\Theta}}^{\ }="t"
    &&
    \mathcal{E}
    \ar@/^1pc/[ddl]^{\Omega}
    \ar[dl]^{\mathbf{\Theta}}
    \\
    & \mathbf{B}^{p+1}_L(\mathbb{R}/_{\!\hbar}\mathbb{Z})
    \ar[d]^{\mathrm{curv}}
    \\
    & \mathbf{\Omega}^{p,2}_{\mathrm{cl}}
    \ar@{=>}^{\exp(\tfrac{i}{\hbar}t \alpha)} "s"; "t"
  }}
  $
  & $\mathcal{L}_v \Theta = d \alpha$ & $d \underset{H}{\underbrace{( \alpha - \iota_v \Theta )}} = \iota_v \Omega$
  &
  Hamiltonian
  \\
  \hline
  \begin{tabular}{c}
    prequantum
    \\
    field bundle
  \end{tabular}
  &
  $
  \raisebox{30pt}{
  \xymatrix{
    E
    \ar@/_1pc/[rdd]_{\mathrm{EL}}
    \ar[rr]^{\exp(t v)}_{\ }="s"
    \ar[dr]_{\mathbf{L}}^{\ }="t"
    &&
    E
    \ar@/^1pc/[ldd]^{\mathrm{EL}}
    \ar[dl]^{\mathbf{L}}
    \\
    & \mathbf{B}^{p+1}_H(\mathbb{R}/_{\!\hbar} \mathbb{Z})
    \ar[d]^{\mathrm{curv}}
    \\
    &
    \mathbf{\Omega}^{p+1,1}_{S,\mathrm{cl}}
    \ar@{=>}^{\exp(\tfrac{i}{\hbar} t \Delta)} "s"; "t"
  }
  }
  $
  & $\mathcal{L}_v L = d_H \Delta$ & $d_H \underset{J}{\underbrace{( \Delta - \iota_v \Theta )}} = \iota_v \mathrm{EL}$
  &
  conserved current
  \\
  \hline
\end{tabular}

%%%%%%%%%%%%%%%%%%%%%%%%%%%%%%%%%%%%%%%%%%%%%%%%%%%%%%%%
\subsection{The evolution -- correspondingly}
%%%%%%%%%%%%%%%%%%%%%%%%%%%%%%%%%%%%%%%%%%%%%%%%%%%%%%%%

The transgression formula discussed in section \ref{elgerbesinintroduction} generalizes to compact oriented $d$-manifolds $\Sigma$,
possibly with boundary $\partial \Sigma \hookrightarrow \Sigma$. Here it becomes transgression \emph{relative} to the boundary
transgression.

For curvature forms this is again classical: For $\omega \in \Omega^{p+2}(\Sigma \times U)$
a closed differential form,
then $\int_\Sigma \omega \in \Omega^{p+2-d}(U)$ is not in general a closed differential form anymore,
but by Stokes' theorem its differential equals the boundary transgression:
$$
  \begin{aligned}
    d_U \int_\Sigma \omega & = \int_\Sigma d_U \omega
    \\
    & = -\int_\Sigma d_\Sigma \omega
    \\
    &= - \int_{\partial \Sigma} \omega
    \,.
  \end{aligned}
$$
This computation also shows that a sufficient condition for the bulk transgression of $\omega$ to be closed and
for the boundary transgression to vanish is that $\omega$ be also \emph{horizontally} closed, i.e. closed with
respect to $d_\Sigma$.

Applied to the construction of the canonical presymplectic structure on phase spaces in \ref{elgerbesinintroduction}
this has the important implication that the canonical presymplectic form on phase space is indeed canonical.

Namely, by equation (\ref{differentialofLagrangian}),
the presymplectic current $\Omega \in \Omega^{p,2}(E)$ is horizontally closed on shell, hence is indeed a conserved current:
$$
  \begin{aligned}
    d_H \Omega & = d_H d_V \Theta
    \\ & = - d_V d_H \Theta
    \\ & = - d_V (- d_V L + \mathrm{EL})
    \\ & = -d_V \mathrm{EL} \,.
  \end{aligned}
$$
It follows that if $\Sigma$ is a spacetime/worldcolume with, say, two boundary components
$\partial \Sigma = \partial_{\mathrm{in}} \Sigma \sqcup \partial_{\mathrm{out}}\Sigma$, then the presymplectic structures
$\omega_{\mathrm{in}} := \int_{\partial_{\mathrm{in}}\Sigma } [\partial_{\mathrm{in}}\Sigma,\omega ]$
and
$\omega_{\mathrm{out}} := \int_{\partial_{\mathrm{out} \Sigma}\Sigma} [\partial_{\mathrm{out}}\Sigma,\omega ]$
agree on the covariant phase space:\footnote{If one uses a BV-resolution of the covariant phase space, then
they agree up to the BVV-differential of a BV (-1)-shifted 2-form, we come back to this in section \ref{introductionbvcomplex}.}
$$
  \xymatrix{
    & [\Sigma, \mathcal{E}]_{\Sigma}
    \ar[dl]_{(-)|_{\partial_{\mathrm{in}}\Sigma}}
    \ar[dr]^{(-)|_{\partial_{\mathrm{out}}\Sigma }}_{\ }="s"
    \\
    [N^\infty_\Sigma \partial_{\mathrm{in}}\Sigma, \mathcal{E}]_{\Sigma}
    \ar[dr]_{\omega_{\mathrm{in}}}^{\ }="t"
    &&
    [N^\infty_\Sigma\partial_{\mathrm{out}}\Sigma, \mathcal{E}]_{\Sigma}
    \ar[dl]^{\omega_{\mathrm{out}}}
    \\
    & \mathbf{\Omega}^{2}_{\mathrm{cl}}
    \ar@{=} "s"; "t"
  }
  \,.
$$
This diagram may be thought of as expressing an \emph{isotropic correspondence} between the two phase spaces,
where $[\Sigma,\mathcal{E}]_{\Sigma}$ is isotropic in the product of the two boundary phase spaces, regarded as
equipped with the presymplectic form $\omega_{\mathrm{out}} -\omega_{\mathrm{in}}$. In particular, when both
$\partial_{\mathrm{in}}\Sigma$ and $\partial_{\mathrm{out}}\Sigma$ are Cauchy surfaces in $\Sigma$, so that the
two boundary restriction maps in the above diagram are in fact equivalences, then this is a \emph{Lagrangian correspondence}
in the sense of \cite{Weinstein71}\cite{Weinstein83}.

All this needs to have and does have prequantization:
The transgression of a $p$-gerbe
$\nabla : X \to \mathbf{B}^{p+1}(\mathbb{R}/_{\!\hbar}\mathbb{Z})_{\mathrm{conn}}$ to the bulk of a $d$-dimensional $\Sigma$ is no longer
quite a $p-d$-gerbe itself, but is a section of the pullback of the $p-d+1$-gerbe that is the transgression to the boundary
$\partial \Sigma$. Diagrammatically this means that transgression to maps out of $\Sigma$ is a homotopy filling a diagram
of the following form
$$
  \xymatrix{
    [\Sigma,X] \ar[rr]^-{[\Sigma,\nabla]}
    \ar[dd]_{(-)|_{\partial\Sigma}}
    &&
    [\Sigma, \mathbf{B}^{p+1}(\mathbb{R}/_{\!\hbar}\mathbb{Z})_{\mathrm{conn}}]
    \ar[rr]_>{\ }="s"^{\int_\Sigma \mathrm{curv}}
    \ar[dd]|{(-)|_{\partial \Sigma}}^>{\ }="t"
    &&
    \mathbf{\Omega}^{p+2-d}
    \ar[dd]
    \\
    \\
    [\partial \Sigma, X]
    \ar[rr]_-{[\partial \Sigma,\nabla]}
    &&
    [\partial\Sigma, \mathbf{B}^{p+1}(\mathbb{R}/_{\!\hbar}\mathbb{Z})_{\mathrm{conn}}]
    \ar[rr]_-{\int_{\partial \Sigma}}
    &&
    \mathbf{B}^{p+2-d}(\mathbb{R}/_{\!\hbar}\mathbb{Z})_{\mathrm{conn}}
    \ar@{=>}^{\int_\Sigma} "s"; "t"
  }
  \,.
$$
Here the appearance of the differential forms coefficients $\mathbf{\Omega}^{p+2-d}$ in the top right corner witnesses the fact that the bulk term
$\int_\Sigma [\Sigma,\nabla]$ is a trivialization of the pullback of the boundary gerbe $\int_{\partial \Sigma}[\partial \Sigma, \nabla]$
only as a plain gerbe, not necessarily as a gerbe with connection: in general the curvature of the pullback of
$\int_{\partial \Sigma}[\partial \Sigma, \nabla]$ will not vanish, but only be exact, as in the above discussion, and the form that it is the de Rham
differential of is expressed by the top horizontal morphism in the above diagram.

Hence in the particular case of the transgression of a Lepage $p$-gerbe to covariant phase space, this formula
yields a prequantization of the above Lagrangian correspondence, where now the globally defined action functional
$$
  \exp(\tfrac{i}{\hbar}S) = \int_\Sigma [\Sigma, \mathbf{\Theta}] = \int_\Sigma [\Sigma, \mathbf{L}]
$$
exhibits the the equivalence between the incoming and outgoing prequantum bundles
$$
  \mathbf{\theta}_{\mathrm{in/out}} = \int_{\partial_{\mathrm{in/out}}\Sigma} [\partial_{\mathrm{in/out}}\Sigma, \mathbf{\theta}]
$$
on covariant phase space:
$$
 \hspace{-1cm}
 \raisebox{20pt}
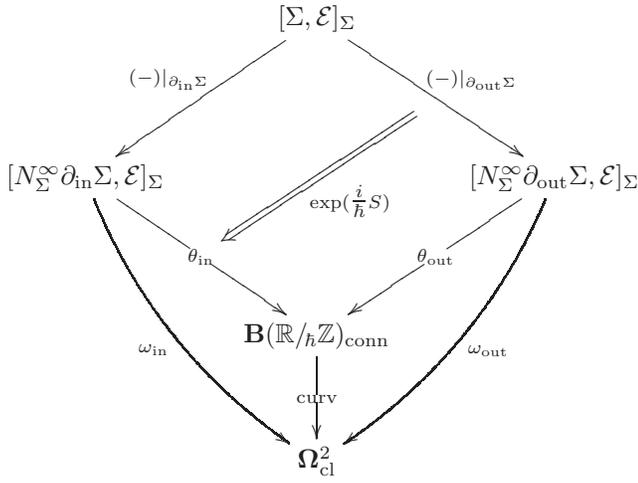
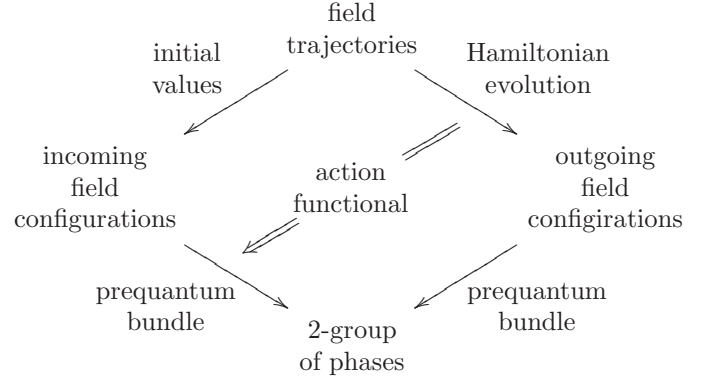
{
  \xymatrix{
    & [\Sigma, \mathcal{E}]_{\Sigma}
    \ar[dl]_{(-)|_{\partial_{\mathrm{in}}\Sigma}}
    \ar[dr]^{(-)|_{\partial_{\mathrm{out}}\Sigma }}_{\ }="s"
    &&&
    \mbox{\begin{tabular}{c}field \\ trajectories\end{tabular}}
    \ar[dl]_{\mbox{\begin{tabular}{c}initial \\ values\end{tabular}}}
    \ar[dr]^{\mbox{\begin{tabular}{c}Hamiltonian \\ evolution\end{tabular}}}_{\ }="s1"
    \\
    [N^\infty_\Sigma \partial_{\mathrm{in}}\Sigma, \mathcal{E}]_{\Sigma}
    \ar[dr]|{\mathbf{\theta}_{\mathrm{in}}}^{\ }="t"
    \ar@/_1pc/[ddr]_{\omega_{\mathrm{in}}}
    &&
    [N^\infty_\Sigma\partial_{\mathrm{out}}\Sigma, \mathcal{E}]_{\Sigma}
    \ar[dl]|{\mathrm{\theta}_{\mathrm{out}}}
    \ar@/^1pc/[ddl]^{\omega_{\mathrm{out}}}
    &
    \mbox{\begin{tabular}{c} incoming \\ field \\configurations \end{tabular}}
    \ar[dr]_{\mbox{\begin{tabular}{c}prequantum \\ bundle\end{tabular}}}^{\ }="t1"
    &&
    \mbox{\begin{tabular}{c} outgoing \\ field \\ configirations \end{tabular}}
    \ar[dl]^{\mbox{\begin{tabular}{c}prequantum \\ bundle\end{tabular}}}
    \\
    & \mathbf{B}(\mathbb{R}/_{\!\hbar}\mathbb{Z})_{\mathrm{conn}}
    \ar[d]|{\mathrm{curv}}
    &&&
    \mbox{\begin{tabular}{c}2-group \\ of phases \end{tabular}}
    \\
    & \mathbf{\Omega}^{2}_{\mathrm{cl}}
    \ar@{=>}|{\mbox{\begin{tabular}{c}action \\ functional\end{tabular}}} "s1"; "t1"
    \ar@{=>}^{\exp(\tfrac{i}{\hbar}S)} "s"; "t"
  }
  }
  \,.
$$
This \emph{prequantized Lagrangian correspondence} hence reflects the prequantum evolution from
fields on the incoming piece $\partial_{\mathrm{in}}\Sigma$ of spacetime/worldvolume
to the outgoing piece $\partial_{\mathrm{out}}\Sigma$ via trajectories of field configurations
along $\Sigma$.

%%%%%%%%%%%%%%%%%%%%%%%%%%%%%%%%%%%%%%%%%%%%%%%%%%%%%%%%
\section{Examples of prequantum field theory}
\label{examplesinmotivation}
%%%%%%%%%%%%%%%%%%%%%%%%%%%%%%%%%%%%%%%%%%%%%%%%%%%%%%%%

We now survey classes of examples of prequantum field theory in the sense of section \ref{PrequantumLocalFieldTheoryInMotivation}, 
see also the survey in \cite{PQFTFromShiftedSymplectic}.

%%%%%%%%%%%%%%%%%%%%%%%%%%%%%%%%%%%%%%%%%%%
\subsection{Gauge fields}
\label{introductiongaugefields}
%%%%%%%%%%%%%%%%%%%%%%%%%%%%%%%%%%%%%%%%%%

Modern physics rests on two fundamental principles. One is the \emph{locality principle};
its mathematical incarnation in terms of differential cocycles on PDEs was the content
of section \ref{PrequantumLocalFieldTheoryInMotivation}. The other is the \emph{gauge principle}.

In generality, the gauge principle says that given any two field configurations $\phi_1$ and $\phi_2$
-- and everything in nature is some field cofiguration --
then it is physically meaningless to ask whether they are \emph{equal}, instead one has to ask whether they are
\emph{equivalent} via a \emph{gauge transformation}
$$
  \xymatrix{
    \phi_1
    \ar@/^1pc/[rr]^-\simeq_-{\mathrm{gauge} \atop \mathrm{equivalence}}
    &&
    \phi_2
  }
  \,.
$$
There may be more than one gauge transformation between two field configurations, and hence there may
be auto-gauge equivalences that non-trivially re-identify a field configuration with itself.
Hence a space of physical field configurations does not really look like a set of points, it looks more
like this cartoon:
$$
  \left\{
  \raisebox{20pt}{
  \xymatrix{
    \phi_1
    \ar@/_1pc/[dr]_\simeq
    \ar@/^1pc/[dr]^\simeq
    &&
    \phi_3
    \ar[dr]^\simeq
    \ar@(ul,ur)^\simeq
    &
    \\
    & \phi_2
    && \phi_4
  }}
  \cdots
  \right\}
  \,.
$$
Moreover, if there are two
gauge transformations, it is again physically meaningless to ask whether they are equal, instead one
has to ask whether they are equivalent via a \emph{gauge-of-gauge} transformation.
$$
  
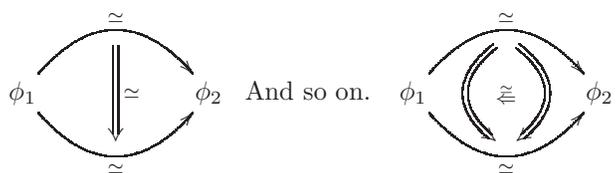
\xymatrix{
    \phi_1
    \ar@/^2pc/[rr]^\simeq_{\ }="s"
    \ar@/_2pc/[rr]_\simeq^{\ }="t"
    &&
    \phi_2
    \ar@{=>}^\simeq "s"; "t"
  }
  \;\;
\mbox{And so on.}
\:\;\;
  \xymatrix{
    \phi_1
    \ar@/^2pc/[rr]^\simeq_{\ }="s"
    \ar@/_2pc/[rr]_\simeq^{\ }="t"
    &&
    \phi_2
    \ar@{}|{\stackrel{\simeq}{\Leftarrow}} "s"; "t"
    \ar@/_1.3pc/@{=>} "s"; "t"
    \ar@/^1.3pc/@{=>} "s"; "t"
  }
$$

In this generality, the gauge principle of physics is the mathematical principle of
\emph{homotopy theory}: in general it is meaningless to say that some objects form a set
whose elements are either equal or not, instead
one has to consider the \emph{groupoid} which they form, whose morphisms are the equivalences
between these objects. Moreover, in general it is meaningless to assume that any two such
morphisms are equal or not, rather one has to consider the groupoid which these form,
which then in total makes a \emph{2-groupoid}. But in general it is also meaningless to ask
whether two equivalences of two equivalences are equal or not, and continuing this
way one finds that objects in general form an \emph{$\infty$-groupoid}, also called a
\emph{homotopy type}.

Of particular interest in physics are smooth gauge transformations that arise by integration of
\emph{infinitesimal} gauge transformations.  An infinitesimal smooth groupoid is a
\emph{Lie algebroid} and an infinitesimal smooth $\infty$-groupoid is an \emph{$L_\infty$-algebroid}.
The importance of infinitesimal symmetry transformations in physics,
together with the simple fact that they are easier to handle than finite transformations,
makes them appear more prominently in the physics literature. In particular, the physics literature
is secretly well familiar with smooth $\infty$-groupoids in their infinitesimal incarnation
as $L_\infty$-algebroids: these
are equivalently what in physics are called \emph{BRST complexes}.
What are called \emph{ghosts} in the BRST complex are the cotangents to the space of equivalences
between objects, and what are called higher order \emph{ghosts-of-ghosts} are cotangents to
spaces of higher order equivalences-of-equivalences. We indicate in a moment how to see this.

\medskip

While every species of fields in physics is subject to the gauge principle, one speaks
specifically of \emph{gauge fields} for those fields which are locally given by differential forms
$A$  with values in a Lie algebra (for ordinary gauge fields) or more generally with values
in an $L_\infty$-algebroid (for higher gauge fields).

\begin{center}
\begin{tabular}{|c|c|c|c|c|}
 \hline
 &
 \multicolumn{2}{|c|}{{\bf infinitesimal}}
 &
 \multicolumn{2}{|c|}{{\bf finite}}
 \\
 & {\bf by itself} & {\bf acting on fields} & {\bf by itself} & {\bf acting on fields}
 \\
 \hline
 \hline
 {\bf symmetry} & Lie algebra & Lie algebroid & Lie group & Lie groupoid
 \\
 \hline
 {\bf
 \begin{tabular}{c}
   symmetries
   \\
   of
   \\
   symmetries
 \end{tabular}
 }
 &
 Lie 2-algebra & Lie 2-algebroid & smooth 2-group & smooth 2-groupoid
 \\
 \hline
 {\bf
 \begin{tabular}{c}
   higher
   \\
   order
   \\
   symmetries
 \end{tabular}
 }
 &
 $L_\infty$-algebra
 &
 $L_\infty$-algebroid
 &
 smooth $\infty$-group
 &
 smooth $\infty$-groupoid
 \\
 \hline
 \hline
 {\bf
 \begin{tabular}{c} physics \\ terminology \end{tabular}} &
 \raisebox{-8pt}{\begin{tabular}{c} FDA \\ (e.g. \cite{CDF})\end{tabular}} &
  \raisebox{-8pt}{\begin{tabular}{c} BRST complex \\ (e.g. \cite{HenneauxTeitelboim})\end{tabular}} & gauge group & ---
 \\
 \hline
\end{tabular}
\end{center}

We now indicate how
such gauge fields and higher gauge fields come about.

\begin{itemize}
  \item \ref{introductionordinarygaugefields} -- Ordinary gauge fields;
  \item \ref{introductionhighergaugefields} -- Higher gauge fields;
  \item \ref{introductionbvcomplex} -- The BV-BRST complex.
\end{itemize}

%%%%%%%%%%%%%%%%%%%%%%%%%%%%%%%%%%%
\paragraph{Ordinary gauge fields.}
\label{introductionordinarygaugefields}
%%%%%%%%%%%%%%%%%%%%%%%%%%%%%%%%%%%

To start with, consider a plain group $G$.
For the standard applications mentioned in section \ref{TheNeedForPrequantumGeometry} we would take $G = U(1)$
or $G = \mathrm{SU}(n)$ or products of these, and then the gauge fields we are about to find would be those
of electromagnetism and of the nuclear forces, as they appear in the standard model of particle physics.

In order to highlight that we think of $G$ as a \emph{group of symmetries} acting on some (presently unspecified object)
$\ast$, we write
$$
  G
  \;=\;
  \left\{
    \xymatrix{
      \ast \ar@/^1pc/[rr]^{g}
      &&
      \ast
    }
  \right\}
  \,.
$$
In this vein, the product operation $(-)\cdot (-) : G \times G \to G$ in the group reflects the result of
applying two symmetry operations
$$
  G \times G
  \simeq
  \left\{
    \raisebox{20pt}{
    \xymatrix{
      & \ast \ar[dr]^{g_2}
      \\
      \ast
      \ar[ur]^{g_1}
      \ar[rr]_{g_1 \cdot g_2}
      &&
      \ast
    }
    }
  \right\}
  \,.
$$
Similarly, the associativity of the group product operation reflects the result of
applying three symmetry operations:
$$
  G \times G \times G
  \simeq
  \left\{
    \raisebox{36pt}{
    \xymatrix{
      \ast \ar[rr]^{g_2} && \ast \ar[dd]^{g_3}
      \\
      \\
      \ast
      \ar[uu]^{g_1}
      \ar[uurr]|{g_1 \cdot g_2}
      \ar[rr]_{(g_1 \cdot g_2) \cdot g_3 }
      &&
      \ast
   }
   }
   \;\;
   =
   \;\;
    \raisebox{36pt}{
    \xymatrix{
      \ast \ar[rr]^{g_2}
      \ar[ddrr]|{g_2 \cdot g_3}
      && \ast \ar[dd]^{g_3}
      \\
      \\
      \ast
      \ar[uu]^{g_1}
      \ar[rr]_{g_1 \cdot (g_2 \cdot g_3) }
      &&
      \ast
   }
   }
  \right\}
$$
Here the reader should think of the diagram on the right as a tetrahedron, hence a 3-simplex,
that has been cut open only for notational purposes.

Continuing in this way, $k$-tuples of symmetry transformations serve to label $k$-simplices
whose edges and faces reflect all the possible ways of consecutively applying the
corresonding symmetry operations. This forms a \emph{simplicial set},
called the \emph{simplicial nerve} of $G$, hence a system
$$
  \mathbf{B}G
  :
  k \mapsto G^{\times_k}
$$
of sets of $k$-simplices for all $k$, together with compatible maps between these that
restrict $k+1$-simplices to their $k$-faces (the face maps) and those that regard
$k$-simplices as degenerate $k+1$-simplices (the degeneracy maps).
From the above picture, the face maps of $\mathbf{B}G$ in low degree look as follows
(where $p_i$ denotes projection onto the $i$th factor in a Cartesian product):
$$
  \mathbf{B}G
  :=
  \left[
  \xymatrix{
    \ar@{..}[r]
    &
    G \times G \times G
    \ar@<+16pt>[rr]^-{(p_1,p_2)}
    \ar@<+6pt>[rr]|-{(\mathrm{id},(-)\cdot (-))}
    \ar@<-6pt>[rr]|-{((-)\cdot (-), \mathrm{id})}
    \ar@<-16pt>[rr]_-{(p_2,p_3)}
    &&
    G \times G
    \ar@<+10pt>[rr]^{p_1}
    \ar[rr]|{(-) \cdot (-)}
    \ar@<-10pt>[rr]_{p_2}
    &&
    G
    \ar@<+4pt>[rr]
    \ar@<-4pt>[rr]
    &&
    \ast
  }
  \right]
$$

%The diagrams above make manifest some of the relations that hold between
%the face maps in the simplicial nerve, notably the associativity of the group operation.
%These relations are called the \emph{simplicial identities}.

It is useful to remember the smooth structure on these spaces of $k$-fold symmetry operation by
remembering all possible ways of forming smoothly $U$-parameterized collections of $k$-fold symmetry
operations, for any abstract coordinate chart $U = \mathbb{R}^n$. Now a smoothly $U$-parameterized
collection of $k$-fold $G$-symmetries is simply a smooth function from $U$ to $G^{\times_k}$,
hence equivalently is $k$ smooth functions from $U$ to $G$. Hence the symmetry group $G$
together with its smooth structure is encoded in the system of assignments
$$
  \mathbf{B}G : (U,k) \mapsto C^\infty(U, G^{\times_k}) = C^\infty(U,G)^{\times_k}
$$
which is contravariantly functorial in abstract coordinate charts $U$ (with smooth functions between them)
and in abstract $k$-simplices (with cellular maps between them).
This is the incarnation of $\mathbf{B}G$ as a \emph{smooth simplicial presheaf}.

Another basic example of a smooth simplicial presheaf is the nerve of an open cover. Let
$\Sigma$ be a smooth manifold and let $\{U_i \hookrightarrow \Sigma\}_{i \in I}$ be a cover
of $\Sigma$ by coordinate charts $U_i \simeq \mathbb{R}^n$. Write
$U_{i_0 \cdots i_k} := U_{i_0} \underset{X}{\times} U_{i_1} \underset{X}{\times} \cdots \underset{X}{\times} U_{i_k}$
for the intersection of $(k+1)$ coordinate charts in $X$. These arrange into a simplicial object like so
$$
  C(\{U_i\})
  =
  \left[
   \xymatrix{
    \ar@{..}[r]
    &
    \underset{i_0, i_1, i_2}{\coprod} U_{i_0, i_1, i_2}
    \ar@<+6pt>[rr]
    \ar[rr]
    \ar@<-6pt>[rr]
    &&
    \underset{i_0, i_1 }{\coprod} U_{i_0,i_1}
    \ar@<+4pt>[rr]
    \ar@<-4pt>[rr]
    &&
    \underset{i_0}{\coprod} U_{i_0}
   }
  \right]
  \,.
$$
A map of simplicial objects
$$
  C(U_i) \longrightarrow \mathbf{B}G
$$
is in degree 1 a collection of smooth $G$-valued functions $g_{i j} : U_{i j} \longrightarrow G$
and in degree 2 it is the condition that on $U_{i j k}$ these functions satisfy the cocycle condition
$g_{i j} \cdot g_{j k} = g_{i k}$. Hence this defines the transition functions for a $G$-principal bundle on $\Sigma$.
In physics this may be called the \emph{instanton sector} of a $G$-gauge field.
A $G$-gauge field itself is a connection on such a $G$-principal bundle, we come to this in a moment.

We may also think of the manifold $\Sigma$ itself as a simplicial object, one that does not actually depend on the simplicial degree.
Then there is a canonical projection map $C(\{U_i\}) \stackrel{\simeq}{\longrightarrow} \Sigma$.
When restricted to arbitrarily small open neighbourhoods (stalks) of points in $\Sigma$, then this projection
becomes a \emph{weak homotopy equivalence} of simplicial sets. We are to regard smooth simplicial presheaves
which are connected by morphisms that are stalkwise weak homotopy equivalences as equivalent. With this
understood, a smooth simplicial presheaf is also called a \emph{higher smooth stack}. Hence a $G$-principal bundle
on $\Sigma$ is equivalently a morphism of higher smooth stacks of the form
$$
  \Sigma \longrightarrow \mathbf{B}G
  \,.
$$

For analysing smooth symmetries it is useful to focus on infinitesimal symmetries.
To that end, consider the (first order) infinitesimal neighbourhood $\mathbb{D}_e(-)$ of the neutral element in the
simplicial nerve.
Here $\mathbb{D}_e(-)$ is the space around the neutral element that is ``so small'' that for any smooth function
on it which vanishes at $e$, the square of that function is ``so very small'' as to actually
be equal to zero.

We denote the resulting system of $k$-fold infinitesimal $G$-symmetries by $\mathbf{B}\mathfrak{g}$:
$$
  \mathbf{B}\mathfrak{g}
  =
  \left[
  \xymatrix{
    \ar@{..}[r]
    &
    \mathbb{D}_e(G \times G \times G)
    \ar@<+16pt>[rr]^{(p_1,p_2)}
    \ar@<+6pt>[rr]|-{(\mathrm{id},(-)\cdot (-))}
    \ar@<-6pt>[rr]|-{((-)\cdot (-), \mathrm{id})}
    \ar@<-16pt>[rr]_{(p_2,p_3)}
    &&
    \mathbb{D}_e(G \times G)
    \ar@<+8pt>[rr]^{p_1}
    \ar[rr]|{(-) \cdot (-)}
    \ar@<-8pt>[rr]_{p_2}
    &&
    \mathbb{D}_e(G)
    \ar@<+4pt>[rr]
    \ar@<-4pt>[rr]
    &&
    \ast
  }
  \right]
  \,.
$$
The alternating sum of pullbacks along the simplicial face maps shown above
defines a differential $d_{\mathrm{CE}}$ on the spaces of functions on these infinitesimal neighbourhoods.
The corresponding \emph{normalized chain complex}
is the differential-graded algebra on those functions which vanish when at least one of their
arguments is the neutral element in $G$. One finds that this is the Chevalley-Eilenberg complex
$$
  \mathrm{CE}(\mathbf{B}\mathfrak{g})
  =
  \left(
    \wedge^\bullet \mathfrak{g}^\ast,
    \;
    d_{\mathrm{CE}} = [-,-]^\ast
  \right)
  \,,
$$
which is the Grassmann algebra on the linear dual of the Lie algebra $\mathfrak{g}$ of $G$
equipped with the differential whose component $\wedge^1 \mathfrak{g}^\ast \to \wedge^2 \mathfrak{g}^\ast$
is given by the linear dual of the Lie bracket
$[-,-]$, and which hence extends to all higher degrees by the graded Leibnitz rule.

For example, when we choose $\{t_a\}$ a linear basis for $\mathfrak{g}$, with structure constants of the Lie bracket
denoted $[t_a,t_b] = C^c{}_{a b} t_c$, then with a dual basis $\{t^a\}$ of $\mathfrak{g}^\ast$ we have that
$$
  d_{\mathrm{CE}} t^a = \tfrac{1}{2} C^a{}_{b c} \, t^b \wedge t^c
  \,.
$$
Given any structure constants for a skew bracket like this, then the condition $(d_{\mathrm{CE}})^2 = 0$
is equivalent to the Jacobi identity, hence to the condition that the skew bracket indeed makes a Lie algebra.

Traditionally, the Chevalley-Eilenberg complex is introduced in order to define and
to compute Lie algebra cohomology: a $d_{\mathrm{CE}}$-closed element
$$
  \mu \in \wedge^{p+1} \mathfrak{g}^\ast \hookrightarrow \mathrm{CE}(\mathbf{B}\mathfrak{g})
$$
is equivalently a Lie algebra $(p+1)$-cocycle. This phenomenon will be crucial further below.

Thinking of $\mathrm{CE}(\mathbf{B}\mathfrak{g})$ as the algebra of functions on the infinitesimal
neighbourhood of the neutral element inside $\mathbf{B}G$ makes it plausible that this is an equivalent incarnation
of the Lie algebra of $G$. This is also easily checked directly: sending finite dimensional Lie algebras to
their Chevalley-Eilenberg algebra constitutes a fully faithful inclusion
$$
  \mathrm{CE}
  \;:\;
  \mathrm{LieAlg}
  \hookrightarrow
  \mathrm{dgcAlg}^{\mathrm{op}}
$$
of the category of Lie algebras into the opposite of the category of differential graded-commutative algebras.
This perspective turns out to be useful for computations in gauge theory and in higher gauge theory.
Therefore it serves to see how various familiar constructions on Lie algebras look when viewed in terms
of their Chevalley-Eilenberg algebras.

Most importantly, for $\Sigma$ a smooth manifold and $\Omega^\bullet(\Sigma)$ denoting its de Rham
dg-algebra of differential forms, then \emph{flat} $\mathfrak{g}$-valued 1-forms on $\Sigma$ are
equivalent to dg-algebra homomorphisms like so:
$$
  \Omega^\bullet_{\mathrm{flat}}(\Sigma,\mathfrak{g})
  :=
  \left\{
    A \in \Omega^1(\Sigma)\otimes \mathfrak{g} \;|\; F_A := d_{\mathrm{dR}}A - \tfrac{1}{2}[A \wedge A] = 0
  \right\}
  \;\simeq\;
  \left\{
    \;\;
    \Omega^\bullet(\Sigma)
    \longleftarrow
    \mathrm{CE}(\mathbf{B}\mathfrak{g})
    \;\;
  \right\}
  \,.
$$
To see this, notice that the underlying homomorphism of graded algebras $\Omega^\bullet(\Sigma) \longleftarrow \wedge^\bullet \mathfrak{g}^\ast$
is equivalently a $\mathfrak{g}$-valued 1-form, and that the respect for the differential forces it to be flat:
$$
  \xymatrix{
    A^a \ar@{|->}[d]^{d_{\mathrm{dR}}} \ar@{<-|}[rr] && t^a \ar@{|->}[dd]^{d_{\mathrm{CE}}}
    \\
    d_{\mathrm{dR}}A^a
    \ar@{=}[dr]
    \\
    & \tfrac{1}{2} C^a{}_{b c} A^b \wedge A^c
    \ar@{<-|}[r]
    &
    \tfrac{1}{2}C^{a}{}_{b c} t^b \wedge t^c
  }
$$

The flat Lie algebra valued forms play a crucial role in recovering a Lie group from its Lie algebra as the group of finite paths of
infinitesimal symmetries. To that end, write $\Delta^1 := [0,1]$ for the abstract interval. Then a $\mathfrak{g}$-valued
differential form $A \in \Omega^1_{\mathrm{flat}}(\Delta^1,\mathfrak{g})$ is at each point of $\Delta^1$ an infinitesimal symmetry,
hence it encodes the finite symmetry transformation that is given by applying the infinitesimal transformation
$A_t$ at each $t \in \Delta^1$ and then ``integrating these''. This integration is called the \emph{parallel transport}
of $A$ and is traditionally denoted by the symbols
$P \exp(\int_0^1 A) \in G$. Now of course different paths of infinitesimal transformations may have the same
integrated effect. But precisely if $A_1$ and
$A_2$ have the same integrated effect, then there is a flat $\mathfrak{g}$-valued 1-form on the disk which restricts
to $A_1$ on the upper semicircle and to $A_2$ on the lower semicircle.

In particular, the composition of two paths of infinitesimal gauge transformations is in general not
equal to any given such path with the same integrated effect, but there will always be a flat 1-form
$\hat A$ on the 2-simplex $\Delta^2$ which interpolates:

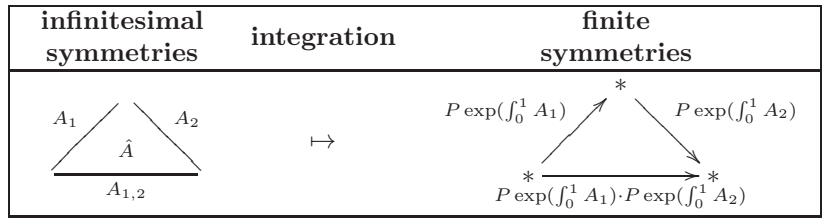
\begin{center}
\begin{tabular}{|ccc|}
\hline
\begin{tabular}{c}
  {\bf infinitesimal}
  \\
  {\bf symmetries}
\end{tabular}
&{\bf integration}&
\begin{tabular}{c}
  {\bf finite}
  \\
  {\bf symmetries}
\end{tabular}
\\
\hline
$
  \raisebox{20pt}{
  \xymatrix{
    &
    \ar@{-}[dr]^{A_2}
    \\
    \ar@{-}[ur]^{A_1}
    \ar@{-}[rr]_{A_{1,2}}^{\ }="s"
    &&
    \ar@{}|{\hat A} "s"; "s"+(0,5)
  }
  }
$
&
$\mapsto$
&
$
  \raisebox{24pt}{
  \xymatrix{
    & \ast
    \ar[dr]^{P \exp(\int_0^1 A_2)}
    \\
    \ast
    \ar[rr]_{P \exp(\int_0^1 A_1) \cdot P\exp(\int_0^1 A_2)}
    \ar[ur]^{P\exp(\int_{0}^1 A_1)}
    &&
    \ast
  }
  }
$
\\
\hline
\end{tabular}
\end{center}

In order to remember how the group obtained this way is a Lie group, we simply need to
remember how the above composition works in smoothly $U$-parameterized collections
of 1-forms. But a $U$-parameterized collection of 1-forms on $\Delta^k$ is simply
a 1-form on $U \times \Delta^k$ which vanishes on vectors tangent to $U$, hence a
vertical 1-form on $U \times \Delta^k$, regarded as a simplex bundle over $U$.

All this is captured by saying that there is a simplicial smooth presheaf $\exp(\mathfrak{g})$
which assigns to an abstract coordinate chart $U$ and a simplicial degree $k$ the
set of flat vertical $\mathfrak{g}$-valued 1-forms on $U \times \Delta^k$:
$$
  \begin{aligned}
  \exp(\mathfrak{g}) :=
  (U,k)
  \;\;
  & \mapsto
  \Omega^\bullet_{\mathrm{flat} \atop \mathrm{vert}}( U \times \Delta^k,\mathfrak{g})
  \\
  & =
  \left\{
    \;
    \Omega^\bullet_{\mathrm{vert}}(U \times \Delta^k)
    \longleftarrow
    \mathrm{CE}(\mathbf{B}\mathfrak{g})
    \;
  \right\}
  \end{aligned}
  \,.
$$
By the above discussion, we do not care which of various possible flat 1-forms $\hat A$
on 2-simplices are used to exhibit the composition of finite gauge transformation.
The technical term for retaining just the information that there is any such 1-form
on a 2-simplex at all
is to form the \emph{2-coskeleton}
$\mathrm{cosk}_2(\exp(\mathfrak{g}))$. And one finds that this indeed recovers the smooth
gauge group $G$, in that there is a weak equivalence of simplicial presheaves:
$$
  \mathrm{cosk}_3(\exp(\mathfrak{g}))
  \simeq
  \mathbf{B}G
  \,.
$$

So far this produces the gauge group itself from the infinitesimal symmetries. We now discuss
how similarly its action on gauge fields is obtained.
 To that end, consider the \emph{Weil algebra} of $\mathfrak{g}$, which
is obtained from the Chevalley-Eilenberg algebra by throwing in another copy of
$\mathfrak{g}$, shifted up in degree
$$
  W(\mathbf{B}\mathfrak{g})
  :=
  \left(
    \wedge^\bullet( \mathfrak{g}^\ast \oplus \mathfrak{g}^\ast[1] ), d_W = d_{\mathrm{CE}} + \mathbf{d}
  \right)
  \,,
$$
where $\mathbf{d} : \wedge^1 \mathfrak{g}^\ast \stackrel{\simeq}{\to} \mathfrak{g}^\ast[1]$ is the degree
shift and we declare $d_{\mathrm{CE}}$ and $\mathbf{d}$ to anticommute. So if $\{t^a\}$ is the dual basis
of $\mathfrak{g}^\ast$ from before, write $\{r^a\}$ for the same elements thought of in one degree higher as a basis
of $\mathfrak{g}^\ast[1]$; then
$$
  \begin{aligned}
    d_{\mathrm{W}} & : t^a \;\mapsto\; \tfrac{1}{2}C^a{}_{b c} t^b \wedge t^c + r^a
    \\
    d_{\mathrm{W}} & : r^a \;\mapsto\;  C^a{}_{b c}  t^b \wedge r^c
  \end{aligned}
  \,.
$$
A key point of this construction is that
dg-algebra homomorphisms out of the Weil algebra into a de Rham algebra are equivalent to unconstrained
$\mathfrak{g}$-valued differential forms:
$$
  \Omega(\Sigma,\mathfrak{g})
  :=
  \left\{
    A \in \Omega^1(\Sigma)\otimes \mathfrak{g}
  \right\}
  \;\;\simeq\;\;
  \left\{
    \;\;
    \Omega^\bullet(\Sigma)
    \longleftarrow
    \mathrm{W}(\mathbf{B}\mathfrak{g})
    \;\;
  \right\}
  \,.
$$
This is because now the extra generators $r^a$ pick up the failure of the respect for the $d_{\mathrm{CE}}$-differential,
that failure is precisely the curvature $F_A$:
$$
  \xymatrix{
    A^a \ar@{|->}[d]^{d_{\mathrm{dR}}} \ar@{<-|}[rr] && t^a \ar@{|->}[dd]^{d_{\mathrm{W}}}
    \\
    d_{\mathrm{dR}}A^a
    \ar@{=}[dr]
    \\
    & \tfrac{1}{2} C^a{}_{b c} A^b \wedge A^c + F_A^a
    \ar@{<-|}[r]
    &
    \tfrac{1}{2}C^{a}{}_{b c} t^b \wedge t^c + r^a
  }
  \;\;\;\;
  \xymatrix{
    F_A^a \ar@{|->}[d]^{d_{\mathrm{dR}}} \ar@{<-|}[rr] & & r^a \ar@{|->}[dd]^{d_{\mathrm{W}}}
    \\
    d_{\mathrm{dR}} F_A^a
    \ar@{=}[dr]
    \\
    &
    C^a{}_{b c} A^b \wedge F_A^c
    \ar@{<-|}[r]
    &
    C^a{}_{b c} t^b \wedge r^c
  }
  \,.
$$
Notice here that once $t^a \mapsto A^a$ is chosen, then the diagram on the left uniquely specifies that $r^a \mapsto F_A^a$
and then the diagram on the right is already implied: its commutativity is the \emph{Bianchi idenity} $d F_A = [A\wedge F_A]$
that is satisfied by curvature forms.

Traditionally, the Weil algebra is introduced in order to define and compute invariant polynomials
on a Lie algebra.
A $d_{\mathrm{W}}$-closed element in the shifted generators
$$
  \langle -,-,\cdots\rangle \in \wedge^{k} \mathfrak{g}^\ast[1] \hookrightarrow \mathrm{W}(\mathbf{B}\mathfrak{g})
$$
is equivalently a invariant polynomial of order $k$ on the Lie algebra $\mathfrak{g}$.
Therefore write
$$
  \mathrm{inv}(\mathbf{B}\mathfrak{g})
$$
for the graded commutative algebra of invariant polynomials,
thought of as a dg-algebra with vanishing differential.

(For notational convenience we will later often abbreviate $\mathrm{CE}(\mathfrak{g})$ for $\mathrm{CE}(\mathbf{B}\mathfrak{g})$, etc.
This is unambiguous as long as no algebroids with nontrivial bases spaces appear.)

There is a canonical projection map from the Weil algebra to the Chevalley-Eilenberg algebra,
given simply by forgetting the shifted generators
($t^a \mapsto t^a$; $r^a \mapsto 0$). And there is the defining inclusion
$\mathrm{inv}(\mathbf{B}\mathfrak{g}) \hookrightarrow \mathrm{W}(\mathbf{B}\mathfrak{g})$.
$$
  \xymatrix{
    \mathrm{CE}(\mathbf{B}\mathfrak{g})
    \\
    \mathrm{W}(\mathbf{B}\mathfrak{g})
    \ar[u]
    \\
    \mathrm{inv}(\mathbf{B}\mathfrak{g})
    \ar[u]
  }
$$

Cartan had introduced all these dg-algebras as algebraic models of the universal $G$-principal bundle.
We had seen above that homomorphisms $\Omega^\bullet_{\mathrm{vert}}(U \times \Delta^k) \longleftarrow \mathrm{CE}(\mathbf{B}\mathfrak{g})$
constitute the gauge symmetry group $G$ as integration of the paths of infinitesimal symmetries. Here the vertical forms
on $U \times \Delta^k$ are themselves part of the sequence of differential forms on the trivial $k$-simplex bundle over the given coordinate chart $U$.
Hence consider compatible dg-algebra homomorphisms between these two sequences
$$
  \xymatrix{
    \Omega^{\bullet}_{\mathrm{vert}}(U \times \Delta^k)
    &&
    \mathrm{CE}(\mathbf{B}\mathfrak{g})
    \ar[ll]_-\kappa
    \\
    \Omega^\bullet(U \times \Delta^k)
    \ar[u]
    &&
    \mathrm{W}(\mathbf{B}\mathfrak{g})
    \ar[u]
    \ar[ll]_-{A}
    \\
    \Omega^\bullet(U)
    \ar[u]
    &&
    \mathrm{inv}(\mathbf{B}\mathfrak{g})
    \ar[u]
    \ar[ll]^-{\langle F_A \wedge \cdots \wedge F_A\rangle}
  }
$$
We unwind what this means in components: The middle morphism
is an unconstrained Lie algebra valued form $A\in \Omega^1(U \times \Delta^k, \mathfrak{g})$,
hence is a sum
$$
  A = A_U + A_{\Delta^k}
$$
of a 1-form $A_U$ along $U$ and 1-form $A_{\Delta^k}$ along $\Delta^k$.
The second summand $A_{\Delta^k}$
is the vertical component of $A$. The commutativity of the top square above says that as a vertical differential
form, $A_{\Delta^k}$ has to be flat. By the previous discussion this means that $A_{\Delta^k}$ encodes
a $k$-tuple of $G$-gauge transformations. Now we will see how these gauge transformations naturally act on the
gauge field $A_U$:

Consider this for the case $k = 1$, and write $t$ for the canonical coordinate along $\Delta^1 = [0,1]$.
Then $A_U$ is a smooth $t$-parameterized collection of 1-forms, hence of $\mathfrak{g}$-gauge fields, on $U$;
and $A_{\Delta^1} = \kappa \,d t$ for $\kappa$ a smooth Lie algebra valued function, called the \emph{gauge parameter}.
Now the equation for the $t$-component of the total curvature $F_A$ of $A$ says how the gauge parameter together with the
mixed curvature component causes infinitesimal transformations of the  gauge field $A_U$ as $t$ proceeds:
$$
  \frac{d}{dt} A_U = d_U \kappa - [\kappa, A] + \iota_{\partial_t} F_A
  \,.
$$
But now the commutativity of the lower square above demands that the curvature forms evaluated in
invariant polynomials have vanishing contraction with $\iota_t$. In the case that
$\mathfrak{g} = \mathbb{R}$ this means that $\iota_{\partial t} F_A = 0$, while for nonabelian $\mathfrak{g}$
this is still generically the necessary condition. So for vanishing $t$-component of the curvature
the above equation says that
$$
  \frac{d}{dt} A_U = d\kappa - [\kappa, A]
  \,.
$$
This is the traditional formula for infinitesimal gauge transformations $\kappa$ acting on
a gauge field $A_U$. Integrating this up, $\kappa$ integrates to a gauge group element
$g := P \exp(\int_0^1 \kappa d t)$ by the previous discussion, and this equation
becomes the formula for finite gauge transformations (where we abbreviate now $A_t := A_U(t)$):
$$
  A_1
  =
  g^{-1} A_0 g + g^{-1} d_{\mathrm{dR}}g
  \,.
$$
This gives the smooth groupoid $\mathbf{B}G_{\mathrm{conn}}$ of $\mathfrak{g}$-gauge fields with $G$-gauge transformations
between them.
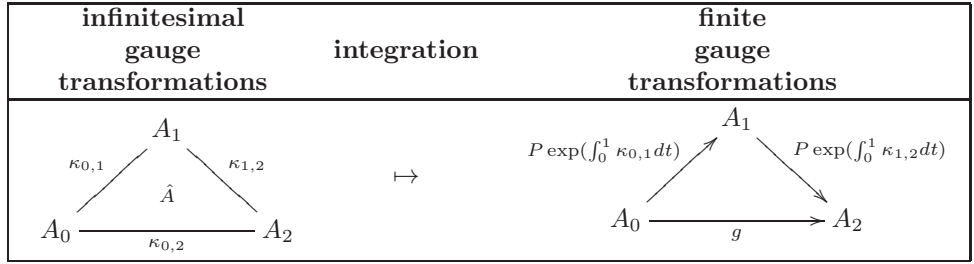
\begin{center}
\begin{tabular}{|ccc|}
\hline
\begin{tabular}{c}
  {\bf infinitesimal}
  \\
  {\bf gauge}
  \\
  {\bf transformations}
\end{tabular}
&{\bf integration}&
\begin{tabular}{c}
  {\bf finite}
  \\
  {\bf gauge}
  \\
  {\bf transformations}
\end{tabular}
\\
\hline
$
  \raisebox{20pt}{
  \xymatrix{
    &
    A_1
    \ar@{-}[dr]^{\kappa_{1,2}}
    \\
    A_0
    \ar@{-}[ur]^{\kappa_{0,1}}
    \ar@{-}[rr]_{\kappa_{0,2}}^{\ }="s"
    &&
    A_2
    \ar@{}|{\hat A} "s"; "s"+(0,8)
  }
  }
$
&
$\mapsto$
&
$
  \raisebox{24pt}{
  \xymatrix{
    & A_1
    \ar[dr]^{P \exp(\int_0^1 \kappa_{1,2}dt)}
    \\
    A_0
    \ar[rr]_g
    \ar[ur]^{P\exp(\int_{0}^1 \kappa_{0,1} dt)}
    &&
    A_2
  }
  }
$
\\
\hline
\end{tabular}
\end{center}
Hence $\mathbf{B}G_{\mathrm{conn}}$ is the smooth groupoid such that for $U$ an abstract coordinate chart,
the smoothly $U$-parameterized collections of its objects are
$\mathfrak{g}$-valued differential forms $A \in \Omega(U,\mathfrak{g})$ of, and whose $U$-parameterized collections
of gauge transformations are $G$-valued functions $g$ acting by
$$
  \xymatrix{
   U
   \ar@/_1.8pc/[rr]_{g^{-1}A g + g^{-1}d g}^{\ }="t"
   \ar@/^1.8pc/[rr]^{A}_{\ }="s"
   &&
   \mathbf{B}G_{\mathrm{conn}}
   \ar@{=>}^g "s"; "t"
  }
  \,.
$$

This dg-algebraic picture of gauge fields with gauge transformations between them
now immediately generalizes to higher gauge fields with higher gauge transformations
between them. Moreover, this picture allows to produce prequantized higher Chern-Simons-type
Lagrangians by Lie integration of transgressive $L_\infty$-cocycles.

%%%%%%%%%%%%%%%%%%%%%%%%%%%%%%%%%%%%%%%%%
\paragraph{Higher gauge fields.}
\label{introductionhighergaugefields}
%%%%%%%%%%%%%%%%%%%%%%%%%%%%%%%%%%%%%%%%%

Ordinary gauge fields are characterized by the property that there are no non-trivial
gauge-of-gauge transformations, equivalently that their BRST complexes contain no higher
order ghosts.  Mathematically, it is natural to generalize beyond this case to
\emph{higher gauge fields}, which do have non-trivial higher gauge transformations.
The simplest example is a ``2-form field'' (``$B$-field''), generalizing the ``vector potential'' 1-form
$A$ of the electromagnetic field. Where such a 1-form has gauge transformations given by
0-forms (functions) $\kappa$ via
$$
  \xymatrix{
    A
    \ar@/^1pc/[rr]^-\kappa
    &&
    A' = A +d\kappa
  }
  \,,
$$
a 2-form $B$ has gauge transformations given by 1-forms $\rho_1$, which themselves then have
gauge-of-gauge-transformations given by 0-forms $\rho_0$:
$$
  \xymatrix{
    B
    \ar@/_2pc/[rr]_{\rho'_1 = \rho_1+ d\rho_0}^{\ }="t"
    \ar@/^2pc/[rr]^{\rho_1}_{\ }="s"
    &&
    \rlap{$B' = B + d \rho_1 = B + d\rho'_1$}
    \ar@{=>}^{\rho_0} "s"; "t"
  }
  ,.
$$
Next a ``3-form field'' (``$C$-field'') has third order gauge transformations:
$$
  \xymatrix{
    C
    \ar@/_2pc/[rr]_{{\rho'_2 = \rho_2+ d\rho_1} \atop {\;\;\;\; = \rho_2 + d \rho'_1}}^{\ }="t"
    \ar@/^2pc/[rr]^{\rho_2}_{\ }="s"
    &&
    \rlap{$C' = C + d \rho_2 = C + d\rho'_2$}
    \ar@{}|{\stackrel{\rho_0}{\Leftarrow}} "s"; "t"
    \ar@/^1pc/@{=>}^{\rho_1} "s"; "t"
    \ar@/_1pc/@{=>}_{\rho'_1  } "s"; "t"
  }
  ,.
$$
Similarly ``$n$-form fields'' have order-$n$ gauge-of-gauge transformations and hence have
order-$n$ ghost-of-ghosts in their BRST complexes.

Higher gauge fields have not been experimentally observed, to date, as fundamental fields of nature,
but they appear by necessity and ubiquitously in higher dimensional supergravity and in the
hypothetical physics of strings and $p$-branes. The higher differential geometry which we develop
is to a large extent motivated by making precise and tractable the global structure of higher
gauge fields in string and M-theory.

Generally, higher gauge fields are part of mathematical physics
just as the Ising model and $\phi^4$-theory are, and as such they do serve to illuminate
the structure of experimentally verified physics.
For instance the Einstein equations of motion
for ordinary (bosonic) general relativity on 11-dimensional spacetimes are equivalent
to the full super-torsion constraint in 11-dimensional supergravity with its 3-form
higher gauge field \cite{CandielloLechner}.
From this point of view one may regard the
the 3-form higher gauge field in supergravity, together with the gravitino, as auxiliary fields
that serve to present Einstein's equations for the graviton in a particularly neat mathematical way.

We now use the above dg-algebraic formulation of ordinary gauge fields above in section \ref{introductionordinarygaugefields}
in order to give a quick but accurate idea of the mathematical structure of higher gauge fields.

\medskip

Above we saw that (finite dimensional) Lie algebras are equivalently the formal duals of those differential graded-commutative
algebras whose underlying graded commutative algebra is freely generated from a (finite dimensional)
vector space over the ground field. From this perspective, there are two evident generalizations
to be considered: we may take the underlying vector space to already have contributions in
higher degrees itself, and we may pass from vector spaces, being modules over the ground field $\mathbb{R}$,
to (finite rank) projective modules over an algebra of smooth functions on a smooth manifold.

Hence we say that an \emph{$L_\infty$-algebroid} (of finite type) is a smooth manifold $X$ equipped with a
$\mathbb{N}$-graded vector bundle (degreewise of finite rank),
whose smooth sections hence form an $\mathbb{N}$-graded projective $C^\infty(X)$-module
$\mathfrak{a}_\bullet$, and equipped with an $\mathbb{R}$-linear differential $d_{\mathrm{CE}}$ on the Grassmann algebra of
the $C^\infty(X)$-dual $\mathfrak{a}^\ast$ modules
$$
  \mathrm{CE}(\mathfrak{a})
  :=
  \left(
    \wedge^\bullet_{C^\infty(X)}(\mathfrak{a}^\ast),
    \;
    d_{\mathrm{CE}(\mathfrak{a})}
  \right)
  \,.
$$
Accordingly, a homomorphism of $L_\infty$-algebroids we take to be a dg-algebra homomorphism (over $\mathbb{R}$)
of their CE-algebras going the other way around. Hence the category of $L_\infty$-algebroids
is the full subcategory of the opposite of that of differential graded-commutative algebras over $\mathbb{R}$
on those whose underlying graded-commutative algebra is free on graded locally free projective
$C^\infty(X)$-modules:
$$
  L_\infty\mathrm{Algbd}
  \hookrightarrow
  \mathrm{dgcAlg}^\mathrm{op}
  \,.
$$
We say we have a \emph{Lie $n$-algebroid} when $\mathfrak{a}$ is concentrated in the lowest $n$-degrees.
Here are some important examples of $L_\infty$-algebroids:

When the base space is the point, $X = \ast$, and $\mathfrak{a} $ is concentrated in degree 0, then
we recover {\bf Lie algebras}, as above. Generally, when the base space is the point, then the $\mathbb{N}$-graded module $\mathfrak{a}$
is just an $\mathbb{N}$-graded vector space $\mathfrak{g}$. We write $\mathfrak{a} = \mathbf{B}\mathfrak{g}$
to indicate this, and then $\mathfrak{g}$ is an {\bf $L_\infty$-algebra}.
When in addition $\mathfrak{g}$ is concentrated in the lowest $n$ degrees, then these are also
called {\bf Lie $n$-algebras}. With no constraint on the grading but assuming that the differential sends single generators
always to sums of wedge products of at most two generators, then we get {\bf dg-Lie algebras}.

The Weil algebra of a Lie algebra $\mathfrak{g}$ hence exhibits a Lie 2-algebra. We may think of this as the Lie 2-algebra
$\mathrm{inn}(\mathfrak{g})$ of inner derivations of $\mathfrak{g}$. By the above discussion, it is suggestive to
write $\mathbf{E}\mathfrak{g}$ for this Lie 2-algebra, hence
$$
  \mathrm{W}(\mathbf{B}\mathfrak{g}) = \mathrm{CE}(\mathbf{B}\mathbf{E}\mathfrak{g})
  \,.
$$

If $\mathfrak{g} = \mathbb{R}[n]$ is concentrated in degree $p$ on the real line
(so that the CE-differential is necessarily trivial), then we speak
of the \emph{line Lie $(p+1)$-algebra} $\mathbf{B}^p \mathbb{R}$, which as an $L_\infty$-algebroid over the point
is to be denoted
$$
  \mathbf{B}\mathbf{B}^{p}\mathbb{R} = \mathbf{B}^{p+1}\mathbb{R}
  \,.
$$

All this goes through verbatim, up to additional signs, with all vector spaces generalized to super-vector spaces. The Chevalley-Eilenberg
algebras of the resulting {\bf super $L_\infty$-algebras} are known in parts of the supergravity literature
as {\bf FDA}s \cite{CDF}.

Passing now to $L_\infty$-algebroids over non-trivial base spaces, first of all every smooth manifold $X$
may be regarded as the $L_\infty$-algebroid over $X$, these are the Lie 0-algebroids. We just write
$\mathfrak{a} = X$ when the context is clear.

For the tangent bundle $T X$ over $X$
then the graded algebra of its dual sections is the wedge product algebra of differential forms,
$\mathrm{CE}(T X ) = \Omega^\bullet(X)$ and hence the de Rham differential makes
$\wedge^\bullet \Gamma(T^\ast X)$ into a dgc-algebra and hence makes $T X$ into a
Lie algebroid. This is called the  {\bf tangent Lie algebroid} of $X$.
We usually write $\mathfrak{a} = TX$ for the tangent Lie algebroid
(trusting that context makes it clear that we do not mean the Lie 0-algebroid over the underlying manifold
of the tangent bundle itself).
In particular this means that
for any other $L_\infty$-algebroid $\mathfrak{a}$ then
flat $\mathfrak{a}$-valued differential forms on some smooth manifold $\Sigma$ are equivalently
homomorphisms of $L_\infty$-algebroids like so:
$$
  \Omega_{\mathrm{flat}}(\Sigma,\mathfrak{a})
  \;\;=\;\;
  \left\{
    \; T \Sigma \longrightarrow \mathfrak{a}\;
  \right\}
  \,.
$$
In particular ordinary closed differential forms of degree $n$ are equivalently
flat $\mathbf{B}^n\mathbb{R}$-valued differential forms:
$$
  \Omega^n_{\mathrm{cl}}(\Sigma)
  \;\; \simeq \;\;
  \left\{
    \; T \Sigma \longrightarrow \mathbf{B}^n \mathbb{R}\;
  \right\}
  \,.
$$

More generally,
for $\mathfrak{a}$ any $L_\infty$-algebroid over some base manifold $X$, then we have its Weil dgc-algebra
$$
  \mathrm{W}(\mathfrak{a})
  :=
  \left(
    \wedge^\bullet_{C^\infty(X)}( \mathfrak{a}^\ast \oplus \Gamma(T^\ast X) \oplus \mathfrak{a}^\ast[1] ),
      d_{\mathrm{W}} = d_{\mathrm{CE}} + \mathbf{d} )
  \right)
  \,,
$$
where $\mathbf{d}$ acts as the degree shift isomorphism in the components
$\wedge^1_{C^\infty(X)}\mathfrak{a}^\ast \longrightarrow \wedge^1_{C^\infty(X)}\mathfrak{a}^\ast[1]$
and as the de Rham differential in the components $\wedge^k \Gamma(T^\ast X) \to \wedge^{k+1}\Gamma(T^\ast X)$.
This defines a new $L_\infty$-algebroid that may be called the {\bf tangent $L_\infty$-algebroid} $T \mathfrak{a}$
$$
  \mathrm{CE}(T \mathfrak{a}) := \mathrm{W}(\mathfrak{a})
  \,.
$$
We also write $\mathbf{E}\mathbf{B}^p \mathbb{R}$ for the $L_\infty$-algebroid with
$$
  \mathrm{CE}(\mathbf{E}\mathbf{B}^p \mathbb{R})
  :=
  \mathrm{W}(\mathbf{B}^{p}\mathbb{R})
  \,.
$$

In direct analogy with the discussion for Lie algebras, we then say that an unconstrained
$\mathfrak{a}$-valued differential form $A$ on a manifold $\Sigma$ is a dg-algebra homomorphism from the
Weil algebra of $\mathfrak{a}$ to the de Rham dg-algebra on $\Sigma$:
$$
  \Omega(\Sigma,\mathfrak{a})
  :=
  \left\{
    \; \Omega^\bullet(\Sigma) \longleftarrow \mathrm{W}(\mathfrak{a}) \;
  \right\}
  \,.
$$

For $G$ a Lie group acting on $X$ by diffeomorphisms, then there is the {\bf action Lie algebroid} $X/\mathfrak{g}$
over $X$ with
$\mathfrak{a}_0 = \Gamma_X(X \times \mathfrak{g})$ the $\mathfrak{g}$-valued smooth functions over $X$.
Write $\rho : \mathfrak{g} \to \mathrm{Vect}$
for the linearized action. With a choice of basis $\{t_a\}$ for $\mathfrak{g}$ as before
and assuming that $X = \mathbb{R}^n$ with canonical coordinates $x^i$, then $\rho$
has components $\{\rho_a^\mu\}$ and the CE-differential on $\wedge^\bullet_{C^\infty(X)} (\Gamma_X(X \times \mathfrak{g}^\ast))$
is given on generators by
$$
  \begin{aligned}
    d_{\mathrm{CE}} &: f \mapsto t^a \rho_a^\mu \partial_\mu f
    \\
    d_{\mathrm{CE}} &: t^a \mapsto \tfrac{1}{2}C^{a}{}_{b c} t^b \wedge t^c
  \end{aligned}
  \,.
$$
In the physics literature this Chevalley-Eilenberg algebra
$\mathrm{CE}(X/\mathfrak{g})$ is known as the {\bf BRST complex}
of $X$ for infinitesimal symmetries $\mathfrak{g}$. If $X$ is thought of as a space of fields, then
the $t_a$ are called \emph{ghost fields}.\footnote{More generally the base manifold $X$ may be a derived manifold/BV-complex
as in footnote \ref{BVremark}. Then $\mathrm{CE}(X/\mathfrak{g})$ is known as the ``BV-BRST complex''.}

Given any $L_\infty$-algebroid, it induces further $L_\infty$-algebroids via its extension by higher cocycles.
A \emph{$p+1$-cocycle} on an $L_\infty$-algebroid $\mathfrak{a}$ is a closed element
$$
  \mu \in (\wedge^\bullet_{C^\infty(X)}\mathfrak{a}^\ast)_{p+1} \hookrightarrow \mathrm{CE}(\mathfrak{a})
  \,.
$$
Notice that now cocycles are \emph{representable} by the higher line $L_\infty$-algebras
$\mathbf{B}^{p+1}\mathbb{R}$ from above:
$$
  \begin{aligned}
  \left\{
    \mu \in \mathrm{CE}(\mathfrak{a})_{p+1} \;|\; d_{\mathrm{CE}}\mu = 0
  \right\}
  &
  \simeq
  \left\{
    \;
    \mathrm{CE}(\mathfrak{a})
    \stackrel{\mu^\ast}{\longleftarrow}
    \mathrm{CE}(\mathbf{B}^{p+1}\mathbb{R})
    \;
  \right\}
  \\
  & =
  \left\{
    \;
    \mathfrak{a} \stackrel{\mu}{\longrightarrow} \mathbf{B}^{p+1}\mathbb{R}
    \;
  \right\}
  \end{aligned}
  \,.
$$
It is a traditional fact that $\mathbb{R}$-valued 2-cocycles on a Lie algebra induce central Lie algebra extensions. More generally,
higher cocycles $\mu$ on an $L_\infty$-algebroid induce $L_\infty$-extensions $\hat {\mathfrak{a}}$, given by the pullback
$$
  \raisebox{20pt}{
  \xymatrix{
    \hat{\mathfrak{a}} \ar@{}[ddrr]|{\mbox{\tiny (pb)}} \ar[dd] \ar[rr] && \mathbf{E}\mathbf{B}^{p}\mathbb{R} \ar[dd]
    \\
    \\
    \mathfrak{a}
    \ar[rr]^\mu
    &&
    \mathbf{B}^{p+1}\mathbb{R}
  }
  }
  \,.
$$
Equivalently this makes $\hat{\mathfrak{a}}$ be the homotopy fiber of $\mu$ in the homotopy theory of
$L_\infty$-algebras, and induces a long homotopy fiber sequence of the form
$$
  \raisebox{20pt}{
  \xymatrix{
    \mathbf{B}^p \mathbb{R} \ar[r] & \hat{\mathfrak{a}} \ar[d]
    \\
    & \mathfrak{a} \ar[rr]^{\mu}
    &&
    \mathbf{B}^{p+1}\mathbb{R}
  }
  }\,.
$$
In components this means simply that $\mathrm{CE}(\hat{\mathfrak{a}})$ is obtained from $\mathrm{CE}(\mathfrak{a})$
by adding one generator $c$ in degre $p$ and extending the differential to it by the formula
$$
  d_{\mathrm{CE}} : c = \mu
  \,.
$$
This construction has a long tradition in the supergravity literature \cite{CDF}\cite{InfinityWZW},
we come to the examples considered there below in section \ref{introductionWZWmodels}.
Iterating this construction, out of every $L_\infty$-algebroid their grows a whole bouquet of
further $L_\infty$-algebroids
$$
  \xymatrix{
    \ar@{..}[d]
    &
    \ar@{..}[dr]
    &
    \ar@{..}[d]
    &
    \ar@{..}[dl]
    \\
    \ar[dr] && \ar[dl]
    \\
    &\hat{\hat{\mathfrak{a}}}
    \ar[d]
    \\
    &\hat{\mathfrak{a}}
    \ar[d]
    \ar[rr]^-{\mu_2}
    &&
    \mathbf{B}^{p_2+2}
    \\
    & \mathfrak{a}
    \ar[rr]^-{\mu_1}
    &&
    \mathbf{B}^{p_1+2}
  }
$$

For example for $\mathfrak{g}$ a semisimple Lie algebra with binary invariant polynomial $\langle -,-\rangle$
(the Killing form), then $\mu_3 = \langle-,[-,-]\rangle$ is a 3-cocycle. The $L_\infty$-extension by this
cocycle is a Lie 2-algebra called the \emph{string Lie 2-algebra} $\mathfrak{string}_\mathfrak{g}$.
If $\{t^a\}$ is a linear basis of $\mathfrak{g}^\ast$ as before write $k_{a b} := \langle t_a, t_b\rangle$
for the components of the Killing form; the components of the 3-cocycle are
$\mu_{a b c} = k_{a a'}C^{a'}{}_{b c}$.  The CE-algebra of the string Lie 2-algebra then is that of
$\mathfrak{g}$ with a generator $b$ added and with CE-differential defined by
$$
  \begin{aligned}
    d_{\mathrm{CE}(\mathfrak{string})} & : t^a \mapsto \tfrac{1}{2}C^a{}_{b c} t^b \wedge t^c
    \\
    d_{\mathrm{CE}(\mathfrak{string})} & : b \mapsto k_{a a'}C^{a'}{}_{c b} t^a \wedge t^b \wedge t^c
    \,.
  \end{aligned}
$$
Hence a flat $\mathfrak{string}_{\mathfrak{g}}$-valued differential form on some $\Sigma$
is a pair consisting of an ordinary  flat $\mathfrak{g}$-valued 1-form $A$ and of a 2-form
$B$ whose differential has to equal the evaluation of $A$ in the 3-cocoycle:
$$
  \Omega_{\mathrm{flat}}(\Sigma,\mathfrak{string}_{\mathfrak{g}})
  \;\simeq\;
  \left\{
    (A,B) \in \Omega^1(\Sigma,\mathfrak{g})\times\Omega^2(\Sigma)
    \;|\;
    F_A = 0 \,,\;\; d B = \langle A \wedge [A \wedge A]\rangle
  \right\}
  \,.
$$
Notice that since $A$ is flat, the 3-form $\langle A \wedge [A \wedge A]\rangle$ is its Chern-Simons 3-form.
More generally, Chern-Simons forms are such that their differential is the evaluation of the curvature
of $A$ in an invariant polynomial.

An invariant polynomial $\langle -\rangle$ on an $L_\infty$-algebroid
we may take to be a $d_{\mathrm{W}}$-closed element in the shifted generators of its Weil algebra $\mathrm{W}(\mathfrak{a})$
$$
  \langle -\rangle \in \wedge^\bullet_{C^\infty(X)}(\mathfrak{a}^\ast[1]) \hookrightarrow \mathrm{W}(\mathfrak{a})
  \,.
$$
When one requires the invariant polynomial to be binary, i.e. in
$\wedge^2 (\mathfrak{a}^\ast[1]) \to \mathrm{W}(\mathfrak{a})$ and non-degenerate,
then it is also called a \emph{shifted symplectic form} and it makes
$\mathfrak{a}$ into a ``symplectic Lie $n$-algebroid''. For $n = 0$ these are the
symplectic manifolds, for $n = 1$ these are called \emph{Poisson Lie algebroids},
for $n = 2$ they are called \emph{Courant Lie 2-algebroids} \cite{RoytenbergCourant}.
There are also plenty of non-binary invariant polynomials, we discuss further examples
below in section \ref{introductionchernsimonstypefieldtheories}.

Being $d_{\mathrm{W}}$-closed, an invariant polynomial on $\mathfrak{a}$
is represented by a dg-homomorphism:
$$
  W(\mathfrak{a}) \longleftarrow \mathrm{CE}(\mathbf{B}^{p+2}\mathbb{R}) : \langle - \rangle
$$
This means that given an invariant polynomial $\langle -\rangle$ for an
$L_\infty$-algebroid $\mathfrak{a}$, then it assigns to any $\mathfrak{a}$-valued
differential form $A$ a plain closed $(p+2)$-form $\langle F_A\rangle$ made up of the $\mathfrak{a}$-curvature forms,
namely the composite
$$
  \Omega^\bullet(\Sigma)
  \stackrel{A}{\longleftarrow}
  W(\mathfrak{a})
  \stackrel{\langle -\rangle}{\longleftarrow} \mathrm{CE}(\mathbf{B}^{p+2}\mathbb{R}) : \langle F_A \rangle
  \,.
$$
In other words, $A$ may be regarded as a nonabelian pre-quantization of $\langle F_A\rangle$.

Therefore we may consider now the $\infty$-groupoid of $\mathfrak{a}$-connections whose
gauge transformations preserve the specified invariant polynomial, such as to guarantee that
it remains a globally well-defined differential form. The smooth $\infty$-groupoid of
$\mathfrak{a}$-valued connections with such gauge transformations between them we write
$\exp(\mathfrak{a})_{\mathrm{conn}}$. As a smooth simplicial presheaf, it is hence
given by the following assignment:
$$
  \exp(\mathfrak{a})_{\mathrm{conn}}
  \;:\;
  (U,k)
  \mapsto
  \left\{
    \raisebox{44pt}{
    \xymatrix{
      \Omega^{\bullet}_{\mathrm{vert}}(U \times \Delta^k)
      &&
      \mathrm{CE}(\mathfrak{a})
      \ar[ll]
      \\
      \Omega^\bullet(U \times \Delta^k)
      \ar[u]
      &&
      \mathrm{W}(\mathfrak{a})
      \ar[u]
      \ar[ll]_{A}
      \\
      \Omega^\bullet(U)
      \ar[u]
      &&
      \mathrm{inv}(\mathfrak{a})
      \ar[u]
      \ar[ll]^{\langle F_A\rangle}
    }
    }
  \right\}
$$
Here on the right we have, for every $U$ and $k$, the set of those $A$ on $U \times \Delta^k$
that induce gauge transformations along the $\Delta^k$-direction (that is the commutativity of the
top square) such that the given invariant polynomials evaluated on the curvatures are preserved
(that is the commutativity of the bottom square).

This $\exp(\mathfrak{a})_{\mathrm{conn}}$ is the moduli stack of $\mathfrak{a}$-valued connections with gauge transformations
and gauge-of-gauge transformations between them that preserve the chosen invariant polynomials \cite{FSS}\cite{frs}.

The key example is the moduli stack of $(p+1)$-form gauge fields
$$
  \exp(\mathbf{B}^{p+1}\mathbb{R})_{\mathrm{conn}}/\mathbb{Z}
  \simeq
  \mathbf{B}(\mathbb{R}/_{\!\hbar}\mathbb{Z})_{\mathrm{conn}}
$$

Generically we write
$$
  \mathbf{A}_{\mathrm{conn}} := \mathrm{cosk}_{n+1} (\exp(\mathfrak{a})_{\mathrm{conn}})
$$
for the $n$-truncation of a higher smooth stack of $\mathfrak{a}$-valued gauge field connections
obtained this way. If $\mathfrak{a} = \mathbf{B}\mathfrak{g}$ then we write $\mathbf{B}G_{\mathrm{conn}}$
for this.

Given such, then an $\mathfrak{a}$-gauge field on $\Sigma$ (an $\mathbf{A}$-principal connection) is equivalently a map of smooth higher stacks
$$
  \nabla : \Sigma \longrightarrow \mathbf{A}_{\mathrm{conn}}
  \,.
$$
By the above discussion, a simple map like this subsumes all of the following component data:
\begin{enumerate}
  \item
    a choice of open cover $\{U_i \to \Sigma\}$;
  \item
    a $\mathfrak{a}$-valued differential form $A_i$ on each chart $U_i$;
  \item
    on each intersection $U_{i j}$ of charts
    a path of infinitesimal gauge symmetries whose integrated finite gauge symmetry $g_{i j}$ takes $A_i$
    to $A_j$;
  \item
    on each triple intersection $U_{i j k}$ of charts a path-of-paths of infinitesimal gauge symmetries
    whose integrated finite gauge-of-gauge symmetry takes the gauge transformation $g_{i j}\cdot g_{j k}$
    to the gauge transformation $g_{i k}$
  \item
    and so on.
\end{enumerate}
Hence a $\mathfrak{a}$-gauge field is locally $\mathfrak{a}$-valued differential form data which are coherently
glued together to a global structure by gauge transformations and higher order gauge-of-gauge transformations.

Given two globally defined $\mathfrak{a}$-valued gauge fields this way, then
a globally defined gauge transformation them is equivalently a homotopy between maps
of smooth higher stacks
$$
  \xymatrix{
    \Sigma
    \ar@/^2pc/[rr]^{\nabla_1}_{\ }="s"
    \ar@/_2pc/[rr]_{\nabla_2}^{\ }="t"
    &&
    \mathbf{A}_{\mathrm{conn}}
    \ar@{=>}^\simeq "s"; "t"
  }
  \,.
$$
Again, this concisely encodes a system of local data: this is on each chart $U_i$
a path of inifinitesimal gauge symmetries whose integrated gauge transformation
transforms the local $\mathfrak{a}$-valued forms into each other, together with
various higher order gauge transformations and compatibilities on higher order
intersections of charts.

Then a gauge-of-gauge transformation is a homotopy of homotopies
$$
  \xymatrix{
    \Sigma
    \ar@/^2pc/[rr]^{\nabla_1}_{\ }="s"
    \ar@/_2pc/[rr]_{\nabla_2}^{\ }="t"
    &&
    \mathbf{A}_{\mathrm{conn}}
    \ar@/^1pc/@{=>} "s"; "t"
    \ar@/_1pc/@{=>} "s"; "t"
    \ar@{}|{\Leftarrow} "s"; "t"
  }
$$
and again this encodes a recipe for how to extract the corresponding local differential form data.

%%%%%%%%%%%%%%%%%%%%%%%%%%%%%%%%%%%%%%%%%%%%%%%%
\paragraph{The BV-BRST complex}
\label{introductionbvcomplex}
%%%%%%%%%%%%%%%%%%%%%%%%%%%%%%%%%%%%%%%%%%%%%%%%

(This subsection is momentarily just a quick side remark. )

The category of partial differential equations that we referred to so far, as in \cite{Marvan86},
is modeled on the category of smooth manifolds. Accordingly, it really only contains
differential equations that are non-singular enough such as to guarantee that the shell locus
$\mathcal{E} \hookrightarrow J^\infty E$ is itself a smooth manifold.
This is not the case for all differential equations of interest. For some pairs of differential operators,
their equalizer  $\xymatrix{ E \ar@<-4pt>[r] \ar@<+4pt>[r] & F }$ does not actually exist in smooth bundles
modeled on manifolds.

This is no problem when working in the sheaf topos over $\mathrm{PDE}_\Sigma$, where all limits do exist as diffeological bundles.
However, even though all limits exist here, some do not interact properly with other constructions of interest.
For instance intersection products in cohomology will not properly count non-transversal intersections,
even if they do exist as diffeological spaces.

To fix this, we may pass to a category of ``derived manifolds''.
In generalization of how an ordinary smooth manifold
is the formal dual to its real algebra of smooth functions, via the faithful embedding
$$
  C^\infty : \mathrm{SmoothMfd} \hookrightarrow \mathrm{CAlg}_{\mathbb{R}}^{\mathrm{op}}
  \,,
$$
so a derived manifold is the formal dual to
a differential graded-commutative algebra in non-positive degrees, whose underlying graded algebra is
of the form $\wedge^\bullet_{C^\infty(X)}(\Gamma(V^\ast))$ for $V$ a $-\mathbb{N}$-graded smooth vector
bundle over $X$. In the physics literature these dg-algebras are known as \emph{BV-complexes}.

For example, for $X$ a smooth manifold and $S \in C^\infty(X)$ a smooth function on it, then the vanishing
locus of $S$ in $X$ is represented by the derived manifold $\mathrm{ker}_d(S)$ that is formally dual to the dg-algebra
denoted $C^\infty(\mathrm{ker}_d(S))$ which is spanned over $C^\infty(X)$ by a single generator $t$ of degree -1 and whose
differential (linear over $\mathbb{R}$) is defined by
$$
  d_{\mathrm{BV}} : t \mapsto S
  \,.
$$
For $\Sigma$ an ordinary smooth manifold, then morphisms $\Sigma \longrightarrow \mathrm{ker}_d(\Sigma)$
are equivalently dg-algebra homomorphisms $C^\infty(\Sigma) \longleftarrow C^\infty(\mathrm{ker}_d(\Sigma))$,
and these are equivalently algebra homomorphisms $\phi^\ast : C^\infty(\Sigma)\longleftarrow C^\infty(X)$ such that
$\phi^\ast S = 0$. These, finally, are equivalently smooth functions $\phi : \Sigma \longrightarrow X$ that land
everywhere in the 0-locus of $S$. It is in this way that $\mathrm{ker}_d(S)$ is a resolution of the
possibly singular vanishing locus by a complex of non-singular smooth bundles.

Notice that even if the kernel of $S$ does exist as a smooth submanifold $\mathrm{ker}(S) \hookrightarrow $ it
need not be equivalent to the derived kernel: for instance over $X = \mathbb{R}^1$ with its canonical coordinate
function $x$, then $ker(x) = \{0\}$ but $\mathrm{ker}_d(x^2) \simeq \mathbb{D}_0^{(1)}$ is the infinitesimal
interval around $0$.

Given a derived manifold $X_d$ this way, then for each $k \in \mathbb{N}$ the differential $k$-forms
on $X_d$ also inherit the BV-differential, on top of the de Rham differential. We write $\Omega^{k;-s}(X_d)$
to indicate the differential $k$-forms of BV-degree $-s$. So in particular the 0-forms recover the
BV dg-algebra itself $\Omega^{0,-\bullet}(X_d) = C^\infty(X_d)$.

Hence using underived manifolds, then the conservation of the presymplectic current, $d_H \Omega = 0$, implies that
over a spacetime/worldvolume $\Sigma$ with two boundary components $\Sigma_{\mathrm{in}} = \partial_{\mathrm{in}}\Sigma$
and $\Sigma_{\mathrm{out}} = \partial_{\mathrm{out}}\Sigma$ then
the canonical pre-symplectic forms $\omega_{\mathrm{in}}$ and $\omega_{\mathrm{out}}$ agree
$$
  \raisebox{20pt}{
  \xymatrix{
    & [\Sigma, \mathcal{E}]_{\Sigma}
    \ar[dl]_{\pi_{\mathrm{in}}}
    \ar[dr]_{\ }="s"^{\pi_{\mathrm{out}}}
    \\
    [N^\infty_{\Sigma}\Sigma_{\mathrm{in}},\mathcal{E}]_\Sigma
    \ar[dr]_{\omega_{\mathrm{in}}}^{\ }="t"
    &&
    [N^\infty_{\Sigma}\Sigma_{\mathrm{out}},\mathcal{E}]_\Sigma
    \ar[dl]^{\omega_{\mathrm{out}}}
    \\
    & \mathbf{\Omega}^{2}
    \ar@{=} "s"; "t"
  }
  }
  \;\;\;\;\;\;
  \pi^\ast_{\mathrm{out}} \omega_{\mathrm{out}} -  \pi^\ast_{\mathrm{in}}\omega_{\mathrm{in}} = 0
$$

When the covariant phase space is resolved by a derived space $([\Sigma,\mathcal{E}]_\Sigma)_d$, then
this equation becomes a homotopy which asserts the existence of a 2-form $\omega_{\mathrm{BV}}$ of BV-degree
-1 which witnesses the invariance of the canonical presymplectic form:
$$
  \raisebox{20pt}{
  \xymatrix{
    & ([\Sigma, \mathcal{E}]_{\Sigma})_d
    \ar[dl]_{\pi_{\mathrm{in}}}
    \ar[dr]_{\ }="s"^{\pi_{\mathrm{out}}}
    \\
    [N^\infty_{\Sigma}\Sigma_{\mathrm{in}},\mathcal{E}]_\Sigma
    \ar[dr]_{\omega_{\mathrm{in}}}^{\ }="t"
    &&
    [N^\infty_{\Sigma}\Sigma_{\mathrm{out}},\mathcal{E}]_\Sigma
    \ar[dl]^{\omega_{\mathrm{out}}}
    \\
    & \mathbf{\Omega}^{2;-\bullet}
    \ar@{=>}^{\omega_{\mathrm{BV}}} "s"; "t"
  }
  }
  \;\;\;\;\;\;
  \pi^\ast_{\mathrm{out}} \omega_{\mathrm{out}} -  \pi^\ast_{\mathrm{in}}\omega_{\mathrm{in}} = d_{\mathrm{BV}}\omega_{\mathrm{BV}}
  \,.
$$

The equation on the right appears in the BV-liteature as \cite[equation (9)]{CattaneoMnevReshetikhin12}).

For the purpose of prequantum field theory, we again wish to de-transgress this phenomenon. Instead of
just modelling the covariant phase space by a derived space, we should model the dynamical shell
$\mathcal{E} \hookrightarrow J^\infty_\Sigma E$ itself by a derived bundle.

The derived shell $\mathrm{ker}_d(\mathrm{EL})$
is the derived manifold bundle over $\Sigma$ whose underlying manifold is $J^\infty_\Sigma E$
and whose bundle of antifields is the pullback of $V^\ast E \otimes \wedge^{p+1} T^\ast \Sigma$
to the jet bundle (along the projection maps to $E$).

If $\phi^i$ are a choice of local vertical coordinates on $E$ (the fields) and $\phi_i^\ast$ denotes the corresponding
local antifield coordinates with respect to any chosen volume form on $\Sigma$, then this BV-differential looks like
$$
  d_{\mathrm{BV}} =  \mathrm{EL}_i  \frac{\partial}{\partial \phi^\ast_i}
  \;\; :\;\; \phi_i^\ast \mapsto \mathrm{EL}_i
  \,.
$$
When regarded as an odd graded vector field, this differential is traditionally denoted by $Q$.

In such coordinates there is then the following canonical differential form
$$
  \Omega_{\mathrm{BV}} = d_V \phi_i^\ast \wedge d_V \phi^i \;\;\; \in \Omega^{p+1,2;-1}(\mathrm{ker}_d(\mathrm{EL}))
$$
which, as indicated, is of BV-degree -1 and otherwise is a $(p+3)$-form with horizontal degree $p+1$ vertical
degree 2. More abstractly, this form is characterized by the property that
$$
  \iota_{Q}\Omega_{\mathrm{BV}} = \mathrm{EL} \;\;\;  \in \Omega^{p+1,1;0}(\mathrm{ker}_d(\mathrm{EL}))
  \,.
$$

As before, we write
$$
  \omega_{\mathrm{BV}} := \int_\Sigma \Omega_{\mathrm{BV}}
$$
for the transgression of this form to the covariant phase space. We now claim that there it satisfies
the above relation of witnessing he conservation of the presymplectic current up to BV-exact terms
\footnote{
This was first pointed out by us informally on the $n$Lab in October 2011
{\url{http://ncatlab.org/nlab/revision/diff/phase+space/29}}.}
In fact it satisfies the following stronger relation
\begin{equation}
  \label{differentialofLagrangianintransgressedBVincarnation}
  \iota_Q \omega_{\mathrm{BV}} = d S + \pi^\ast \theta
\end{equation}
which turns out to be the transgressed and BV-theoretic version of the fundamental variational
equation \ref{differentialofLagrangian}:
$$
  \begin{aligned}
    d S & = d \int_\Sigma L
    \\
    & = \int_\Sigma d L
    \\
    & = \int_\Sigma ( \mathrm{EL} -d_H \Theta )
    \\
    & = \int_\Sigma (\iota_Q \Omega_{\mathrm{BV}} - d_H \Theta )
    \\
    & = \iota_Q \omega_{\mathrm{BV}} - \pi^\ast \theta \,.
  \end{aligned}
$$
Equation \ref{differentialofLagrangianintransgressedBVincarnation} has been postulated
as the fundamental compatibility condition for BV-theory on spacetimes $\Sigma$s with boundary in
\cite[equation (7)]{CattaneoMnevReshetikhin12}. Applying $d$ to both sides of this equation
recovers the previous $d_{\mathrm{BV}} \omega_{\mathrm{BV}} = \pi^\ast \omega$.

Notice that equation \ref{differentialofLagrangianintransgressedBVincarnation} may be read as saying that
the action functional is a Hamiltonian, not for the ordinary presymplectic structure, but for the
BV-symplectic structure.

\medskip

\begin{center}
\begin{tabular}{|c|c|}
  \hline
  \begin{tabular}{c}
    {\bf concept in}
    \\
    {\bf classical field theory}
  \end{tabular}
  &
  \begin{tabular}{c}
    {\bf local model in}
    \\
    {\bf in BV-BRST formalism}
  \end{tabular}
  \\
  \hline
  \begin{tabular}{c}
    $[N^\infty_\Sigma\Sigma_p,\mathcal{E}]_\Sigma$
    \\
    phase space
  \end{tabular}  &
  \begin{tabular}{c}
    $d_{\mathrm{BV}}$
    \\
    BV-complex
    \\
    of anti-fields
  \end{tabular}
  \\
  \hline
  \begin{tabular}{c}
    $\omega_{\mathrm{in}} = \omega_{\mathrm{out}}$
    \\
    independence of presymplectic form
    \\
    from choice of Cauchy surface
  \end{tabular}
  &
  \begin{tabular}{c}
    $\xymatrix{\omega_{\mathrm{in}} \ar[rr]|{\simeq}^{\omega_{\mathrm{BV}}} && \omega_{\mathrm{out}} }$
    \\
    coboundary by
    \\
    BV-bracket
  \end{tabular}
  \\
  \hline
  \begin{tabular}{c}
    $[N^\infty_\Sigma \Sigma_p, \mathcal{E}]_\Sigma
      \longrightarrow
     [\Sigma, \mathbf{B}G_{\mathrm{conn}}]
    $
    \\
    smooth groupoid
    \\
    of gauge fields and gauge transformations
  \end{tabular}
  &
  \begin{tabular}{c}
    $d_{\mathrm{BRST}}$
    \\
    BRST complex
    \\
    of ghost fields
  \end{tabular}
  \\
  \hline
  \begin{tabular}{c}
    $[\Sigma_p, \mathcal{E}]_\Sigma  \longrightarrow [\Sigma, \mathbf{B}^k U(1)_{\mathrm{conn}}]$
    \\
    higher smooth groupoid
    \\
    of higher gauge fields
    \\
    and higher gauge transformations
  \end{tabular}
  &
  \begin{tabular}{c}
    $d_{\mathrm{BRST}}$
    \\
    BRST complex
    \\
    of higher order ghost-of-ghost fields
  \end{tabular}
  \\
  \hline
\end{tabular}
\end{center}

%%%%%%%%%%%%%%%%%%%%%%%%%%%%%%%%%%%%%%%%%%%%%%%%%%%%%%%%%%%%%%%%%%%%%%%%
\subsection{Sigma-model field theories}
\label{introductionsigmamodelfieldtheories}
%%%%%%%%%%%%%%%%%%%%%%%%%%%%%%%%%%%%%%%%%%%%%%%%%%%%%%%%%%%%%%%%%%%%%%%%

A \emph{sigma-model} is a field theory whose field bundle (as in section \ref{PrequantumLocalFieldTheoryInMotivation}) is of the simple form
$$
  \xymatrix{
    \Sigma \times X
    \ar[d]^{p_1}
    \\
    \Sigma
  }
$$
for some space $X$.
This means that in this case field configurations, which by definition are sections of the field bundle, are equivalently maps of the form
$$
  \phi : \Sigma \longrightarrow X
  \,.
$$
One naturally thinks of such a map as being a $\Sigma$-shaped trajectory of a $p$-dimensional
object (a $p$-brane) on the space $X$. Hence $X$ is called the \emph{target space} of the model.
Specifically, if this models
$\Sigma$-shaped trajectories of $p$-dimensional relativistic branes, then
$X$ is the \emph{target spacetime}. There are also famous examples of sigma-models where $X$ is a more abstract
space, usually some moduli space of certain scalar fields of a field theory that is itself defined on
spacetime. Historically the first sigma-models were of this kind. In fact in the first examples $X$ was
a linear space. For emphasis that this is not assumed one sometimes speaks of \emph{non-linear sigma models}
for the sigma-models that we consider here. In fact we consider examples where $X$ is not even a manifold, but a
smooth $\infty$-groupoid, a higher moduli stack.

Given a target space $X$, then every $(p+1)$-form $A_{p+1} \in \Omega^{p+1}(X)$ on $X$
induces a local Lagrangian for sigma-model field theories with target $X$:
we may simply pull back that form to the jet bundle $J^\infty_\Sigma(\Sigma \times X)$ and project out its
horizontal component. Lagrangians that arise this way are known as \emph{topological terms}.

The archetypical example of a sigma-model with topological term is that for describing the electron propagating in
a spacetime and subject to the background forces of gravity and of electromagnetism. In this case $p = 0$
(a point particle, hence a ``0-brane''), $\Sigma$  is the interval $[0,1]$
or the circle $S^1$, regarded as the abstract \emph{worldline} of an electron.
Target space $X$ is a spacetime manifold equipped with a pseudo-Riemannian metric $g$
(modelling the background field of gravity) and with a vector potential 1-form $A\in \Omega^1(X)$
whose differential is the Faraday tensor $F = d A$ (modelling the electromagnetic background field).
The local Lagrangian is
$$
  L = L_{\mathrm{kin}} + \underset{L_{\mathrm{int}}}{\underbrace{q (A_\Sigma)_H}} \;\; \in \Omega^{p+1}_H(J^\infty_\Sigma(\Sigma \times X))
  \,,
$$
where $L_{\mathrm{kin}}$ is the standard kinetic Lagrangian for (relativistic) point particles,
$q$ is some constant, the electric charge of the electron, and $(A_\Sigma)_H$ is the horizontal component
of the pullback of $A$ to the jet bundle. The variation of $L_{\mathrm{int}}$ yields the Lorentz force
that the charged electron experiences.

Now, as in the the discussion in section \ref{PrequantumLocalFieldTheoryInMotivation}, in general the Faraday tensor
$F$ is not globally exact, and hence in general there does not exist a globally defined such 1-form on the jet bundle.
But via the sigma-model construction, the prequantization of the
worldline field theory of the electron on its jet bundle is naturally induced by a Dirac charge quantization of
its background electromagnetic field on target spacetime: given
$$
  \xymatrix{
    & \mathbf{B} U(1)_{\mathrm{conn}} \ar[d]
    \\
    X \ar[ur]^{\nabla} \ar[r]_-F & \mathbf{\Omega}^2_{\mathrm{cl}}
  }
$$
a circle-principal connection on target spacetime for the given field strength Faraday tensor $F$
(hence with local ``vector potential'' 1-forms $\{A_i\}$ with respect to some
cover $\{U_i \to X\}$), then the horizontal projection $(\nabla_\Sigma)_H$ of the pullback
of the whole circle-bundle with connection to the jet bundle
constitutes a prequantum field theory in the sense of sections \ref{globalactionfunctionalinintroduction}.
Similarly, the background electromagnetic field $\nabla$ also
serves to prequantize the covariant phase space of the electron, according to
section \ref{elgerbesinintroduction}. This is related to the familiar statement
that in the presence of a magnetic background field the spatial coordinates of the electron no longer
Poisson-commute with each other.

This prequantization of sigma-models via $(p+1)$-form connections on target space works generally:
we obtain examples of prequantum field theories of sigma-model type by adding to
a globally defined kinetic Lagrangian form a prequantum \emph{topological term} given
by the pullback of a $(p+1)$-form connection on target space. The pullback of that target
$(p+1)$-form connection to target space serves to prequantize the entire field theory in all codimensions

\begin{center}
\begin{tabular}{|c|rcl|}
  \hline
  \multicolumn{4}{|c|}{\bf prequantum sigma-model topological terms}
  \\
  \hline
  \hline
  background field & $\nabla :$ & $X$ & $\longrightarrow \mathbf{B}^{p+1}U(1)_{\mathrm{conn}}$
  \\
  \hline
  \begin{tabular}{c}
    prequantum
    \\
    Lagrangian
  \end{tabular}
  & $(\nabla_\Sigma)_H :$ & $\Sigma \times X$ & $\longrightarrow \mathbf{B}^{p+1}_H U(1)_{\mathrm{conn}}$
  \\
  \hline
  \begin{tabular}{c}
    prequantized
    \\
    phase space
  \end{tabular}
  &
  $(\nabla_\Sigma)_L :$ & $\mathcal{E}$ & $\longrightarrow \mathbf{B}^{p+1}_L U(1)_{\mathrm{conn}}$
  \\
  \hline
\end{tabular}
\end{center}

While sigma-models with topological terms are just a special class among all variational field theories,
in the context of higher differential geometry this class is considerably larger than in traditional differential geometry.
Namely we may regard any of the moduli stacks $\mathbf{A}_{\mathrm{conn}}$ of gauge fields
that we discuss in section \ref{introductiongaugefields} as target space, i.e. we may consider higher stacky field bundles
of the form
$$
  \raisebox{20pt}{
  \xymatrix{
    \Sigma \times\mathbf{A}_{\mathrm{conn}}
    \ar[d]
    \\
    \Sigma
  }
  }
  \,.
$$
Everything goes through as before, in particular a field configuration now is a map $\Sigma \longrightarrow \mathbf{A}_{\mathrm{conn}}$
from worldvolume/spacetime $\Sigma$ to this moduli stack.
But by the discussion above in section \ref{introductiongaugefields}, such maps now
are equivalent to gauge fields on $\Sigma$. These are, of course, the field configurations of \emph{gauge theories}.
Hence, in higher differential geometry, the concepts of sigma-model field theories and of gauge field theories
are unified.

In particular both concepts may mix. Indeed, we find below that higher dimensional Wess-Zumino-Witten-type models generally
are ``higher gauged'', this means that their field configurations are a pair consisting of
a map $\phi : \Sigma \to X$ to some target spacetime $X$, together with a $\phi$-twisted higher gauge field
on $\Sigma$.

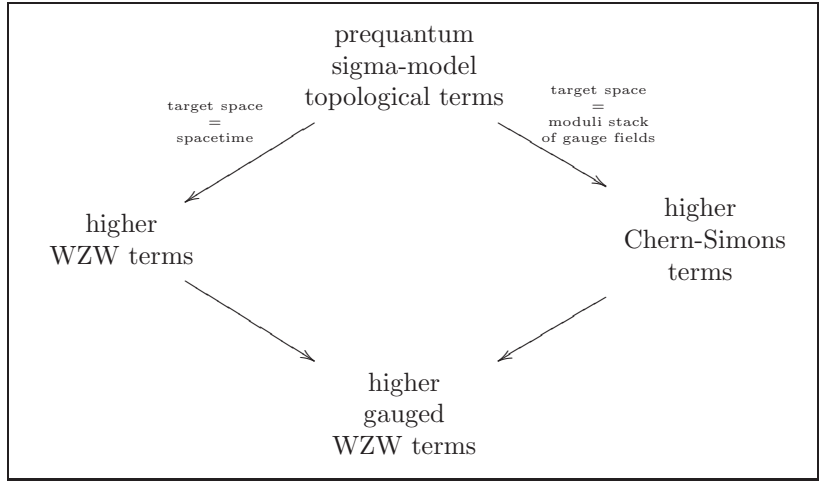
$$
  \fbox{
  \xymatrix{
    & \mbox{\begin{tabular}{c} prequantum \\ sigma-model \\ topological terms  \end{tabular}}
    \ar[dl]_{ \mbox{ \tiny \begin{tabular}{c} target space \\ $=$ \\ spacetime \end{tabular} } }
    \ar[dr]^{ \mbox{\tiny  \begin{tabular}{c} target space \\ $=$ \\ moduli stack \\ of gauge fields \end{tabular}  } }
    \\
    \mbox{\begin{tabular}{c}higher \\ WZW terms\end{tabular}}
    \ar[dr]
    &&
    \mbox{\begin{tabular}{c}higher \\\ Chern-Simons \\ terms\end{tabular}}
    \ar[dl]
    \\
    & \mbox{\begin{tabular}{c}higher \\ gauged \\ WZW terms\end{tabular}}
  }
  }
$$
Examples of a (higher) gauged WZW-type sigma model are the Green-Schwarz-type sigma-models of those
super $p$-branes on which other branes may end. This includes the D-branes and the M5-brane. The
former are gauged by a 1-form gauge field (the ``Chan-Paton gauge field'') while the latter is gauged
by a 2-form gauge field. We say more about these examples below in \ref{introductionWZWmodels}.

\medskip

We may construct examples of prequantized topological terms from functoriality of the Lie integration process
that already gave the (higher) gauge fields themselves in section \ref{introductiongaugefields}.
There we saw that a $(p+2)$-cocycle on
an $L_\infty$-algebroid is a homomorphism of $L_\infty$-algebroids of the form
$$
  \mu : \mathfrak{a} \longrightarrow \mathbf{B}^{p+2}\mathbb{R}
  \,.
$$
Moreover, the $\exp(-)$-construction which sends $L_\infty$-algebroids to simplicial presheaves
representing universal higher moduli stacks of $\mathfrak{a}$-valued gauge fields is clearly functorial,
hence it sends this cocycle to a morphism of simplicial presheaves of the form
$$
  \exp(\mu) : \exp(\mathfrak{a}) \longrightarrow \mathbf{B}^{p+2}\mathbb{R}
  \,.
$$
One finds that this descends  to the $(p+2)$-coskeleton $\mathbf{A} := \mathrm{cosk}_{p+2}\exp(\mathfrak{a})$
after quotienting out the subgroup $\Gamma \hookrightarrow \mathbb{R}$ of periods of $\mu$
\cite{FSS} (just as in the prequantization of the global action functional in section \ref{globalactionfunctionalinintroduction}):
$$
  \raisebox{20pt}{
  \xymatrix{
    \exp(\mathfrak{a}) \ar[d]^{\eta^{\mathrm{cosk}_{p+2}}} \ar[d] \ar[rr]^{\exp(\mu)} && \mathbf{B}^{p+2}\mathbb{R} \ar[d]
    \\
    \mathbf{A} \ar[rr]^-{\mathbf{c}} && \mathbf{B}^{p+2}(\mathbb{R}/\Gamma)
  }
  }
  \,.
$$
To get a feeling for what the resulting morphism $\mathbf{c}$ is, consider the case that $\mathbf{A} = \mathbf{B}G$ for some
group $G$. There is a geometric realization operation $\pi_\infty$ which sends smooth $\infty$-groupoids
to plain homotopy types (homotopy types of topological spaces). Under this operation
a map $\mathbf{c}$ as above becomes a map $c$ of the form
$$
  \raisebox{20pt}{
  \raisebox{20pt}{
  \xymatrix{
    \mathbf{B}G \ar[d]^{\eta^{\pi_\infty}} \ar[rr]^-{\mathbf{c}} && \mathbf{B}^{p+2}(\mathbb{R}/\mathbb{Z}) \ar[d]^{\eta^{\pi_\infty}}
    \\
    B G \ar[rr] \ar[rr]^-c && K(\mathbb{Z},p+3)
  }
  }
  }
  \,,
$$
where $B G$ is the traditional classifying space of a (simplicial) topological group $G$, and where
$K(\mathbb{Z},p+3) = B^{p+3}\mathbb{Z}$ is the Eilenberg-MacLane space that classifies integral cohomology
in degree $(p+3)$. What $B G$ classifies are $G$-principal bundles, and hence for each space $\Sigma$ the map
$c$ turns into a \emph{characteristic class} of equivalence classes of $G$-principal bundles:
$$
  c_\Sigma : G \mathrm{Bund}(\Sigma)_\sim \longrightarrow H^{p+3}(\Sigma,\mathbb{Z})
  \,.
$$
Hence $c$ itself is a \emph{universal characteristic class}. Accordingly, $\mathbf{c}$ is
a refinement of $c$ that knows about gauge transformations: it sends smooth $G$-bundles with smooth gauge
transformations and gauge-of-gauge transformations between these to integral cocycles and coboundaries
and coboundaries-between-coboundaries between these.

Equivalently, we may think of $\mathbf{c}$ as classifying a $(p+1)$-gerbe on the universal moduli stack
of $G$-principal bundles. This is equivalently its homotopy fiber (in direct analogy with the infinitesimal
version of this statement above in section \ref{introductionhighergaugefields}) fitting into a long homotopy fiber
sequence of the form
$$
  \xymatrix{
    \mathbf{B}^{p+1}(\mathbb{R}/\mathbb{Z}) \ar[r] & \mathbf{B}\hat G \ar[d]
    \\
    & \mathbf{B}G
    \ar[rr]^-{\mathbf{c}}
    &&
    \mathbf{B}^{p+2}(\mathbb{R}/\mathbb{Z})
  }
  \,.
$$
Yet another equivalent perspective is that this defines an $\infty$-group extension $\hat G$
of the $\infty$-group $G$ by the $\infty$-group $\mathbf{B}^p (\mathbb{R}/\mathbb{Z})$.

So far all this is without connection data, so far these are just higher instanton sectors without
any actual gauge fields inhabiting these instanton sectors. We now add connection data to the
situation.

\noindent Adding connection data to $\mathbf{c}$ regarded as a higher prequantum bundle on
the moduli stack $\mathbf{B}G$
yields
\begin{itemize}
 \item \ref{introductionchernsimonstypefieldtheories} -- Chern-Simons-type prequantum field theory.
\end{itemize}
\noindent Adding instead connection data to $\mathbf{c}$ regarded as a higher group extension yields
\begin{itemize}
  \item \ref{introductionWZWmodels} -- Wess-Zumino-Witten-type prequantum field theories.
\end{itemize}

%%%%%%%%%%%%%%%%%%%%%%%%%%%%%%%%%%%%%%%%%%%%%%%%%%%%%%%%%%%
\subsection{Chern-Simons type field theory}
\label{introductionchernsimonstypefieldtheories}
%%%%%%%%%%%%%%%%%%%%%%%%%%%%%%%%%%%%%%%%%%%%%%%%%%%%%%%%%%%

For $\mathfrak{g}$ a semisimple Lie algebra with Killing form invariant polynomial $\langle -,-\rangle$,
classical 3-dimensional Chern-Simons theory \cite{FreedCS} has
as fields the space of $\mathfrak{g}$-valued differential 1-forms $A$,
and the Lagrangian is the Chern-Simons 3-form
$$
  L_{\mathrm{CS}}(A) = \mathrm{CS}(A) := \langle A \wedge d A \rangle - \tfrac{1}{3} \langle A \wedge [A\wedge A]\rangle
  \,.
$$
This Chern-Simons form is characterized by two properties: for vanishing curvature it reduces to the value of the 3-cocycle
$\langle -,[-,-]\rangle$ on the connection 1-form $A$, and its differential is the value of the invariant polynomial $\langle -,-\rangle$ on
the curvature 2-form $F_A$.

There is a slick way to express this in terms of the dg-algebraic description from section \ref{introductionhighergaugefields}: there is an element
$\mathrm{cs} \in \mathrm{W}(\mathbf{B}\mathfrak{g})$, which in terms of the chosen basis $\{t^a\}$
for $\wedge^1\mathfrak{g}^\ast$ is given by
$$
  \mathrm{cs} : k_{a b} (d_W t^a)\wedge t^b - \tfrac{1}{3} k_{a a'} C^{a'}{}_{b c} t^a \wedge t^b \wedge t^c
  \,.
$$
Hence equivalently this is a dg-homomorphism of the form
$$
  \mathrm{W}(\mathbf{B}\mathfrak{g})
  \stackrel{\mathrm{cs}}{\longleftarrow}
  \mathrm{W}(\mathbf{B}^3\mathbb{R})
$$
and for $A \in \Omega^1(\Sigma,\mathfrak{g}) = \{ \; \Omega^\bullet(\Sigma) \longleftarrow \mathrm{W}(\mathbf{B}\mathfrak{g}) \; \}$
then the Chern-Simons form of $A$ is the composite
$$
  \Omega^\bullet(\Sigma)
  \stackrel{A}{\longleftarrow}
  \mathrm{W}(\mathbf{B}\mathfrak{g})
  \stackrel{\mathrm{cs}_3}{\longleftarrow}
  :
  \mathrm{CS}(A)
  \,.
$$
Now, the two characterizing properties satisfied by the Chern-Simons equivalently mean in terms
of dg-algebra that the map $\mathrm{cs}$ makes the following two squares commute:
$$
  \xymatrix{
    \mathrm{CE}(\mathbf{B}\mathfrak{g})
    \ar@{<-}[rr]^-{\langle -,[-,-]\rangle}
    &&
    \mathrm{CE}(\mathbf{B}^{3}\mathbb{R})
    \\
    \mathrm{W}(\mathbf{B}\mathfrak{g})
    \ar[u]
    \ar@{<-}[rr]^-{\mathrm{cs}}
    &&
    \mathrm{W}(\mathbf{B}^{3}\mathbb{R})
    \ar[u]
    \\
    \mathrm{inv}(\mathbf{B}\mathfrak{g})
    \ar[u]
    \ar@{<-}[rr]^-{\langle -,-\rangle}
    &&
    \mathrm{inv}(\mathbf{B}^3 \mathbb{R})
    \ar[u]
  }
$$
This shows how to prequantize 3d Chern-Simons theory in codimension 3: the vertical sequences
appearing here are just the Lie algebraic data to which we apply differential Lie integration,
as in section \ref{introductiongaugefields},
to obtain the moduli stacks of $G$-connections and of 3-form connections, resepctively. Moreover,
by the discussion at the end of section \ref{introductionsigmamodelfieldtheories}
and using that $\langle -,[-,-]\rangle$ represents an integral cohomology class on $G$ we get a map
$$
  (\mathbf{c}_2)_{\mathrm{conn}}
  :=
  \exp(\mathrm{cs})
  \;:\;
  \mathbf{B}G_{\mathrm{conn}}
  \longrightarrow
  \mathbf{B}^3(\mathbb{R}/\mathbb{Z})_{\mathrm{conn}}
  \,.
$$
This is the background 3-connection which induces prequantum Chern-Simons field theory
by the general procedure indicated in section \ref{introductionsigmamodelfieldtheories}.

Notice that this map is a refinement of the traditional Chern-Weil homomorphism.
This allows for instance to prequantize the Green-Schwarz anomaly cancellation condition
heterotic strings: the higher moduli stack of GS-anomaly free gauge fields is the homotopy
fiber product of the prequantum Chern-Simons Lagrangians for the simple groups $\mathrm{Spin}$
and $\mathrm{SU}$ \cite{SSSIII}.

\medskip

This higher Lie theoretic formulation of prequantum 3-Chern-Simons theory now immediately generalizes to
produce higher (and lower) dimensional prequantum $L_\infty$-algebroid Chern-Simons theories.

For $\mathfrak{a}$ any $L_\infty$-algebroid as in section \ref{introductionhighergaugefields},
we say that a $(p+2)$-cocycle $\mu$ on $\mathfrak{a}$
is in transgression with an invariant polynomial $\langle -\rangle$ on $\mathfrak{a}$
if there is an element $\mathrm{cs} \in W(\mathfrak{a})$ such that $d_W \mathrm{cs} = \langle -\rangle $
and $\mathrm{cs}|_{\mathrm{CE}} = \mu$. Equivalently this means that $\mathrm{cs}$ fits into a diagram
of dg-algebras of the form
$$
  \xymatrix{
    \mathrm{CE}(\mathfrak{a})
    \ar@{<-}[rr]^-\mu
    &&
    \mathrm{CE}(\mathbf{B}^{p+2}\mathbb{R})
    \\
    \mathrm{W}(\mathfrak{a})
    \ar[u]
    \ar@{<-}[rr]^-{\mathrm{cs}}
    &&
    \mathrm{W}(\mathbf{B}^{p+2}\mathbb{R})
    \ar[u]
    \\
    \mathrm{inv}(\mathfrak{a})
    \ar[u]
    \ar@{<-}[rr]^-{\langle - \rangle}
    &&
    \mathrm{inv}(\mathbf{B}^{p+2}\mathbb{R})
    \ar[u]
  }
$$
Applying $\exp(-)$ to this, this induces maps of smooth moduli stacks of the form
$$
  \mathbf{c}_{\mathrm{conn}} :  \mathbf{A}_{\mathrm{conn}} \longrightarrow \mathbf{B}^{p+2}(\mathbb{R}/\Gamma)_{\mathrm{conn}}
  \,.
$$
This gives a prequantum Chern-Simons-type field theory whose field configurations locally are
$\mathfrak{a}$-valued differential forms, and whose Lagrangian is locally the Chern-Simons element
$\mathrm{cs}$ evaluated on these forms.

For instance if $(\mathfrak{a},\omega)$ is a symplectic Lie $p$-algebroid, then we obtain the
prequantization of $(p+1)$-dimensional AKSZ-type field theories \cite{frs}. For $p = 1$ this
subsumes the topological string A- and B-model \cite{AKSZ}. Generally, the prequantum moduli stack of
fields for 2-dimensional prequantum AKSZ theory is a differential refinement of the symplectic groupoid
of a given Poisson manifold \cite{Bongers}.
The Poisson manifold canonically defines a boundary condition for the
corresponding prequantum 2d Poisson-Chern-Simons theory, and the higher geometric
boundary quantization of this 2d prequantum theory reproduces ordinary Kostant-Souriau geometric quantization of the
symplectic leafs \cite{Nuiten}. This is a non-perturbative improvement of the perturbative algebraic deformation
quantization of the Poisson manifold as the boundary of the perturbative 2d AKSZ field theory
due to \cite{CaFe}.

Generally one expects to find non-topological non-perturbative  $p$-dimensional quantum field theories arising this way as the
higher geometric boundary quantization of
$(p+1)$-dimensional prequantum Chern-Simons type field theories \cite{talkEdingburgh, motquant}.

For instance for $(\mathbf{B}^3 \mathbb{R}, \omega)$ the line Lie 3-algebra equipped with its canonical
binary invariant polynomial, the corresponding prequantum Chern-Simons type field theory is 7-dimensional abelian  cup-product
Chern-Simons theory \cite{FiorenzaSatiSchreiberIV}. This has been argued to induce on its boundary the
conformal 6-dimensional field theory of a self-dual 2-form field \cite{Witten96} \cite{HopkinsSinger}.
This 7-dimensional Chern-Simons theory is one summand in the Chern-Simons term of 11-dimensional supergravity
compactified on a 4-sphere. The $\mathrm{AdS}_7/\mathrm{CFT}_6$ correspondence predicts that this carries on its boundary the refinement of the
self-dual 2-form to a 6-dimensional superconformal field theory. There are also nonabelian summands in this 7d Chern-Simons term.
For instance for $(\mathbf{B}\mathfrak{string}_{\mathfrak{g}},\langle -,-,-,-\rangle)$ the string Lie 2-algebra
equipped with its canonical degree-4 invariant polynomial, then the resulting prequantum field theory
is 7-dimensional Chern-Simons field theory on String 2-connection fields \cite{FiorenzaSatiSchreiberI}.

For more exposition of prequantum Chern-Simons-type field theories see also \cite{FiorenzaSatiSchreiberCS}.

%%%%%%%%%%%%%%%%%%%%%%%%%%%%%%%%%%%%%%%%%%%%%%%%%%%%%%%%
\subsection{Wess-Zumino-Witten type field theory}
\label{introductionWZWfieldtheory}
\label{introductionWZWmodels}
%%%%%%%%%%%%%%%%%%%%%%%%%%%%%%%%%%%%%%%%%%%%%%%%%%%%%%%%

The traditional Wess-Zumino-Witten (WZW) field theory \cite{Ga, Gawedzki} for a semisimple, simply-connected compact Lie group $G$
is a 2-dimensional sigma-model
with target space $G$, in the sense of section \ref{introductionsigmamodelfieldtheories},
given by a canonical kinetic term, and with topological term that is locally a potential
for the left-invariant 3-form $\langle \theta \wedge [\theta \wedge \theta]\rangle \in \Omega^3(G)_{\mathrm{cl}\atop \mathrm{li}}$,
where $\theta$ is the Maurer-Cartan form on $G$. This means that for $\{U_i \to G\}$ a cover of $G$ by coordinate charts
$U_i \simeq \mathbb{R}^n$, then the classical WZW model is the locally variational classical field theory
(in the sense discussed in section \ref{Principleofextremalactioncomonadically}) whose local Lagrangian $L_i$
is (in the notation introduced above in section \ref{introductionsigmamodelfieldtheories}) $L_i = (L_{\mathrm{kin}})_i + ((B_i)_\Sigma)_H$
for $B_i \in \Omega^2(U_i)$ a 2-form such that $d B_i = \langle \theta\wedge [\theta\wedge \theta]\rangle|_{U_i}$.

By the discussion in section \ref{introductionsigmamodelfieldtheories}, in order to prequantize this field
theory it is sufficient that we construct a $U(1)$-gerbe on $G$ whose curvature 3-form is $\langle \theta \wedge [\theta \wedge\theta]\rangle$.
In fact we may ask for a little more: we ask for the gerbe to be \emph{multiplicative} in that it carries 2-group
structure that covers the group structure on $G$, hence that it is given by the 2-group extension classified by
the smooth universal class $\mathbf{c} : \mathbf{B}G \longrightarrow \mathbf{B}^3 U(1)$.

\medskip

An elegant construction of this prequantization, which will set the scene for the general construction of higher
WZW models, proceeds by making use of a universal property of the differential coefficients.
Namely one finds that for all $p \in \mathbb{Z}$, then the moduli stack $\mathbf{B}^{p+1}(\mathbb{R}/\mathbb{Z})_{\mathrm{conn}}$
of $(p+1)$-form connections is the homotopy fiber product of $\mathbf{B}^{p+1}(\mathbb{R}/\mathbb{Z})$ with
$\mathbf{\Omega}^{p+2}_{\mathrm{cl}}$ over $\flat_{\mathrm{dR}}\mathbf{B}^{p+2}\mathbb{R}$.
$$
  \raisebox{20pt}{
  \xymatrix{
    & \mathbf{\Omega}^{p+2}_{\mathrm{cl}}
    \ar[dr]
    \\
    \mathbf{B}^{p+1}(\mathbb{R}/\mathbb{Z})_{\mathrm{conn}}
    \ar@{}[rr]|{\mathrm{(pb)}}
    \ar[ur]
    \ar[dr]
    &&
    \mathbf{B}^{p+2}\flat \mathbb{R}
    \\
    & \mathbf{B}^{p+2}\mathbb{Z}
    \ar[ur]
  }
  }
  \,.
$$
Here ``$\flat$'' indicates the discrete underlying group, and hence this homotopy pullback says that giving a $(p+1)$-form
connection is equivalent to giving an integral $(p+2)$-class and a closed $(p+2)$-form together with a homotopy
the identifies the two as cocycles in real cohomology.

In view of this, consider the following classical Lie theoretic data associated with the semisimple Lie
algebra $\mathfrak{g}$.
\begin{center}
\begin{tabular}{|l|l|}
\hline
$\mathfrak{g}$ & semisimple Lie algebra
\\
\hline
$G$ & its simply-connected Lie group
\\
\hline
$\theta \in \Omega^1(G,\mathfrak{g})$ & Maurer-Cartan form
\\
\hline
$\langle -,-\rangle$ & Killing metric
\\
\hline
$\mu_3 = \langle -,[-,-]\rangle$ & Lie algebra 3-cocycle
\\
\hline
$k \in H^3(G,\mathbb{Z})$ & level
\\
\hline
$ \mu_3(\theta \wedge \theta \wedge \theta) \overset{q}{\underset{\simeq}{\longrightarrow}} k_{\mathbb{R}}$ &  prequantization condition
\\
\hline
\end{tabular}
\end{center}
Diagrammatically, this data precisely corresponds to a diagram as shown on the left in the following,
and hence the universal property of the homotopy pullback uniquely associates a lift
$\nabla_{\mathrm{WZW}}$ as on the right:
\begin{center}
\fbox{$
  \raisebox{20pt}{
  \xymatrix{
    & & \mathbf{\Omega}^3_{\mathrm{cl}}
        \ar[d]
    \\
    G \ar[r]_k \ar[urr]^{\mu_3(\theta)}_>{\ }="s" & \mathbf{B}^3 \mathbb{Z}  \ar[r]^{\ }="t" & \mathbf{B}^3 \flat \mathbb{R}
    \ar@{=>}^q "s"; "t"
  }
  }
  \;\;\;\;\;\;\;\;
  \Leftrightarrow
  \;\;\;\;\;\;\;\;
  \raisebox{20pt}{
  \xymatrix{
    & \mathbf{B}^2 (\mathbb{R}/\mathbb{Z})_{\mathrm{conn}} \ar@{}[dr]|{\mathrm{(pb)}} \ar[d]\ar[r]^-{\mathrm{curv}}
    & \mathbf{\Omega}^3_{\mathrm{cl}}
        \ar[d]_<<{\ }="s"
    \\
    G \ar[r]_k \ar@{-->}[ur]^{\nabla_{\mathrm{WZW}}} & \mathbf{B}^3 \mathbb{Z}  \ar[r]^{\ }="t" & \mathbf{B}^3 \flat \mathbb{R}
    \ar@{=>} "s"; "t"
  }
  }
$
}
\end{center}
This $\nabla_{\mathrm{WZW}}$ is the required prequantum topological term for the 2d WZW model.
Hence the prequantum 2d  WZW sigma-model field theory is the $(p  = 2)$-dimensional prequantum field theory with
target space the group $G$ and with local prequantum Lagrangian, i.e. with Euler-Lagrange gerbe given by
$$
  \mathbf{L}
  :=
  \underset{\mathbf{L}_{\mathrm{kin}}}{\underbrace{ \langle \theta_H \wedge \star \theta_H \rangle}}
  +
  \underset{\mathbf{L}_{\mathrm{WZW}}}{\underbrace{(\nabla_{\mathrm{WZW}})_H}}
  \;\;:\;\;
  \Sigma \times G
  \longrightarrow
  \mathbf{B}^{p+1}_H(\mathbb{R}/_{\!\hbar}\mathbb{Z})_{\mathrm{conn}}
  \,.
$$
This prequantization is a de-transgression of a famous traditional construction. To see this,
write  $\hat \Omega_k G$ for level-$k$ Kac-Moody loop group extension of $G$.
This has an adjoint action by the based path group $P_e G$. Write
$$
  \mathrm{String}(G) :=  P_e G/\!/ \hat \Omega_k G
$$
for the homotopy quotient. This is a differentiable group stack, called the \emph{string 2-group} \cite{BCSS}.
It turns out to be the total space of the 2-bundle underlying $\nabla_{\mathrm{WZW}}$
$$
  \xymatrix{
    \mathrm{String}(G) \ar[d] \ar@{}[drr]|{\mathrm{(pb)}} \ar[rr] && \ast \ar[d]_{\ }="s"
    \\
    G \ar[r]^-{\nabla_{\mathrm{WZW}}} & \mathbf{B}^2(\mathbb{R}/\mathbb{Z})_{\mathrm{conn}} \ar[r]^{\ }="t" & \mathbf{B}^2 (\mathbb{R}/\mathbb{Z})
    \ar[r]^{\mbox{\tiny$\pi_\infty$}}
    &
    K(\mathbb{Z},3)
  }
$$
and it is a de-transgression of the Kac-Moody loop group extension $\hat L_K G$: transgressing to fields over the circle gives:
$$
  \xymatrix{
    \hat L_k G
    \ar[rr]
    \ar[d]
    \ar@{}[drr]|{\mathrm{(pb)}}
    &&
    \ast
    \ar[d]
    \\
    L G
    \ar[r]^-{\int_S^1 \nabla_{\mathrm{WZW}}}
    &
    \mathbf{B}(\mathbb{R}/\mathbb{Z})_{\mathrm{conn}}
    \ar[r]
    &
    \mathbf{B}(\mathbb{R}/\mathbb{Z})
  }
$$
The string 2-group also appears again as the 2-group of Noether symmetries, in the sense of
section \ref{symmetriescurrentsinintroduction}, of the prequantum 2d WZW model. The
Noether homotopy fiber sequence for the prequantum 2d WZW model looks as follows \cite{hgp}:
\begin{center}
\begin{tabular}{|ccccc|}
\hline
$
  \left\{
  \raisebox{29pt}{
  \xymatrix@R=39pt{
    G
      \ar@/^1.7pc/[d]^{\nabla_{\mathrm{WZW}}}_<<<{\ }="s"
      \ar@/_1.7pc/[d]_{\nabla_{\mathrm{WZW}}}^>>>{\ }="t"
    \\
    \mathbf{B}^{2}_H(\mathbb{R}/_{\!\hbar}\mathbb{Z})_{\mathrm{conn}}
    \ar@{=>}|<<<<{\mathrm{topological} \atop \mathrm{current}} "s"; "t"
  }
  }
  \right\}
$
&
$\longrightarrow$
&
$
  \left\{
  \raisebox{29pt}{
  \xymatrix@R=39pt{
    G \ar[dr]_{\nabla_{\mathrm{WZW}}}^>>>{\ }="t" \ar[rr]|\simeq^{\mathrm{symmetry}}_>>>>>>>>>>>>>>{\ }="s"
    && G \ar[dl]^{\nabla_{\mathrm{WZW}}}
    \\
    & \mathbf{B}^{p+1}_H(\mathbb{R}/_{\!\hbar}\mathbb{Z})_{\mathrm{conn}}
    \ar@{=>}|<<<<<<{\mathrm{Noether}\;\mathrm{current}} "s"; "t"
  }
  }
  \right\}
$
&
$\longrightarrow$
&
$
  \left\{
  \xymatrix{
    G  \ar[rr]|\simeq^{\mathrm{point}}_-{\mathrm{symmetry}} && G
  }
  \right\}
$
\\
\hline
&&&&
\\
$\mathbf{B}U(1)$ & $\longrightarrow$
& $\mathrm{String}(G)$
& $\longrightarrow$ & $G$
\\
&&&&
\\
\hline
\end{tabular}
\end{center}
In fact, this extension is classified by the smooth universal characteristic class
$\mathbf{c} :\mathbf{B}G \longrightarrow \mathbf{B}^3 U(1)$, whose differential refinement gave
3d Chern-Simons theory in section \ref{introductionchernsimonstypefieldtheories}.

Given a $G$-principal bundle $P \to X$, then one may aks for a fiberwise parameterization of $\nabla_{\mathrm{WZW}}$
over $P$. If such \emph{definite parameterization} $\nabla : P \to \mathbf{B}^{2}(\mathbb{R}/\mathbb{Z})_{\mathrm{conn}}$ exists, then it defines
the prequantum topological term for the \emph{parameterized WZW model} with target space $P$.
$$
  \xymatrix{
    G \simeq P_x
    \ar@{}[dr]|{\mathrm{(pb)}}
    \ar[d]
    \ar@/^2.3pc/[rrr]^-{\nabla_{\mathrm{WZW}}}_{\ }="s"
    \ar[r]
    &
    P
    \ar[rr]|-\nabla_-{ \mathrm{definite} \atop \mathrm{parameterization}  }^<<<{\ }="t"
    \ar[d]
    && \mathbf{B}^2(\mathbb{R}/\mathbb{Z})_{\mathrm{conn}}
    \\
    \{x\} \ar@{^{(}->}[r] & X
    \ar@{=>}^\simeq "s"; "t"
  }
$$
Such a parameterization is equivalent to a lift of a structure group of $P$ through the
above extension $\mathrm{String}(G) \longrightarrow G$ \cite{SchreiberWZWterms}.
Accordingly, the obstruction to parameterizing $\nabla_{\mathrm{WZW}}$ over $P$ is the universal extension class $\mathbf{c}$
evaluated on $P$.
Specifically for the case that $G = \mathrm{Spin}\times \mathrm{SU}$, this is the sum of fractional Pontryagin and second Chern class:
$$
  \tfrac{1}{2}p_1 - c_2 \in H^4(X,\mathbb{Z})
  \,.
$$
The vanishing of this class is the \emph{Green-Schwarz anomaly cancellation}
condition for the 2d field theory describing propagation of the heterotic string on $X$. This perspective on the Green-Schwarz anomaly via
parameterized WZW models had been suggested in \cite{DistlerSharpe}. The prequantum field theory we present
serves to make this precise and to generalize it to higher dimensional parameterized WZW-type field theories.

\medskip

Generally, given any $L_\infty$-cocycle $\mu : \mathbf{B}\mathfrak{g} \longrightarrow \mathbf{B}^{p+2}\mathbb{R}$
as in section \ref{introductionhighergaugefields} with induced smooth $\infty$-group cocycle
$\mathbf{c} : \mathbf{B}G \longrightarrow \mathbf{B}^{p+2}(\mathbb{R}/\Gamma)$ as in section \ref{introductionsigmamodelfieldtheories},
then there is a higher analog of the universal construction of the WZW-type topological term $\nabla_{\mathrm{WZW}}$.

First of all, the homotopy pullback characterization of $\mathbf{B}^{p+1}(\mathbb{R}/\mathbb{Z})_{\mathrm{conn}}$
refines to one that does not just involve the geometrically discrete coefficients $\mathbf{B}^{p+2}\mathbb{Z}$,
but the smooth coefficients $\mathbf{B}^{p+1}(\mathbb{R}/\mathbb{Z})$.
$$
  \raisebox{20pt}{
  \xymatrix{
    & \mathbf{\Omega}^{p+2}_{\mathrm{cl}}
    \ar[dr]
    \\
    \mathbf{B}^{p+1}(\mathbb{R}/\mathbb{Z})_{\mathrm{conn}}
    \ar@{}[rr]|{\mathrm{(pb)}}
    \ar[ur]^{\mathrm{curv}}
    \ar[dr]
    &&
    \flat_{\mathrm{dR}}\mathbf{B}^{p+2} \mathbb{R}
    \\
    & \mathbf{B}^{p+1}(\mathbb{R}/\mathbb{Z})
    \ar[ur]
  }
  }
  \,.
$$
Here $\flat_{\mathrm{dR}}(-)$ denotes the homotopy fiber of the canonical map
$\flat(-)\longrightarrow (-)$ embedding the underlying discrete smooth structure
of any object into the given smooth object. A key aspect of the theory is that
the further homotopy fiber of $\flat_{\mathrm{dR}}(-) \longrightarrow \flat(-)$
has the interpretation of being the Maurer-Cartan form $\theta$ on the given
smooth $\infty$-groupoid.
$$
  \xymatrix{
    & G
    \ar[d]^{\theta}_{\mbox{\tiny \begin{tabular}{c} Maurer-Cartan \\ form \end{tabular}}}
    \\
    & \flat_{\mathrm{dR}}\mathbf{B}G  \ar[r] & \flat \mathbf{B}G \ar[d]
    \\
    & & \mathbf{B}G
  }
$$
Or rather, one finds that $\flat_{\mathrm{dR}}\mathbf{B}G \simeq \mathbf{\Omega}_{\mathrm{flat}}^{1 \leq \bullet\leq p+2}(-,\mathfrak{g})$
is the coefficient for ``hypercohomology'' in flat $\mathfrak{g}$-valued differential forms, hence for $G$ a higher smooth group
then its Maurer-Cartan form $\theta$ is not, in general, a globally defined differential form, but instead a system
of locally defined forms with higher coherent gluing data.

But one may universally force $\theta$ to become globally defined, so to speak, by pulling it back
along the inclusion $\mathbf{\Omega}_{\mathrm{flat}}(-,\mathfrak{g})$ of the globally defined
flat $\mathfrak{g}$-valued forms. This defines a differential extension $\tilde G$ of $G$ equipped
with a globally defined Maurer-Cartan form $\tilde \theta$, by the following homotopy pullback diagram
$$
  \xymatrix{
    & \mathbf{\Omega}_{\mathrm{flat}}(-,\mathfrak{g})
     \ar[dr]
    \\
    \tilde G
    \ar@{}[rr]|{\mathrm{(pb)}}
    \ar[dr]
    \ar[ur]^{\tilde \theta}
    && \flat_{\mathrm{dR}}\mathbf{B}G
    \\
    & G
    \ar[ur]^{\theta}
  }
  \,.
$$
When $G$ is an ordinary Lie group, then it so happens that
$\flat_{\mathrm{dr}}\mathbf{B}G \simeq \mathbf{\Omega}_{\mathrm{flat}}(-,\mathfrak{g})$, and so in this
case $\tilde G \simeq G$ and $\tilde \theta \simeq \theta$, so that nothing new happens.

At the other extreme, when $G = \mathbf{B}^{p+1}(\mathbb{R}/\mathbb{Z})$, then $\theta \simeq \mathrm{curv}$
as above, and so in this case one find that $\tilde G$ is $\mathbf{B}^{p+1}(\mathbb{R}/\mathbb{Z})_{\mathrm{conn}}$
and that $\tilde \theta \simeq F_{(-)}$ is the map that sends an $(p+1)$-form connection to its globally
defined curvature $(p+2)$-form.

More generally these two extreme cases mix: when $G$ is a $\mathbf{B}^p (\mathbb{R}/Z)$-extension of an ordinary Lie group,
then $\tilde G$ is a twisted product of $G$ with $\mathbf{B}^{p}(\mathbb{R}/\mathbb{Z})_{\mathrm{conn}}$, hence then
a single map
$$
  (\phi,B) : \Sigma \longrightarrow \tilde G
$$
is a pair consisting of an ordinary sigma-model field $\phi$ together with a $\phi$-twisted $p$-form
connection on $\Sigma$.

Hence the construction of $\tilde G$ is a twisted generalization of the construction of differential
coefficients. In particular, given an $L_\infty$-cocycle
$\mu : \mathbf{B}\mathfrak{g}\longrightarrow \mathbf{B}^{p+2}\mathbb{R}$
Lie-integrating to an $\infty$-group cocycle $\mathbf{c} : \mathbf{B}G \to \mathbf{B}^{p+2}(\mathbb{R}/\Gamma)$,
 then it Lie integrates to a
prequantum topological term $\nabla_{\mathrm{WZW}} : \tilde G \longrightarrow \mathbf{B}^{p+1}(\mathbb{R}/\Gamma)_{\mathrm{conn}}$
via the universal dashed map in the following induced diagram:
$$
  
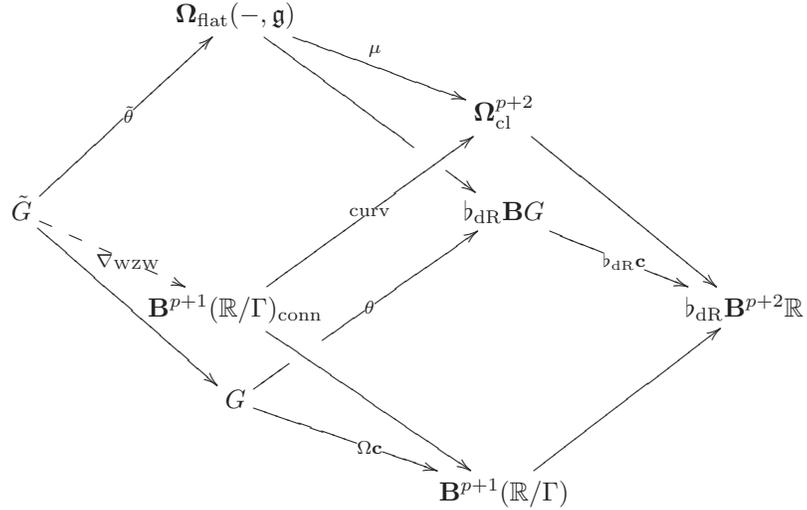
\xymatrix@C=16pt@R=20pt{
    &&
    &&&&
    \mathbf{\Omega}_{\mathrm{flat}}(-,\mathfrak{g})
    \ar[drr]^\mu
    \ar[ddrr]|>>>>>>>>{\phantom{AA} \atop \phantom{AA}}
    \\
    &&&&
    &&&&
    \mathbf{\Omega}^{p+2}_{\mathrm{cl}}
    \ar[ddrr]
    \\
    &&&&
    \tilde G
    \ar[uurr]|{\tilde \theta}
    \ar[ddrr]
    \ar@{-->}[drr]|{\nabla_{\mathrm{WZW}}}
    &&&&
    \flat_{\mathrm{dR}}\mathbf{B}G
    \ar[rrd]|{\flat_{\mathrm{dR}}\mathbf{c}}
    \\
    &&
    &&
    &&
    \mathbf{B}^{p+1}(\mathbb{R}/\Gamma)_{\mathrm{conn}}
    \ar[uurr]|{\mathrm{curv}}
    \ar[ddrr]
    &&&&
    \flat_{\mathrm{dR}}\mathbf{B}^{p+2}\mathbb{R}
    \\
    &&
    &&&&
    G
    \ar[uurr]|<<<<<<<<<{\phantom{AA} \atop\phantom{AA}}|{\theta}
    \ar[drr]|{\Omega \mathbf{c}}
    \\
    &&&&
    &&&&
    \mathbf{B}^{p+1}(\mathbb{R}/\Gamma)
    \ar[uurr]
  }
$$

This construction provides a large supply of prequantum Wess-Zumino-Witten type field theories.
Indeed, by the discussion in \ref{introductionhighergaugefields}, from every $L_\infty$-algebroid
there emanates a bouquet of $L_\infty$-extensions with $L_\infty$-cocycles on them, hence
for every WZW-type sigma model prequantum field theories we find a whole bouquet of prequantum field
theories emanating from it.

Therefore it is interesting to consider the simplest non-trivial $L_\infty$-algebroids and
see which bouquets of prequantum field theories they induce. The abelian line Lie algebra $\mathbb{R}$
is arguably the simplest non-vanishing $L_\infty$-algebroid, but it is in fact a little too simple
for this purpose, the bouquet it induces is not interesting. But all of the above generalizes
essentially verbatim to super-algebra and super-geometry, and in super-Lie-algebra theory we have the
``odd lines'' $\mathbb{R}^{0|q}$, i.e. the superpoints. The bouquet which emanates from these turns out to be remarkably
rich \cite{InfinityWZW}, it gives the entire $p$-brane spectrum of string theory/M-theory.

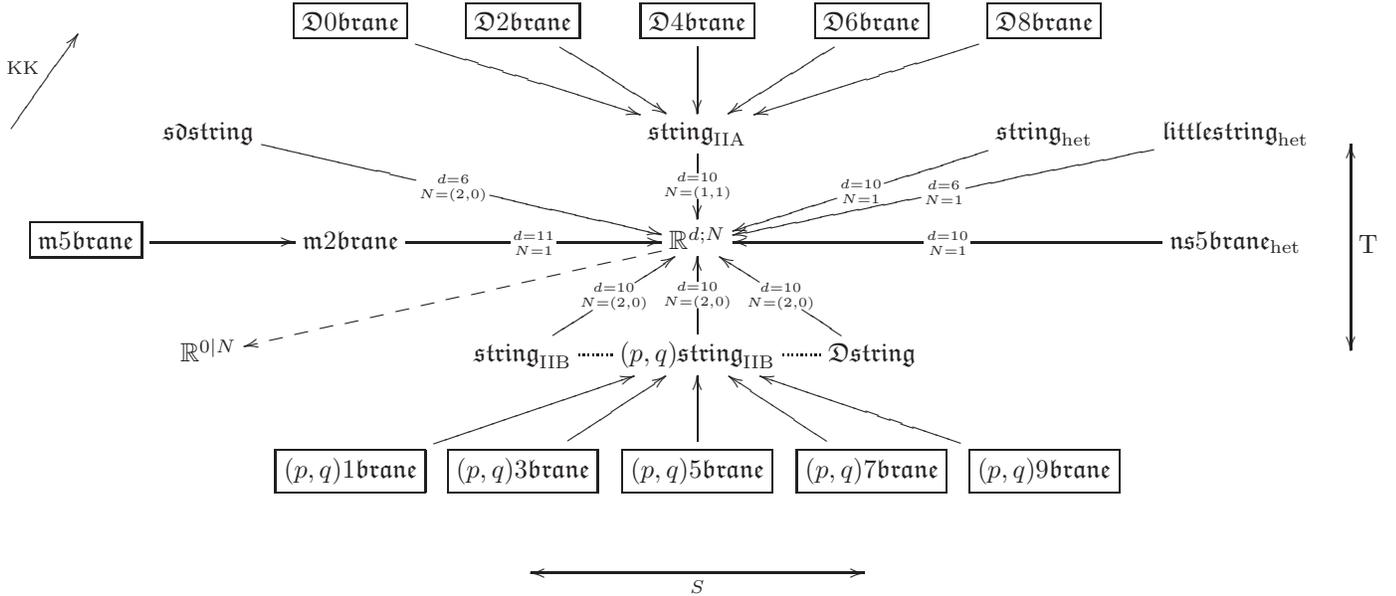
\vspace{-.4cm}

$$
  \hspace{-.6cm}
  \xymatrix@C=2pt{
    &&
	&& \fbox{$\mathfrak{D}0\mathfrak{brane}$} \ar[drr]
	& \fbox{$\mathfrak{D}2\mathfrak{brane}$} \ar[dr]
	& \fbox{$\mathfrak{D}4\mathfrak{brane}$} \ar[d]
	& \fbox{$\mathfrak{D}6\mathfrak{brane}$} \ar[dl]
	& \fbox{$\mathfrak{D}8\mathfrak{brane}$} \ar[dll]
    \\
    & \ar[ur]^{\mathrm{KK}}&
	& \mathfrak{sdstring} \ar[drrr]|{{d = 6} \atop {N = (2,0)}}
	&
	&& \mathfrak{string}_{\mathrm{IIA}} \ar[d]|-{{d=10} \atop {N=(1,1)}}
	&& \mathfrak{string}_{\mathrm{het}} \ar[dll]|-{{d=10}\atop {N = 1}}
	&& \mathfrak{littlestring}_{\mathrm{het}} \ar[dllll]|-{{d=6}\atop {N = 1}}
	&&
     \ar@{<->}[dd]^{\mbox{T}}	
	&&
    \\
    && \fbox{$\mathfrak{m}5\mathfrak{brane}$} \ar[rr]
	&& \mathfrak{m}2\mathfrak{brane} \ar[rr]|-{{d=11} \atop {N=1}}
	&& \mathbb{R}^{d;N}
       \ar@{-->}[dlll]
	&&
	&& \mathfrak{ns}5\mathfrak{brane}_{\mathrm{het}} \ar[llll]|-{{d = 10}\atop {N =1}}
	&&
	\\
	&&
    &\mathbb{R}^{0|N}
	&
	& \mathfrak{string}_{\mathrm{IIB}} \ar[ur]|-{{d = 10}\atop {N=(2,0)}}
	\ar@{.}[r]
	& (p,q)\mathfrak{string}_{\mathrm{IIB}} \ar[u]|-{{d = 10}\atop {N=(2,0)}}
	\ar@{.}[r]
	& \mathfrak{Dstring} \ar[ul]|-{{d = 10}\atop {N=(2,0)}}
	&&
	&&
	&&
    \\
    &
    &
	&& \fbox{$(p,q)1\mathfrak{brane}$} \ar[urr]
	&\fbox{$(p,q)3\mathfrak{brane}$} \ar[ur]
	& \fbox{$(p,q)5\mathfrak{brane}$} \ar[u]
	& \fbox{$(p,q)7\mathfrak{brane}$} \ar[ul]
	& \fbox{$(p,q)9\mathfrak{brane}$} \ar[ull]
	\\
	&& &&  & \ar@{<->}[rr]_S &  &&
  }
$$

Each entry in this diagram denotes a super $L_\infty$-algebra extension of some
super Minkowski spacetime $\mathbb{R}^{d-1,1|N}$ (regarded as the corresponding supersymmetry
super Lie algebra), and each arrow denotes a
super-$L_\infty$-extension classified by a $p+2$ cocycle for some $p$.
By the above general
construction, this cocycle
induces a $(p+1)$-dimensional WZW-type sigma-model prequantum field theory with target space
a higher extension of super-Minkowski spacetime \cite{InfinityWZW}, and the names of the super $L_\infty$-algebras
in the above diagram correspond to the traditional names of these super $p$-branes.

As for the traditional WZW-models, all of this structure naturally generalizes to its
parameterized versions \cite{SchreiberWZWterms}: given any higher extended super Minkowski spacetime $V$
equipped with a prequantum topological term $\nabla_{\mathrm{WZW}} : V \longrightarrow \mathbf{B}^{p+1}(\mathbb{R}/\Gamma)_{\mathrm{conn}}$
for a super $p$-brane sigma model, we may ask for globalizations of $\nabla$ over $V$-manifolds
($V$-{\'e}tale stacks) $X$, hence for topological term $\nabla$ on all of $X$ that is suitably equivalent
on each infinitesimal disk $\mathbb{D}^X_x \simeq \mathbb{D}^V_e$ to $\nabla_{\mathrm{WZW}}$.
$$
  \xymatrix{
    \mathbb{D}_x^{(1)}
    \ar@{}[dr]|{\mathrm{(pb)}}
    \ar[d]\ar[r]
    \ar@/^2.2pc/[rrr]^{\nabla_{\mathrm{WZW}}}_{\ }="s"
    &
    T^{(1)} X
    \ar[d]
    \ar[rr]|-{\nabla}_-{\mathrm{definite} \atop \mathrm{globalization}}^<<{\ }="t"
    &&
    \mathbf{B}^{p+1}(\mathbb{R}/\Gamma)_{\mathrm{conn}}
    \\
    \{x\} \ar@{^{(}->}[r] & X
    \ar@{=>}^\simeq "s"; "t"
  }
$$
Such globalizations serve as prequantum topological
terms for  WZW-type sigma-models describing the propagation of super $p$-branes
on $V$-manifolds $X$ (e.g. \cite[sections 2,3]{Duff99}). One finds
(this is proven with the abstract theory surveyed below in section \ref{abstractprequantumgeometry})
that such globalizations equip the higher frame bundle of $X$ with a lift of
its structure group through a canonical map $\mathbf{Stab}_{\mathrm{GL}(V)}(\nabla) \longrightarrow \mathrm{GL}(V)$
from the homotopy stabilizer group of the WZW term,
in direct analogy to the previous examples. Apart from ``cancelling the classical anomalies''
of making the super $p$-brane WZW-type sigma-model be globally defined on $X$, such a lift induces
metric structure on $X$:

Since the homotopy stabilization of $\nabla$ in particular stabilizes its curvature form,
  there is a reduction of the structure group of the $V$-manifold in direct analogy to how
  a globalization of the ``associative'' 3-form $\alpha$ on $\mathbb{R}^7$ equips a 7-manifold with $G_2$-structure.
  For the above super $p$-brane models the relevant stabilizer is the spin-cover of
  the Lorentz group,
  and hence globalizing the prequantum $p$-brane model over $X$ in particular induces orthogonal
  structure on $X$, hence equips $X$ with a field configuration of supergravity.

Given such a globalization of a topological term $\nabla$ over a $V$-manifold $X$,
it is natural to require it to be infinitesimally integrable. In the present example this comes out
to imply that the torsion of the orthogonal structure on $X$ vanishes. This is
particularly interesting at the top end of the brane bouquet: for globalization over
over an 11-dimensional supermanifold, the vanishing of the torsion is equivalent to
$X$ satisfying the equations of motion of 11-dimensional gravity \cite{CandielloLechner}.
The Noether charges of the corresponding WZW-type prequantum field theory are
the \emph{supergravity BPS-charges} \cite{SatiSchreiber15}.

Here the relation to $G_2$-structure is more than an analogy. We may naturally lift the topological term
for the M2-brane sigma-model from $\mathbb{R}/\mathbb{Z}$-coefficients to $\mathbb{C}/\mathbb{Z}$-coefficients
by adding $\alpha : \mathbb{R}^{10,1\vert \mathbf{32}} \to \mathbb{R}^7 \to \mathbf{\Omega}^3_{\mathrm{cl}}$.
Then a globalization of the complex linear combination
$$
  \nabla_{\mathrm{M2}} + i \alpha :  \mathbb{R}^{10,1\vert \mathbf{32}} \longrightarrow \mathbf{B}^{3}(\mathbb{C}/\Gamma)_{\mathrm{conn}}
$$
over an 11-dimensional supermanifold $X$ equips $X$ with the structure of a $G_2$-fibration over a 4-dimensional $N=1$ supergravity spacetime.
The volume holonomy of $\nabla_{\mathrm{M2}} + i \alpha$ around supersymmetric 3-cycles are the ``M2-instanton contributions''.
This setup of
11-dimensional supergravity Kaluza-Klein-compactified on $G_2$-manifolds to 4 spacetime dimensions and with the
prequantum M2/M5-brane charges and instantons included -- known as \emph{M-theory on $G_2$-manifolds} \cite{Acharya02, AtiyahWitten03}
-- comes at least close to capturing the qualitative structure of experimentally observed fundamental physics.

$$
  
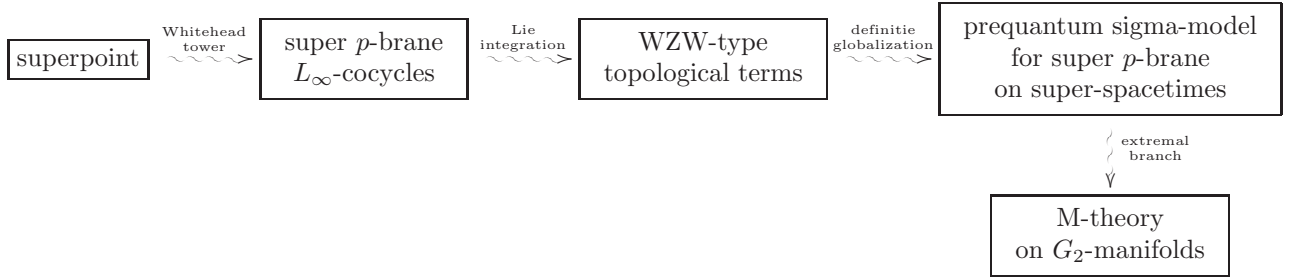
\xymatrix@C=36pt{
    \fbox{superpoint}
    \ar@{~>}[r]^-{\mathrm{Whitehead} \atop \mathrm{tower}}
    &
    \fbox{\begin{tabular}{c}super $p$-brane\\ $L_\infty$-cocycles\end{tabular}}
    \ar@{~>}[r]^-{\mathrm{Lie} \atop \mathrm{integration}}
    &
    \fbox{\begin{tabular}{c}WZW-type \\ topological terms\end{tabular}}
    \ar@{~>}[r]^-{\mathrm{definitie} \atop \mathrm{globalization}}
    &
    \fbox{\begin{tabular}{c}prequantum sigma-model \\ for super $p$-brane \\ on super-spacetimes\end{tabular}}
    \ar@{~>}[d]^{\mathrm{extremal} \atop \mathrm{branch}}
    \\
    &&&
    \fbox{\begin{tabular}{c}M-theory \\ on $G_2$-manifolds\end{tabular}}
  }
$$

This shows that a fair bit of fundamental physics is encoded in
those prequantum field theories that are canonically induced from a minimum of input data.
For more on this see \cite{5BraneCohomology,M2M5Charges}.
We continue in section \ref{abstractprequantumgeometry} with indicating that the concept of
prequantum field theory itself arises from first principles.

%%%%%%%%%%%%%%%%%%%%%%%%%%%%%%%%%%%%%%%%%%%%%
\section{Abstract prequantum field theory}
\label{abstractprequantumgeometry}
\label{syntheticpqg}
%%%%%%%%%%%%%%%%%%%%%%%%%%%%%%%%%%%%%%%%%%%%%

Above in section \ref{PrequantumLocalFieldTheoryInMotivation} we have surveyed
mathematical structure that captures prequantum local field theory.
While the constructions and results proceed smoothly, the whole setup may still look
somewhat intricate. One needs a good abstract machinery to be practically able to
analyze properties of, say, Euler-Lagrange 5-gerbes on 3-stacky jet super-bundles (as they do arise in the
formulation of the M5-brane prequantum sigma-model field theory as in section \ref{introductionWZWfieldtheory}), because it is unfeasible to
do so in components. Moreover, if prequantum field theory is part of the fundamental
description of nature, one may expect that its mathematical formulation is
indeed natural and neat. We now survey results from \cite{dcct} showing that a good abstract formalization
of the differential super-geometry and of the differential cohomology and of the PDE-theory necessary for
formulating prequantum field theory does exist.
For further exposition of the following see also \cite{IHPTalk, Bristol,DMV}.

%%%%%%%%%%%%%%%%%%%%%%%%%%%%%%%%%
\subsection{Modal homotopy theory}
\label{introductionmodalhomotopytheory}
%%%%%%%%%%%%%%%%%%%%%%%%%%%%%%%%%

The homotopy theory in which all pre-quantum physics that has been considered in section \ref{examplesinmotivation} naturally
finds its place is that of  super formal smooth higher stacks. We briefly state the definition below.
Then we claim that this homotopy theory carries a rich progression of adjoint idempotent $\infty$-(co-)monads.
Such idempotent (co-)monads equip the homotopy theory with what in formal logic is known as \emph{modalities},
hence we may speak of \emph{modal homotopy theory}. The particular system of modalities that we find and consider
we call (super-)\emph{differential cohesive homotopy theory}.  Below in
sections \ref{introductionabstractdifferentialcohomology} and \ref{introductionabstractdifferentialgeometry} we
survey the rich differential cohomology and differential geometry that is implied formally from just this
abstract modal homotopy theoretic structure.

\begin{definition}
  \label{introductionsuperformalsmoothsite}
  The site of \emph{super formal smooth Cartesian spaces}
  $$
    \mathrm{SupFormCartSp} \hookrightarrow \mathrm{sCAlg}_{\mathbb{R}}^{\mathrm{op}}
  $$
  is the full subcategory of that of super-commutative superalgebras over $\mathbb{R}$ on those which are tensor products
  $$
    C^\infty(\mathbb{R}^n \times \mathbb{D}) := C^\infty(\mathbb{R}^n)\otimes_{\mathbb{R}}(\mathbb{R}\oplus V)
  $$
  of the algebra of smooth real functions in $n$ variables, for any $n \in \mathbb{N}$, with
  a supercommutative superalgebra $(\mathbb{R} \oplus V)$ for finite dimensional nilpotent $V$.
  Take this to be a site by equipping it with the coverage whose coverings are of the form
  $
    \{ U_i \times\mathbb{D} \stackrel{(\phi_i,\mathrm{id})}{\longrightarrow} X \times \mathbb{D}\}
  $
  for $\{U_i \stackrel{\phi_i}{\longrightarrow} X\}$ being an open cover of smooth manifolds.
\end{definition}
\begin{itemize}
  \item $\mathrm{Mfd}$ is the category of smooth manifolds;
  \item $\mathrm{FormMfd}$ is the category of formal smooth manifolds \cite{Kock80} \cite[sections I.17 and I.19]{KockBookA};
  \item $\mathrm{DiffeolSp}$ is the category of diffeological spaces \cite{IglesiasZemmour}, which is the category of concrete sheaves
  on $\mathrm{Mfd}$;
  \item $\mathrm{Sh}(\mathrm{FormMfd})$ is the ``Cahiers topos'' \cite{Dubuc79} that was introduced
    as a model for the Kock-Lawvere axioms \cite[I.12]{KockBookA} \cite[1.3]{KockBook} for synthetic differential geometry.
\end{itemize}

\newpage

\begin{proposition}
The sheaf toposes and $\infty$-sheaf $\infty$-toposes \cite{Lurie} over the sites
in def. \ref{introductionsuperformalsmoothsite} form the following system of full inclusions
of categories of geometric spaces.
\end{proposition}

\hspace{-3.0cm}
\fbox{
$
  
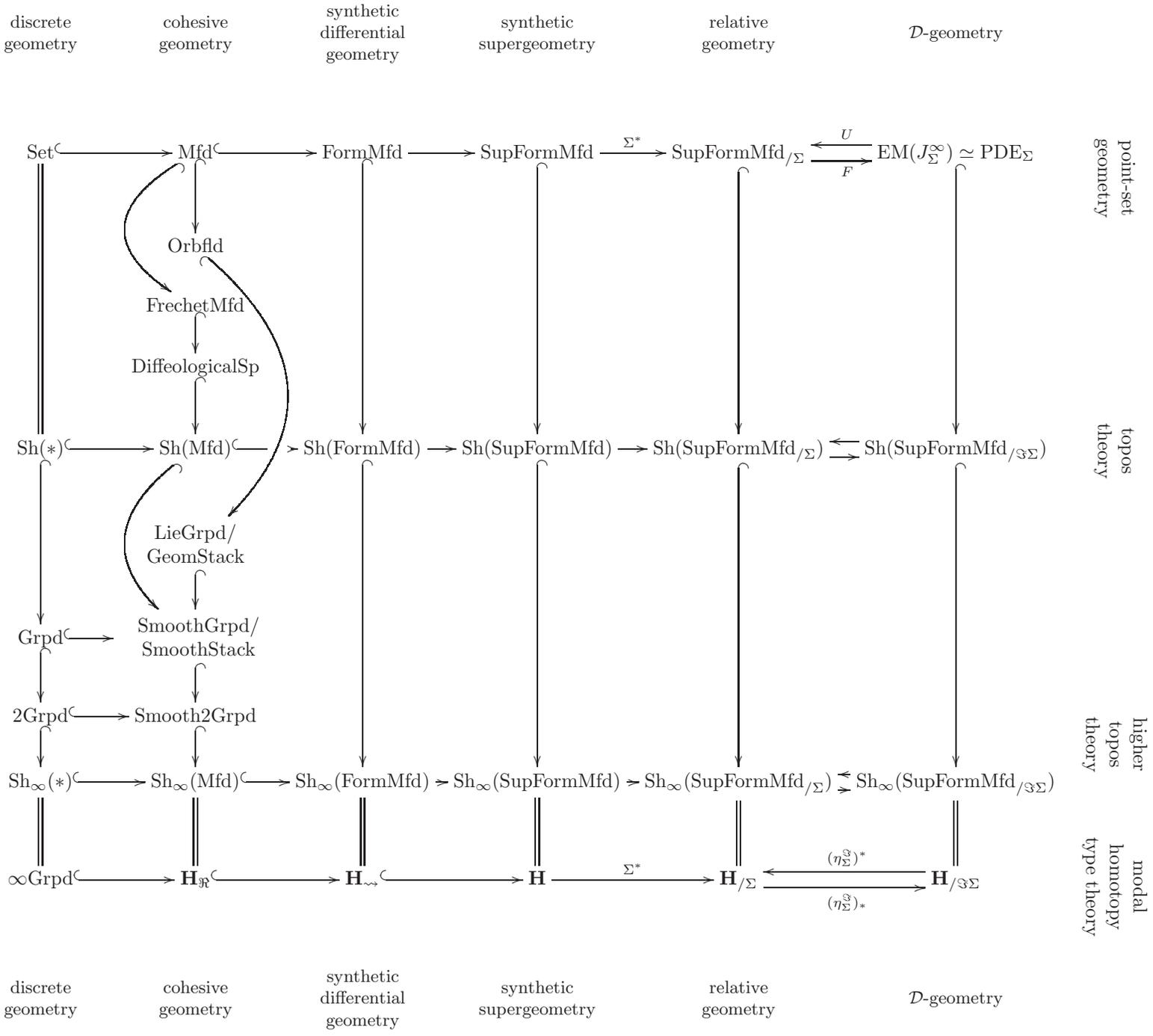
\xymatrix@R=17pt@C=6pt{
    \mbox{
      \small
      \begin{tabular}{c}
        discrete
        \\
        geometry
      \end{tabular}
    }
    &
    \mbox{
      \small
      \begin{tabular}{c}
        cohesive
        \\
        geometry
      \end{tabular}
    }
    &
    \mbox{
      \small
      \begin{tabular}{c}
        synthetic
        \\
        differential
        \\
        geometry
      \end{tabular}
    }
    &
    \mbox{
      \small
      \begin{tabular}{c}
        synthetic
        \\
        supergeometry
      \end{tabular}
    }
    &
    \mbox{
      \small
      \begin{tabular}{c}
        relative
        \\
        geometry
      \end{tabular}
    }
    &
    \mbox{
     \small
     \begin{tabular}{c}
       $\mathcal{D}$-geometry
     \end{tabular}
    }
    \\
    \mathrm{Set}
    \ar@{^{(}->}[r]
    \ar@{=}[dddd]
    &
    \mathrm{Mfd}
    \ar@{^{(}->}[d]
    \ar@/_3pc/@{^{(}->}[dd]
    \ar@{^{(}->}[r]
    & \mathrm{FormMfd}
    \ar@{^{(}->}[dddd]
    \ar[r]
    &
    \mathrm{SupFormMfd}
    \ar@{^{(}->}[dddd]
    \ar[r]^-{\Sigma^\ast}
    &
    \mathrm{SupFormMfd}_{/\Sigma}
    \ar@{^{(}->}[dddd]
    \ar@{<-}@<+4pt>[r]^U
    \ar@<-4pt>[r]_F
    &
    \mathrm{EM}(J^\infty_\Sigma)
    \simeq
    \mathrm{PDE}_\Sigma
    \ar@{^{(}->}[dddd]
    &
    \begin{turn}{-90} \mbox{ \hspace{-0.8cm} \begin{tabular}{c}  point-set \\ geometry \end{tabular}} \end{turn}
    \\
    &
    \mathrm{Orbfld}
    \ar@/^3.9pc/@{_{(}->}[dddd]
    &&&&&&&
    \\
    &
    \mathrm{FrechetMfd}
    \ar@{^{(}->}[d]
    &&&&&&
    \\
    & \mathrm{DiffeologicalSp}
    \ar@{^{(}->}[d]
    &&&&&&&
    \\
    \mathrm{Sh}(\ast)
    \ar@{^{(}->}[dd]
    \ar@{^{(}->}[r]
    &
    \mathrm{Sh}(\mathrm{Mfd})
    \ar@/_3pc/@{^{(}->}[dd]
    \ar@{^{(}->}[r]|{\phantom{AA}}
    &
    \mathrm{Sh}(\mathrm{FormMfd})
    \ar@{^{(}->}[dddd]
    \ar[r]
    &
    \mathrm{Sh}(\mathrm{SupFormMfd})
    \ar@{^{(}->}[dddd]
    \ar[r]
    &
    \mathrm{Sh}(\mathrm{SupFormMfd}_{/\Sigma})
    \ar@{<-}@<+4pt>[r]
    \ar@<-4pt>[r]
    \ar@{^{(}->}[dddd]
    &
    \mathrm{Sh}(\mathrm{SupFormMfd}_{/\Im \Sigma})
    \ar@{^{(}->}[dddd]
    &
    \begin{turn}{-90} \mbox{ \hspace{-1cm} \begin{tabular}{c} topos \\ theory \end{tabular} }\end{turn}
    \\
    &
    \mbox{ \begin{tabular}{c} $\mathrm{LieGrpd}$/ \\ $\mathrm{GeomStack}$ \end{tabular} }
    \ar@{^{(}->}[d]
    &&&&&
    \\
    \mathrm{Grpd}
    \ar@{^{(}->}[d]
    \ar@{^{(}->}[r]
    &
    \mbox{ \begin{tabular}{c} $\mathrm{SmoothGrpd}$/ \\ $\mathrm{SmoothStack}$ \end{tabular}}
    \ar@{^{(}->}[d]
    &
    &&
    &&
    \\
    2\mathrm{Grpd}
    \ar@{^{(}->}[d]
    \ar@{^{(}->}[r]
    &
    \mathrm{Smooth}2\mathrm{Grpd}
    \ar@{^{(}->}[d]
    &&&&&
    \\
    \mathrm{Sh}_\infty(\ast)
    \ar@{^{(}->}[r]
    \ar@{=}[d]
    &
    \mathrm{Sh}_\infty(\mathrm{Mfd})
    \ar@{=}[d]
    \ar@{^{(}->}[r]
    &
    \mathrm{Sh}_\infty(\mathrm{FormMfd})
    \ar[r]
    \ar@{=}[d]
    &
    \mathrm{Sh}_\infty(\mathrm{SupFormMfd})
    \ar[r]
    \ar@{=}[d]
    &
    \mathrm{Sh}_\infty(\mathrm{SupFormMfd}_{/\Sigma})
    \ar@<+4pt>@{<-}[r]
    \ar@<-4pt>[r]
    \ar@{=}[d]
    &
    \mathrm{Sh}_\infty(\mathrm{SupFormMfd}_{/\Im \Sigma})
    \ar@{=}[d]
    &
    \begin{turn}{-90} \mbox{ \hspace{-1.7cm} \begin{tabular}{c}  higher \\ topos \\ theory  \end{tabular}  } \end{turn}
    \\
    \infty \mathrm{Grpd}
    \ar@{^{(}->}[r]
    &
    \mathbf{H}_\Re
    \ar@{^{(}->}[r]
    &
    \mathbf{H}_{\rightsquigarrow}
    \ar@{^{(}->}[r]
    &
    \mathbf{H}
    \ar[r]^-{\Sigma^\ast}
    &
    \mathbf{H}_{/\Sigma}
    \ar@{<-}@<+4pt>[r]^{(\eta_\Sigma^\Im)^\ast}
    \ar@<-4pt>[r]_{(\eta_\Sigma^\Im)_\ast}
    &
    ´\mathbf{H}_{/\Im\Sigma}
    &
    \begin{turn}{-90}  \mbox{ \hspace{-1.3cm} \begin{tabular}{c} modal \\ homotopy \\ type theory  \end{tabular}  } \end{turn}
    \\
    \mbox{
      \small
      \begin{tabular}{c}
        discrete
        \\
        geometry
      \end{tabular}
    }
    &
    \mbox{
      \small
      \begin{tabular}{c}
        cohesive
        \\
        geometry
      \end{tabular}
    }
    &
    \mbox{
      \small
      \begin{tabular}{c}
        synthetic
        \\
        differential
        \\
        geometry
      \end{tabular}
    }
    &
    \mbox{
      \small
      \begin{tabular}{c}
        synthetic
        \\
        supergeometry
      \end{tabular}
    }
    &
    \mbox{
      \small
      \begin{tabular}{c}
        relative
        \\
        geometry
      \end{tabular}
    }
    &
    \mbox{
     \small
     \begin{tabular}{c}
       $\mathcal{D}$-geometry
     \end{tabular}
    }
  }
  \!\!\!\!\!\!\!\!\!\!\!\!\!\!\!\!\!\!
$
}

\newpage

The key now is that super formal smooth homotopy theory exhibits the following
abstract structure.

\begin{theorem}
\label{theprogression}
The homotopy theory $\mathbf{H} := \mathrm{Sh}_\infty(\mathrm{SupFormMfd})$
over the site of def. \ref{introductionsuperformalsmoothsite}
carries a system of idempotent $\infty$-(co-)monads as follows:

\begin{center}
\fbox{
$
  \xymatrix{
    &
    &
    \mathrm{id}
    \ar@{}[dd]|{\vee}
    \ar@{}[rr]|{\dashv}
    &&
    \mathrm{id}
    \ar@{}[dd]|{\vee}
    \\
    \\
    \ar@{}[d]|{ \mbox{ \small \begin{tabular}{c} synthetic \\ supergeometry \end{tabular} } }
    &
    &
    \mbox{\Large $\rightrightarrows$}
    \ar@{}[d]|{\bot}
    \ar@{}[rr]|{\dashv}
    &&
    \mbox{\Large $\rightsquigarrow$}
    \ar@{}[d]|{\mathrm{\bot}}
    \\
    &
    &
    \mbox{\Large $\rightsquigarrow$}
    \ar@{}[dd]|{\vee}
    \ar@{}[rr]|{\dashv}
    &&
    \mbox{\Large Rh}
    \ar@{}[ddll]|{\begin{turn}{-45} \tiny $\vee$ \end{turn} }
    \ar@{}[dd]|{\vee}
    \ar@{}[r]|-{\simeq}
    &
    \mathrm{loc}_{\mathbb{R}^{0\vert 1}}
    \\
    \\
    \ar@{}[d]|{ \mbox{\small \begin{tabular}{c} synthetic \\ differential \\ geometry \end{tabular}}  }
    &
    &
    \mbox{\Large $\Re$}
    \ar@{}[d]|{\bot}
    \ar@{}[rr]|{\dashv}
    &&
    \mbox{\Large $\Im$}
    \ar@{}[d]|{\bot}
    \ar@{}[r]|-{\simeq}
    &
    \mathrm{loc}_{\mathbb{D}}
    \\
    &
    \mathrm{loc}_{\mathbb{D}}
    \ar@{}[r]|-{\simeq}
    &
    \mbox{\Large $\Im$}
    \ar@{}[dd]|{\vee}
    \ar@{}[rr]|{\dashv}
    &&
    \mbox{\Large \it \&}
    \ar@{}[dd]|{\vee}
    \\
    \\
    \ar@{}[d]|{\mbox{\small \begin{tabular}{c} cohesive \\ geometry \end{tabular}}}
    &
    \mathrm{loc}_{\mathbb{R}^1}
    \ar@{}[r]|-{\simeq_{\Re}}
    &
    \mbox{\Large $\pi_\infty$} \ar@{}[d]|{\bot} \ar@{}[rr]|{\dashv}
    &&
    \mbox{\Large $\flat$} \ar@{}[d]|\bot
    \\
    &
    & \mbox{\Large $\flat$} \ar@{}[dd]|{\vee} \ar@{}[rr]|{\dashv}
    &&
    \mbox{\Large $\sharp$} \ar@{}[dd]|{\vee}
    \ar@{}[ddll]|{\begin{turn}{-45} \tiny $\vee$ \end{turn} }
    \\
    \\
    \mbox{\small \begin{tabular}{c} discrete \\ geometry\end{tabular}}
    &
    &
    \mbox{\Large $\emptyset$}
    \ar@{}[rr]|{\dashv}
    &&
    \mbox{\Large $\ast$}
  }
$
}
\end{center}

\medskip

Here
\begin{itemize}
  \item each $\bigcirc \dashv \Box$ is an adjunction of idempotent $\infty$-(co-)monads arising from an adjoint triple;
  \item $\bigcirc_1 < \bigcirc_2$ means that $(\bigcirc_1 X \simeq X) \Rightarrow (\bigcirc_2 X \simeq X) $.
\end{itemize}

\end{theorem}

\newpage

The existence of a progression of modal operators in theorem \ref{theprogression} is strong condition
on an $\infty$-topos $\mathbf{H}$. This suggests that much of the differential geometry
available in $\mathrm{Sh}_\infty(\mathrm{SupFormMfd})$ may be seen abstractly from reasoning in the internal
language of $\infty$-toposes with such a progression of modal operators added. This abstract homotopy theory
might be called \emph{super differential cohesive homotopy theory}, or just \emph{cohesive homotopy theory}
for short. In \cite{dcct} it is show that:
\begin{claim}
  In super differential cohesive homotopy theory, fundamental physics is synthetically\footnote{
    A synthetic axiomatization specifies intended properties of an object, in contrast to
    an analytic axiomatization which specifies how to build the intended object from more basic constituents.
    In synthetic formalization, a duck is what quacks like a duck, whereas in analytic formalization a
    duck has to be built out of its molecules.
    Euclid's plane geometry is synthetic, Descartes' plane geometry is analytic.
  }
  axiomatized
  \begin{enumerate}
   \item {\it naturally} -- the formalization is elegant and meaningful;
   \item {\it faithfully} -- the formalization captures the deep nontrivial phenomena;
   \item {\it usefully} -- the formalization yields proofs and constructions that are unfeasible otherwise.
  \end{enumerate}
\end{claim}

At the {International Congress of Mathematics} in Paris, 1900,
David Hilbert stated his list of 23 central open questions of mathematics
\cite{Hi1900}. Among them, the sixth problem
has possibly received the least attention
from mathematicians \cite{Corry04}, but:
``From all the problems in the list, the sixth is the only one that continually engaged
[Hilbert's] efforts over a very long period, at least between 1894 and 1932.''
\cite{Corry06}.
Hilbert stated the problem as follows:

\vspace{4pt}

\noindent{\bf Hilbert's problem 6.}
{\it To treat by means of {axioms}, those {physical sciences} in
	which mathematics plays an important part}
[...]
{\it   try first by a {small number of axioms}
    to include as large a class as possible of physical phenomena,
    and then by adjoining new axioms to arrive gradually at the more special theories. }
[...]
{\it take account not only of those theories coming near to reality, but also,
	as in geometry, all {logically possible theories} }.

\vspace{4pt}

Since then, various aspects of physics have been given a mathematical formulation.
The following table, necessarily incomplete, gives a broad idea of central
concepts in theoretical physics and the mathematics
that captures them.

\vspace{.5cm}

\begin{center}
\begin{tabular}{|l|ll|}
    \hline
	&&
	\\
    & {\bf physics} & {\bf mathematics}
	\\
	&&
    \\
	\hline\hline
    & {\it prequantum physics} & {\it differential geometry}
    \\
	\hline\hline
    18xx-19xx & {Lagrangian mechanics} & variational calculus
    \\
    18xx-19xx & {Hamiltonian mechanics} & {symplectic geometry}
	\\
	1910s & {gravity} &
	{Riemannian geometry}
	\\
	1950s & {gauge theory} &
	{Chern-Weil theory}
	\\
	2000s & {higher gauge theory} & {differential cohomology}
	\\
	\hline
	& &
	\\
	& &
	\\
	\hline \hline
    & {\it quantum physics} & {\it noncommutative algebra}
    \\
	\hline\hline
	1920s & {quantum mechanics} & {operator algebra}
	\\
	1960s & {local observables}  & {co-sheaf theory}
    \\
    1990s-2000s & {local field theory}
    & {$(\infty,n)$-category theory}
	\\
	\hline
\end{tabular}
\end{center}

\vspace{.5cm}

These are traditional solutions to aspects of Hilbert's sixth  problem.
Two points are noteworthy: on the one hand the items in the list are
crown jewels of mathematics; on the other hand their appearance is somewhat
unconnected and remains piecemeal.

Towards the end of the 20th century, William Lawvere, the founder
of categorical logic and of categorical algebra, aimed for a more
encompassing answer that rests the axiomatization of physics on a decent
unified foundation. He suggested to
\begin{enumerate}
  \item rest the foundations of mathematics itself in topos theory
     \cite{Lawvere65};
  \item build the foundations of physics \emph{synthetically} inside topos theory by
  \begin{enumerate}
    \item
  imposing properties on a topos which ensure that the
  objects  have the
  structure of \emph{differential geometric spaces}
  \cite{Lawvere98};
  \item formalizing classical mechanics on this basis by
        universal constructions
        \\
  (``Categorical dynamics'' \cite{Lawvere67},
   ``Toposes of laws of motion'' \cite{LawvereSynth}).
\end{enumerate}
\end{enumerate}
\noindent But
\begin{enumerate}
  \item modern mathematics naturally wants foundations not in
  topos theory, but in \emph{higher topos theory};
  \item modern physics needs to refine classical mechanics to
  \emph{quantum mechanics} and \emph{quantum field theory}
  at small length scales/high energy scales.
\end{enumerate}
Hence there is need for refining Lawvere's synthetic approach on Hilbert's sixth
problem from classical physics formalized in synthetic differential geometry
axiomatized in topos theory
to
high energy physics formalized in higher differential geometry
axiomatized in higher topos theory.
%Moreover, following Hilbert's problem description to consider
%``all logically possible theories'', we consider also string theoretic
%physics \cite{DeligneMorgan, PolchinskiBook}, in fact we find it arise from the axioms.

\medskip

Differential cohesive homotopy theory is close in spirit to the
classical approach of synthetic differential geometry \cite{KockBookA}. There one considers
1-toposes such as $\mathrm{Sh}(\mathrm{FormMfd})$, notices that they validate an axiom scheme called the
\emph{Kock-Lawvere axioms} (KL-axioms) and then develops differential geometry by reasoning in
the internal language of toposes validating the KL-axioms. The internal language of elementary toposes
is essentially intuitionistic type theory with coinductive types (TT+IT), except for two issues: interpreting
identity types is problematic and that the type universe tends to lack expected extensionality properties.

Both these issues have been resolved by passing to homotopy type theory \cite{Shulman16, HoTT} with univalent universes
and higher inductive types (HoTT+UV+HIT). This is argued to be the internal language of
(elementary) $\infty$-toposes and hence is just what is needed for synthetic higher differential geometry.

\begin{theorem}
$\,$
\begin{itemize}
\item HoTT has semantics in locally presentable locally Cartesian closed $\infty$-categories \cite{Shulman12};
\item HoTT+UV$_{\mathrm{strict}}$ has semantics in the $\infty$-topos  $\infty \mathrm{Grpd}$ \cite{KLV};
\item HoTT+UV$_\mathrm{strict}$ has semantics in a few infinite classes of $\infty$-presheaf $\infty$-toposes \cite{Shulman13,Shulman15};
\end{itemize}
\end{theorem}
\begin{remark}
$\,$
\begin{itemize}
\item HoTT+UV$_\mathrm{weak}$ is argued to have semantics in all $\infty$-toposes \cite{Shulman14};
\item HoTT+UV+Cohesion is developed in \cite{LicataShulman,Shulman15b,RSP}.
\end{itemize}
\end{remark}

\begin{center}
\begin{tabular}{|c|c|c|}
  \hline
  & \begin{tabular}{c}
    Kock-Lawvere
    \\
    {\bf synthetic }
    \\ {\bf diff. geometry }
  \end{tabular}
  &
  \begin{tabular}{c}
    $\,.$
    \\
    {\bf synthetic}
    \\
    {\bf higher diff. geometry}
  \end{tabular}
  \\
  \hline
  {\bf model} & topos & $\infty$-topos
  \\
  \hline
  \begin{tabular}{c}
    {\bf internal}
    \\
    {\bf language}
  \end{tabular}
  &
  \begin{tabular}{c}
    intuitionistic
    \\
    type theory
  \end{tabular}
  &
  \begin{tabular}{c}
    homotopy
    \\
    type theory
  \end{tabular}
  \\
  \hline
  {\bf axioms}
  &
  \begin{tabular}{c}
    KL-axiom scheme
    \\
    forcing internal
    \\
    infinitesimal
  \end{tabular}
  &
  \begin{tabular}{c}
    progression of
    \\
    modal operators
    \\
    forcing cohesion
  \end{tabular}
  \\
  \hline
\end{tabular}
\end{center}

In the following we will not reason fully formally in cohesive homotopy type theory, but instead
proceed in the familiar pseudocode formerly known as mathematics. But all constructions that follow
are manifestly such that they should lend themselves to full formalization in cohesive homotopy type theory.

%%%%%%%%%%%%%%%%%%%%%%%%%%%%%%%%%%%%%%%%%%%
\subsection{Abstract differential cohomology}
\label{introductionabstractdifferentialcohomology}
%%%%%%%%%%%%%%%%%%%%%%%%%%%%%%%%%%%%%%%%%%%

We now survey a list of abstract constuctions and theorems that follow formally for
every homotopy theory $\mathbf{H}$ which is equipped with the first stage of adjoint (co-)monads in
theorem \ref{theprogression}. These we call \emph{cohesive} homotopy theories.

\begin{definition}
  For $\mathbf{H}$ an $\infty$-topos, write
  $T \mathbf{H}$ for the $\infty$-category of parameterized spectrum objects in
  $\mathbf{H}$.
\end{definition}
\begin{proposition}[{\cite[section 35]{JoyalLogoi}}]
  $T \mathbf{H}$ is itself an $\infty$-topos. The spectra $\mathrm{Stab}(\mathbf{H}) \simeq T_\ast \mathbf{H}$
  are precisely the stable (linear) objects.
\end{proposition}
\begin{example}
  \label{twistedcohomologyfromtangenttoposhom}
  For $\mathbf{H} = \infty \mathrm{Grpd}$ then an object
  $E \in T_\ast \infty \mathrm{Grpd}$ is equivalently a spectrum,
  and for any $X \in \infty \mathrm{Grpd} \hookrightarrow T \infty \mathrm{Grpd}$ then
  $$
    E^\bullet(X) \simeq [X,E]
  $$
  is the $E$-cohomology spectrum of $X$.
  More generally, for $\tau \in T_X \infty \mathrm{Grpd}$ a bundle
  of spectra whose fibers are equivalent to $E$, then
  $$
    E^{\bullet + \tau}(X) \simeq [X,\tau]_X
  $$
  is the $\tau$-twisted $E$-cohomology spectrum of $X$ \cite{ABGHR13}.
\end{example}
\begin{example}
  For $S$ a site, let $\mathbf{H} := \mathrm{Sh}_\infty(S)$ be the hypercomplete
  $\infty$-topos over that site. The stable Dold-Kan correspondence turns
  a sheaf of chain complexes $A \in \mathrm{Ch}_\bullet(\mathrm{Sh}(S))$ into
  a spectrum object $H A \in T_\ast \mathbf{H} \hookrightarrow T \mathbf{H}$.
  Then
  $$
    H A^\bullet(X) \simeq [X, H A]
  $$
  is the sheaf hypercohomology of $X$ with coefficients in $A$ \cite{Brown}.
\end{example}

\begin{proposition}
  For $(\pi_\infty \dashv \flat \dashv \sharp) :  \mathbf{H} \to \mathbf{H}$ a cohesive $\infty$-topos
  then also its tangent $\infty$-topos is cohesive
  $$
    (T \pi_\infty \dashv T \flat \dashv T \sharp) :  T \mathbf{H} \to T\mathbf{H}
    \,.
  $$
\end{proposition}
\begin{definition}
  For $\bigcirc$ an idempotent $\infty$-(co-)monad on $T \mathbf{H}$, write $\overline{\bigcirc}$ for the homotopy (co-)fiber
  of its (co-unit).
\end{definition}
\begin{proposition}[\cite{BunkeNikolausVoelkl}]
\label{differentialhexagon}
For $(\pi_\infty \dashv \flat \dashv \sharp) :  \mathbf{H} \to \mathbf{H}$ a cohesive $\infty$-topos, then
for every $A \in T_\ast \mathbf{H}$ the canonical hexagon
$$
  \xymatrix{
    & \overline{\pi_\infty} A \ar[dr] \ar[rr]^{d} &&  \overline{\flat} A \ar[dr]
    \\
    \overline{\pi_\infty} \flat A  \ar[dr] \ar[ur]  && A \ar[ur] \ar[dr] && \pi_\infty \overline{\flat A}
    \\
    & \flat A \ar[ur] \ar[rr]_{\beta} && \pi_\infty A \ar[ur]
  }
$$
is homotopy exact, in that in addition to the diagonals being homotopy fiber sequences, also
\begin{enumerate}
  \item both squares are homotopy cartesian;
  \item both outer sequences are homotopy fiber sequences.
\end{enumerate}
\end{proposition}
\begin{proof}
  Use that homotopy pullback of stable objects is detected on homotopy fibers. Then use
  cohesion and idempotency to find that the squares are homotopy cartesian.
\end{proof}
\begin{example}
\label{introductionDelignecomplexfromhexagon}
Let $\mathbf{H}$ be as in theorem \ref{theprogression}.
Inside $\mathbf{H}$ the traditional Poincar{\'e} lemma is equivalent to the statement that there
is an equivalence
$$
  \flat \mathbb{R}
  \simeq
  \mathbf{\Omega}^\bullet
  \;\;\;\;
  \in T_\ast \mathbf{H}
  \,.
$$
This induces for each $p \in \mathbb{N}$ an instance of the  exact hexagon of
prop. \ref{differentialhexagon}:
$$
  \xymatrix{
    &
    \mathbf{\Omega}^{\bullet \leq p+1} \ar[rr]^{d_{\mathrm{dR}}}
    \ar[dr]
    &&
    \mathbf{\Omega}^{p+2}_{\mathrm{cl}}
    \ar[dr]^{\mbox{ \tiny \begin{tabular}{c} de Rham \\ theorem \end{tabular} }}
    \\
    \flat \mathbf{B}^{p+1} \mathbb{R}
    \ar[ur]^{\mbox{ \tiny \begin{tabular}{c} Poincar{\'e} \\ lemma\end{tabular} } }
    \ar[dr]
    &&
    \mathbf{B}^{p+1}(\mathbb{R}/_{\!\hbar}\mathbb{Z})
    \ar[ur]|{\mathrm{curv}}
    \ar[dr]
    &&
    \flat \mathbf{B}^{p+1}\mathbb{R}
    \\
    & \flat \mathbf{B}^{p+1} (\mathbb{R}/_{\!\hbar}\mathbb{Z})
    \ar[ur]
    \ar[rr]_{\mathrm{Bockstein}}
    &&
    \mathbf{B}^{p+2}\mathbb{Z}
    \ar[ur]_{2\pi \hbar}
  }
  \,.
$$
The object appearing the middle is the Deligne complex
$$
  \mathbf{B}^{p+1}(\mathbb{R}/_{\!\hbar} \mathbb{Z})_{\mathrm{conn}}
  \simeq
  H [\mathbb{Z} \stackrel{2 \pi \hbar}{\hookrightarrow} \mathbf{\Omega}^0 \stackrel{d_{\mathrm{dR}}}{\to} \mathbf{\Omega}^1 \stackrel{d_{\mathrm{dR}}}{\to} \cdots \to \mathbf{\Omega}^{p+1}]
  \,.
$$
For $X \in \mathbf{H} \stackrel{0}{\hookrightarrow} T_\ast \mathbf{H}$ then
$$
  \hat H^{p+2}(X,\mathbb{Z}) \simeq \pi_0 [X, \mathbf{B}^{p+1}(\mathbb{R}/_{\!\hbar} \mathbb{Z})]
$$
is known equivalently as
\begin{enumerate}
  \item the ordinary differential cohomology of $X$ in degree $(p+2)$;
  \item the Deligne cohomology of $X$ in degree $(p+2)$;
  \item the equivalence classes of $p$-gerbe connections for band $(\mathbb{R}_{/\!\hbar} \mathbb{Z})$;
  \item the equivalence classes of $\mathbf{B}^p(\mathbb{R}/_{\!\hbar} \mathbb{Z}$-principal bundles with connection.
\end{enumerate}
\end{example}
\begin{remark}
  In \cite{SimonsSullivan} it was observed that the natural hexagon that ordinary differential cohomology sits in already characterizes
  it. The authors suggested that this may be true for generalized differential cohomology theories.
  In view of this prop. \ref{differentialhexagon} may be regarded as the lift of the Brown
  representability theorem from generalized cohomology theories to generalized differential cohomology theories.
\end{remark}
Combining this observation with example \ref{twistedcohomologyfromtangenttoposhom}, we find that
as we vary the slices of the cohesive $\infty$-topos, it knows also about twisted differential cohomology:
\begin{center}
\begin{tabular}{|c|c|c|}
  \hline
    &
    \begin{tabular}{c}
      {\bf cohomology}
    \end{tabular}
    &
    \begin{tabular}{c}
      {\bf differential cohomology}
    \end{tabular}
    \\
    \hline
    \hline
    {\bf plain}
    &
    $A \in T_\ast \infty \mathrm{Grpd}$
    &
    $\mathbf{A} \in T_\ast \mathbf{H}$
    \\
    \hline
    {\bf twisted}
    &
    $\tau \in T_{\mathrm{Pic}(A)} \infty \mathrm{Grpd}$
    &
    $\mathbf{\tau} \in T_{\mathrm{Pic}(\mathbf{A})} \mathbf{H}$
    \\
    \hline
  \end{tabular}
\end{center}
Hence the hexagon in prop. \ref{differentialhexagon} generally has the following interpretation:
$$
  \hspace{-1cm}
  \xymatrix{
    & \fbox{\begin{tabular}{c} connections \\ on trivial bundles \end{tabular}}
    \ar[rr]|-{\mbox{de Rham differential}}
    \ar[dr]|-{\mbox{regard as}}
    && \fbox{\begin{tabular}{c} curvature \\ forms \end{tabular}}
    \ar[dr]|{\mbox{de Rham theorem}}
    \\
    \fbox{\begin{tabular}{c} closed \\ differential forms \end{tabular}}
    \ar[ur]|-{\mbox{regard as}}
    \ar[dr]|-{\mbox{regard as}}
    && \fbox{\begin{tabular}{c} connections on \\ geometric bundles \end{tabular}}
    \ar[ur]|{\mbox{curvature}}
    \ar[dr]|<<<<<<<<{\mbox{\begin{tabular}{l}topol. class\end{tabular}}}
    && \fbox{\begin{tabular}{c} rationalized \\ bundle \end{tabular}}
    \\
    & \fbox{\begin{tabular}{c} flat \\ connections \end{tabular}}
    \ar[ur]|-{\mbox{regard as}}
    \ar[rr]|-{\mbox{\begin{tabular}{l}comparison map\end{tabular}}}
    &&
    \fbox{\begin{tabular}{c} shape of \\ bundles \end{tabular}}
    \ar[ur]|{\mbox{Chern character}}
  }
$$
\begin{definition}[{\cite{BunkeNikolausVoelkl}}]
  \label{introductionabstractintegration}
  Using that $\pi_\infty \simeq \mathrm{loc}_{\mathbb{R}^1}$
  the universal property of $\pi_\infty$ induces for each linear cohesive object $A \in T_\ast \mathbf{H}$
  a canonical morphism of the form
  $$
    \int_0^1 \;:\; [\mathbb{R}^1, \overline{\flat} A] \longrightarrow \overline{\pi_\infty} A
    \,.
  $$
\end{definition}
\begin{proposition}[{\cite{BunkeNikolausVoelkl}}]
  In the situation of example \ref{introductionDelignecomplexfromhexagon},
  the abstracty defined map $\int_0^1$ from def. \ref{introductionabstractintegration}
  is equivalent to traditional fiber integration of differential forms.
\end{proposition}
\begin{proposition}[fundamental theorem of calculus {\cite{BunkeNikolausVoelkl}}]
  \label{introductionfundamentaltheoremofcalculus}
  For every linear object $A\in T_\ast \mathbf{H}$
   we have
  $$
    \int_0^1 \circ d \;\; \simeq \;\; (-)|_1 - (-)|_0
    \,,
  $$
  where $d$ is the top morphism in prop. \ref{differentialhexagon} and
  where $\int_0^1$ is the morphism from def. \ref{introductionabstractintegration}.
\end{proposition}
\begin{remark}
  The statement of prop. \ref{introductionfundamentaltheoremofcalculus}
  was long imposed as an extra axiom on differential cohomology theories
  (see \cite{BunkeDifferentialCohomology}).
\end{remark}
In summary, this shows that $(\pi_\infty \dashv \flat)$ synthetically axiomatizes the
existence of differential cocycles. The remaining monad $\sharp$ turns out to give
the moduli spaces of such cocycles:

\begin{definition}
  \label{introductionsharpk}
  For $n \in \mathbb{N}$, write $\sharp_n$ for the $n$-image factorization
  of the unit $\mathrm{id}\longrightarrow \sharp$.
\end{definition}
\begin{proposition}
  For $\mathbf{H}$ as in theorem \ref{theprogression},
  the diffeological spaces
  are equivalently the reduced 0-truncated objects which are in addition $\sharp_1$-model
  $$
    \mathrm{DiffeologicalSp}
    \simeq
    \mathbf{H}_{\tau_0, \Re, \sharp_1}
    \hookrightarrow
    \mathbf{H}
    \,.
  $$
\end{proposition}
\begin{example}
  For $\mathbf{H}$ as in theorem \ref{theprogression},
  and $\mathbf{\Omega}^{p+1}\in \mathbf{H}$ the sheaf of differential forms
  and for $\Sigma \in \mathbf{H}$ any smooth manifold, then
  the mapping space $[\Sigma,\mathbf{\Omega}^{p+1}]$ is not the diffeological
  space of differential $(p+1)$-forms on $\Sigma$, but $\sharp_1[\Sigma,\mathbf{\Omega}^{p+1}]$ is.
\end{example}
For object which are not 0-truncated, concretification depends on a choice of co-filtration:
\begin{definition}
  \label{introductionconcretification}
  For $X \in \mathbf{H}$ equipped with a co-filtration $F^\bullet X$, we say
  that its \emph{concretification} is the iterated homotopy fiber product
  $$
    \mathrm{Conc}(F^\bullet X)
    :=
    \sharp_1 F^0 X \underset{\sharp_1 F^1 X}{\times} \sharp_2 F^1 X \underset{\sharp_2 F^2 X}{\times} \cdots
    \,,
  $$
  with $\sharp_k$ from def. \ref{introductionsharpk}  or rather, is the canonical morphism
  $$
    \mathrm{conc} : X \longrightarrow \mathrm{Conc}(F^\bullet X)
    \,.
  $$
\end{definition}
\begin{proposition}
  For $\mathbf{B}^{p+1}(\mathbb{R}/_{\!\hbar}\mathbb{Z})\in \mathbf{H}$ from prop. \ref{introductionDelignecomplexfromhexagon}
  and equipped with its canonical co-filtration induced from the Hodge filtration on $\mathbf{\Omega}^\bullet$,
  then for $\Sigma \in \mathbf{H}$ a smooth manifold, the concretification
  $$
    \mathbf{B}^{p}(\mathbb{R}/_{\!\hbar}\mathbb{Z})\mathbf{Conn}(\Sigma)
    :=
    \mathrm{Conc}([\Sigma,\mathbf{B}^{p+1}(\mathbb{R}/_{\!\hbar}\mathbb{Z})])
  $$
  is given by the diffeological $(p+1)$-groupoid of Deligne $(p+2)$-cocycles on $\Sigma$.
\end{proposition}
Using this there is an axiomatization of the higher groups of symmetries of $p$-gerbes as they appeared in
section \ref{symmetriescurrentsinintroduction}:
\begin{proposition}
  \label{introductionhigherKSextension}
  For any $X \in \mathbf{H}$ and give
  $\nabla : X \longrightarrow \mathbf{B}^{p+1}(\mathbb{R}/_{\!\hbar}\mathbb{Z})_{\mathrm{conn}}$,
  then there is a homotopy fiber sequence of the form
  $$
    \raisebox{20pt}{
    \xymatrix{
      \mathbf{B}^{p}(\mathbb{R}/_{\!\hbar}\mathbb{Z})\mathbf{FlatConn}(X)
      \ar[r]
      & \mathbf{Stab}_{\mathbf{Aut}(X)}(\mathrm{conc}(\nabla))
      \ar[d]
      \\
      & \mathbf{HamAut}(X,\nabla)
      \ar[r]^-{\mathbf{KS}_\nabla}
      &
      \mathbf{B}(\mathbf{B}^{p}(\mathbb{R}/_{\!\hbar}\mathbb{Z})\mathbf{FlatConn}(\Sigma))
    }
    }
    \,,
  $$
  where $\mathbf{Stab}(...)$ denotes the stabilizer $\infty$-group of $\nabla$ in
  $\mathbf{B}^p(\mathbb{R}/_{\!\hbar}\mathbb{Z})\mathbf{Conn}(\Sigma)$ under the canonical
  $\infty$-action of the automorphism $\infty$-group of $\Sigma$, and where $\mathbf{HamAut}(X,\nabla)$
  is the 1-image of the canonical map from there to $\mathbf{Aut}(X)$.
\end{proposition}
\begin{example}
In the case that $\nabla$ is a $U(1)$-principal connection, $\mathbf{KS}_\nabla$
is the class of the traditioonal Kostant-Souriau quantomorphism extension.
\end{example}
\begin{proposition}
  \label{introductionglobalizationobstruction}
  With $(X,\nabla)$ as in prop. \ref{introductionhigherKSextension}, given an $X$ fiber bundle $E \to \Sigma$
  then definite globalizations of $\nabla$ over $E$ are equivalent to lifts of the structure group
  of $E$ through $\mathrm{Stab}_{\mathbf{Aut}(X)}(\mathrm{conc}(\nabla)) \to \mathbf{Aut}(X)$. In particular
  $\mathbf{KS}_\nabla(E)$ is the obstruction class to the existence of such a globalization.
\end{proposition}

%%%%%%%%%%%%%%%%%%%%%%%%%%%%%%%%%%%%%%%%%%%%%%%%%%%%%%
\subsection{Abstract differential geometry}
\label{introductionabstractdifferentialgeometry}
%%%%%%%%%%%%%%%%%%%%%%%%%%%%%%%%%%%%%%%%%%%%%%%%%%%%%%

We now survey a list of abstract constructions and theorems that follow formally for
every homotopy theory $\mathbf{H}$ which is equipped with the first and the second stage of adjoint (co-)monads in
theorem \ref{theprogression}. These we call \emph{differential cohesive} homotopy theories.

\begin{proposition}
  In the situation of theorem \ref{theprogression},
  for $\Sigma \in \mathrm{SmoothMfd} \hookrightarrow \mathbf{H}$
  then there is a pullback diagram
  $$
    \xymatrix{
      T^\infty \Sigma \ar@{}[dr]|{\mathrm{(pb)}} \ar[d]_{p_1} \ar[r]^{p_2} & \Sigma \ar@{->>}[d]^{\eta^\Im_\Sigma}
      \\
      \Sigma \ar[r]_{\eta^\Im_\Sigma} & \Im \Sigma
    }
  $$
  where $T^\infty \Sigma$ is the formal neighbourhood of the diagonal of $\Sigma$,
  and $\Im \Sigma$ is the coequalizer of the two projections
  $$
    \xymatrix{
      T^\infty \Sigma
      \ar@<-4pt>[rr]_{p_2}
      \ar@<+4pt>[rr]^{p_1}
      &&
      \Sigma
      \ar[rr]^{\eta^\Im_\Sigma}
      &&
      \Im \Sigma
    }
    \,.
  $$
\end{proposition}
\begin{remark}
  Hence $\Im \Sigma$ is what elsewhere is called the \emph{de Rham stack} of $\Sigma$, also
  denoted $X_{\mathrm{dR}}$.  Its sheaf cohomology is crystalline cohomology.
\end{remark}
\begin{definition}
  \label{introductionVmanifold}
  For $\mathbf{H}$ a differential cohesive $\infty$-topos, say that a morphism
  $f : X \to Y$ is \emph{formally {\'e}tale} if the naturality square of its
  $\Im$-unit is a homotopy pullback
  $$
    \raisebox{20pt}{
    \xymatrix{
      X \ar[d]_f|{\mathrm{et}} \ar@{}[dr]|{\mathrm{(pb)}} \ar[r]^{\Im^\eta_\Sigma} & \Im X \ar[d]^{\Im f}
      \\
      Y \ar[r]_{\Im^\eta_\Sigma} & \Im Y
    }
    }\,.
  $$
  For $V \in \mathrm{Grp}(\mathbf{H})$ a group object, say that a \emph{$V$-manifold}
  is an object $X \in \mathbf{H}$ equipped with a \emph{$V$-atlas}, namely with a correspondence
  of the form
  $$
    \xymatrix{
       & U
       \ar[dl]|{\mathrm{et}}
       \ar@{->>}[dr]|{\mathrm{et}}
       \\
      V && X
    }
  $$
  such that both maps are formally {\'e}tale and such that the right map is in addition a
  1-epimorphism.
\end{definition}
\begin{proposition}
  \label{introductionframebundle}
  For $\mathbf{H}$ a differential cohesive $\infty$-topos and for $V \in \mathrm{Grp}(\mathbf{H})$
  any group object, then its formal disk bundle $p_1 : T^\infty V  \longrightarrow V$
  is canonically trivialized by left translation. Moreover, for $X$ any $V$-manifold,
  def. \ref{introductionVmanifold}, then the formal disk bundle of $X$ is associated to a
  uniquely defined $\mathrm{GL}^\infty(V)$-principal bundle
  $$
    \mathrm{Fr}(X) \longrightarrow X
    \,,
  $$
  its \emph{frame bundle},
  where
  $$
    \mathrm{GL}^\infty(V)
    :=
    \mathbf{Aut}(\mathbb{D}^V_e)
  $$
  is the automorphism group of the formal disk around the neutral element in $V$.
\end{proposition}
Proposition \ref{introductionframebundle} allows to abstractly speak of $G$-structure and torsion-free $G$-structure
on $V$-manifolds, in any differential cohesive $\infty$-topos, hence to formalize Cartan geometry,
which subsumes (pseudo-)Riemannian geometry, complex geometry, symplectic geometry, conformal geometry, etc.
Moreover, $G$-structures naturally arise as follows.
\begin{proposition}
  Given a differential cocycle $\nabla^V : V \longrightarrow \mathbf{B}^{p+1}(\mathbb{R}/_{\!\hbar}\mathbb{Z})$
  and a $V$-manifold $X$,
  then there is an $\infty$-functor from definite globalizations of $\nabla^V$ over $X$ to
  $\mathrm{Stab}_{\mathrm{GL}(V)}(\mathrm{conc}\nabla^{\mathbb{D}^V,E})$-structures on the frame bundle of $X$,
  where $\nabla^{\mathbb{D}^V_e}$ is the restriction of $\nabla$ to the infinitesimal neighbourhood of $e$ in $V$.
  In particular the class $\mathbf{KS}_{\nabla^{\mathbb{D}^V_e}}(\mathrm{Fr}(X))$ from prop. \ref{introductionglobalizationobstruction}
  is an obstruction to the existence of such a globalization.
\end{proposition}

%%%%%%%%%%%%%%%%%%%%%%%%%%%%%%%%%%%%%%%%%%%%%%%%%%
\subsection{Abstract PDE theory}
\label{introductionabstractPDEtheory}
%%%%%%%%%%%%%%%%%%%%%%%%%%%%%%%%%%%%%%%%%%%%%%%%%

We survey more abstract constuctions and theorems that follow formally for
every differential cohesive homotopy theory, $\mathbf{H}$ i.e. one equipped with the first and second stage of adjoint (co-)monads in
theorem \ref{theprogression}.

\begin{definition}
  \label{introductionjetcomonad}
  For any $\Sigma \in \mathbf{H}$, write
  $$
    \left(
      \mbox{\Large $T^\infty_\Sigma$}
      \dashv
      \mbox{\Large $J^\infty_\Sigma$}
    \right)
    :=
    \left(
      (\eta_\Sigma)^\ast \circ (\eta_\Sigma)_!
      \dashv
      (\eta_\Sigma)^\ast \circ (\eta_\Sigma)_\ast
    \right)
    :
    \mathbf{H}_{/\Sigma} \longrightarrow \mathbf{H}_{/\Sigma}
  $$
  for the base change (co-)monad along the unit of the $\Im$-monad.
\end{definition}

\begin{proposition}
  \label{introductionabstractjetsgivetraditionaljets}
  For $\mathbf{H}$ from \ref{theprogression},
  then for $E \in \mathrm{FormMfd}_{/\Sigma} \hookrightarrow \mathbf{H}_{/\Sigma}$
  the (co-)monads in def. \ref{introductionjetcomonad} come out as follows:
  \begin{enumerate}
    \item $T^\infty_\Sigma \Sigma$ is the formal disk bundle of $\Sigma$ \cite[above prop. 2.2]{Kock80};
    \item $J^\infty_\Sigma E$ is the jet bundle of $E$ \cite[remark 7.3.1]{KockBook}.
  \end{enumerate}
\end{proposition}
\begin{proposition}[{\cite{Marvan86},\cite[section 1.1]{Marvan93}}]
  In the situation of prop. \ref{introductionabstractjetsgivetraditionaljets},
  the Eilenberg-Moore category of jet coalgebras over $\Sigma$ is equivalent to
  Vinogradov's category of partial differential equations with free variables in $\Sigma$:
  $$
    \mathrm{EM}(J^\infty_\Sigma) \simeq  \mathrm{PDE}_\Sigma
    \,.
  $$
  In particular the co-Kleisli category of the jet comonad is that of bundles over $\Sigma$
  with differential operators between them as morphisms.
  $$
    \mathrm{Kl}(J^\infty_\Sigma) \simeq \mathrm{DiffOp}_\Sigma
    \,.
  $$
\end{proposition}
Since prop. \ref{introductionjetcomonad} gives the jet comonad by base change,
the $\infty$-Beck monadicity theorem gives in generality that
\begin{proposition}
  There is an equivalence of $\infty$-categories
  $$
    \mathrm{EM}(J^\infty_\Sigma, \mathbf{H})
    \simeq
    \mathbf{H}_{/\Im\Sigma}
    \,.
  $$
\end{proposition}

Write then
$$
  \xymatrix{
    \mathbf{H}
    \ar[rr]^{\Sigma^\ast}
    \ar@/^2pc/[rrrr]^{(-)_\Sigma}
    &&
    \mathbf{H}_{/\Sigma}
    \ar[rr]^{(\eta^\Im_\Sigma)_\ast}
    &&
    \mathbf{H}_{/\Im \Sigma}
  }
$$
for the canonical map that regards objects of the differential cohesive $\infty$-topos
as co-free homotopy partial differential equations:

In the situation of example \ref{introductionDelignecomplexfromhexagon},
consider the universal decomposition of the differential forms
$(\mathbf{\Omega}^{\bullet \leq p+1})_\Sigma$ regarded over $\Im \Sigma$
this way into horizontal and vertical forms.
$$
  \xymatrix{
    0 \ar[d] \ar[rr] & & \mathbf{\Omega}^{p+1}_V
    \ar[d]
    \\
    \Sigma \ar[r]^{\phi} & E \ar[r]^-{L} & (\mathbf{\Omega}^{\bullet \leq p+1})_\Sigma
    \ar[r]
    &\mathbf{\Omega}^{\bullet \leq p+1}_H
  }
$$
\begin{proposition}
\label{introductioninducedelgerbes}
This induces a horizontal projection of the exact hexagon from example \ref{introductionDelignecomplexfromhexagon}:
$$
  \!\!\!\!\!\!\!\!\!\!\!\!\!\!\!\!\!\!\!
  \!\!\!\!\!\!\!\!\!\!\!\!\!\!\!\!\!\!\!\!
  \xymatrix@C=16pt@R=20pt{
    &&
    (\mathbf{\Omega}^{\bullet \leq p+1})_\Sigma
    \ar[rrrr]^{d_{\mathrm{dR}}}
    \ar[drr]
    \ar[ddrr]|>>>>>>>>{\phantom{AA} \atop \phantom{AA}}
    &&&&
    (\mathbf{\Omega}^{p+2}_{\mathrm{cl}})_\Sigma
    \ar[drr]
    \ar[ddrr]|>>>>>>>>>>>>>>>>>>{\phantom{AA} \atop \phantom{AA}}|>>>>>>>{\phantom{AA} \atop \phantom{AA}}
    \\
    &&&& \mathbf{\Omega}_H^{\bullet \leq p+1}
    \ar[rrrr]^{\delta_V}
    \ar@{->>}[ddrr]
    &&&&
    \mathbf{\Omega}^{p+1,1}_S
    \ar[ddrr]
    \\
    (\flat \mathbf{B}^{p+1}\mathbb{R})_\Sigma
    \ar[uurr]
    \ar[ddrr]
    \ar@{=}[drr]
    &&&&
    (\mathbf{B}^{p+1}(\mathbb{R}/_{\!\hbar}\mathbb{Z})_{\mathrm{conn}})_\Sigma
    \ar[uurr]|<<<<<<<{\phantom{AA} \atop \phantom{AA}}|<<<<<<<<<<<<<<<<<<<<{\phantom{AA}\atop \phantom{AA}}|>>>>>>>>>{\mathrm{curv}}
    \ar[ddrr]|>>>>>>>>{\phantom{AA} \atop \phantom{AA}}
    \ar[drr]|{H}
    &&&&
    (\flat \mathbf{B}^{p+2}\mathbb{R})_\Sigma
    \ar@{=}[drr]
    \\
    &&(\flat \mathbf{B}^{p+1}\mathbb{R})_\Sigma
    \ar[uurr]
    \ar[ddrr]
    &&
    &&
    \mathbf{B}^{p+1}_H(\mathbb{R}/_{\!\hbar}\mathbb{Z})_{\mathrm{conn}}
    \ar[uurr]|{\mathrm{curv}}
    \ar[ddrr]
    &&&&
    (\flat \mathbf{B}^{p+1}\mathbb{R})_\Sigma
    \\
    && (\flat \mathbf{B}^{p+1}(\mathbb{R}/_{\!\hbar} \mathbb{Z}))_\Sigma
    \ar[rrrr]|<<<<<<<<<<{\phantom{AAAA}}_{\beta_\Sigma}|>>>>>>>>>>>>>>>{\phantom{AAAA}}
    \ar[uurr]|<<<<<<<<{\phantom{AA} \atop \phantom{AA}}
    \ar@{=}[drr]
    &&&&
    (\mathbf{B}^{p+2}\mathbb{Z})_\Sigma
    \ar[uurr]|<<<<<<<{\phantom{AA} \atop\phantom{AA}}
    \ar@{=}[drr]
    \\
    &&&& (\flat \mathbf{B}^{p+1}(\mathbb{R}/_{\!\hbar} \mathbb{Z}))_\Sigma
    \ar[rrrr]_{\beta_\Sigma}
    \ar[uurr]
    &&&&
    (\mathbf{B}^{p+2}\mathbb{Z})_\Sigma
    \ar[uurr]
  }
$$
\end{proposition}
This is the abstract characterization of the Euler-Lagrange $p$-gerbes of section \ref{globalactionfunctionalinintroduction}.
Hence the front hexagon in prop. \ref{introductioninducedelgerbes} now has the following interpretation.
$$
  
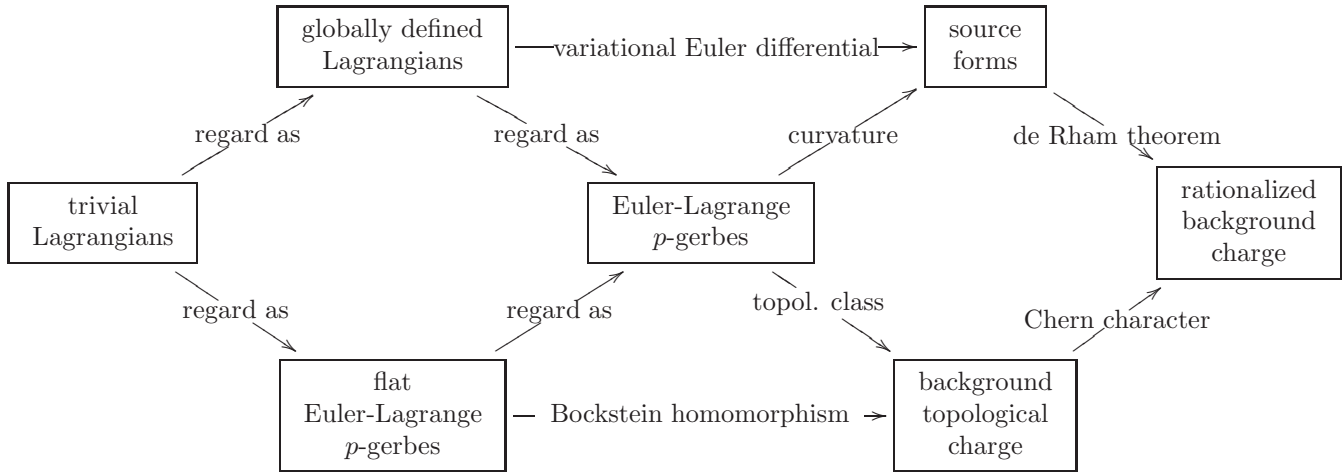
\hspace{-.6cm}
  \xymatrix{
    & \fbox{\begin{tabular}{c} globally defined \\ Lagrangians \end{tabular}}
    \ar[rr]|-{\mbox{variational Euler differential}}
    \ar[dr]|-{\mbox{regard as}}
    && \fbox{\begin{tabular}{c} source \\ forms \end{tabular}}
    \ar[dr]|{\mbox{de Rham theorem}}
    \\
    \fbox{\begin{tabular}{c} trivial \\ Lagrangians \end{tabular}}
    \ar[ur]|-{\mbox{regard as}}
    \ar[dr]|-{\mbox{regard as}}
    && \fbox{\begin{tabular}{c} Euler-Lagrange \\ $p$-gerbes \end{tabular}}
    \ar[ur]|{\mbox{curvature}}
    \ar[dr]|<<<<<<<<{\mbox{\begin{tabular}{l}topol. class\end{tabular}}}
    && \fbox{\begin{tabular}{c} rationalized \\ background \\ charge \end{tabular}}
    \\
    & \fbox{\begin{tabular}{c} flat \\ Euler-Lagrange \\ $p$-gerbes \end{tabular}}
    \ar[ur]|-{\mbox{regard as}}
    \ar[rr]|-{\mbox{\begin{tabular}{l}Bockstein homomorphism\end{tabular}}}
    &&
    \fbox{\begin{tabular}{c} background \\ topological \\ charge \end{tabular}}
    \ar[ur]|{\mbox{Chern character}}
  }
$$
Similarly there is a further filtration of horizontal projections which induces also the
Lepage $p$-gerbes of section \ref{elgerbesinintroduction}.

Hence the abstract differential cohomology in cohesive homotopy theory combined with the
abstract manifold theory and abstract PDE theory of differential cohesive homotopy theory
provides just the right formal language for abstractly speaking about the prequantum field theory
surveyed in section \ref{PrequantumLocalFieldTheoryInMotivation}.

%%%%%%%%%%%%%%%%%%%%%%%%%%%%%%%%%%%%%%%%%%%%%%%%%%%

\end{document}